\documentclass[10pt,thm-restate]{article}

\usepackage[pdftex]{graphicx}
\usepackage{pgfplots}
\pgfplotsset{compat=1.14}
\usepackage{tikz-3dplot}
\usepackage{xspace}
\usepackage{tikz}
\usepackage{multirow}
\usepackage{marvosym}
\usepackage{enumitem}
\usepackage{fullpage}
\usepackage{hyperref}
\usepackage{amsthm}
\usepackage{amssymb}
\usepackage{mathtools}
\usepackage{bm}
\usepackage{amsfonts}
\usepackage{booktabs}
\usepackage{subcaption}
\usepackage{relsize}
\usepackage[T1]{fontenc}

\usepackage{amsmath}
\usepackage{hyperref}
\usepackage{pgfplots}
\pgfplotsset{compat=1.16}
\usepackage{tikz-3dplot}
\usepackage{xspace}
\usepackage{tikz}
\pgfdeclarelayer{background}
\pgfsetlayers{background,main}

\usepackage{float}
\usepackage{multirow}
\usepackage{marvosym}
\usepackage{enumitem}
\usepackage{booktabs}
\usepackage{subcaption}
\usepackage{tcolorbox}
\usepackage{mdframed}
\usepackage{minibox}
\usepackage{xcolor,colortbl}
\usepackage{makecell}
\usepackage{arydshln}

\usepackage{thmtools, thm-restate}

\theoremstyle{plain}                  
\newtheorem{theorem}{Theorem}
\newtheorem{lemma}[theorem]{Lemma}
\newtheorem{proposition}[theorem]{Proposition}

\newtheorem{definition}[theorem]{Definition}
\newtheorem{example}[theorem]{Example}

\makeatletter
\newcommand{\customsize}{%
  \@setfontsize\customfontsize{10.5pt}{12.5pt}
  \selectfont
}

\newcommand{\banz}{Banzhaf}
\newcommand{\shap}{Shapley}

\newcommand{\pptwodnf}{\text{PP2DNF}\xspace}
\newcommand{\vars}{vars}
\newcommand{\defeq}{\stackrel{\text{def}}{=}}
\newcommand{\at}{at}
\newcommand{\Dom}{\mathsf{Dom}}

\newcommand{\FOREACH}{\textbf{foreach}\xspace}%
\newcommand{\DO}{\textbf{do}\xspace}%
\newcommand{\STAB}{\makebox[1ex][r]{}}%
\newcommand{\TAB}{\makebox[2.5ex][r]{}}%
\newcommand{\IF}{\textbf{if}\xspace}%
\newcommand{\RETURN}{\textbf{return}\xspace}%
\newcommand{\LET}{\textbf{let}\xspace}%
\newcommand{\ELSE}{\textbf{else}\xspace}%
\newcommand{\SWITCH}{\textbf{switch}\xspace}%
\newcommand{\CASE}{\textbf{case}\xspace}%
\newcommand{\DEFAULT}{\textbf{default}\xspace}%
\newcommand{\compbanz}{\textsc{RankBan}\xspace}
\newcommand{\countbis}{\#\textsc{BIS}\xspace}

\newcommand{\countnsat}{\#\textsc{NSat}\xspace}
\newcommand{\adaban}{\textsc{AdaBan}\xspace}

\newcommand{\fptas}{\textsc{FPTAS}\xspace}

\newcommand{\fpsharpp}{\textsc{FP}^{\#\textsc{P}}\xspace}

\newcommand{\revision}[1]{{#1}}

\newcommand{\nop}[1]{}

\newcommand{\scream}[1]{\textcolor{blue}{$\ll$\textsf{#1 --TBD}$\gg$}}

\date{}
\title{Banzhaf Values for Facts in Query Answering}

\author{
\begin{tabular}{ccc}
\hspace{0.5cm} Omer Abramovich\hspace{0.5cm}   &\hspace{0.5cm} Daniel Deutch\hspace{0.5cm} &\hspace{0.5cm}  Nave Frost\hspace{0.5cm} \\
\hspace{0.5cm} \customsize{Tel Aviv University}\hspace{0.5cm}   &\hspace{0.5cm}   \customsize{Tel Aviv University}\hspace{0.5cm}  &\hspace{0.5cm}  \customsize{eBay Research}\hspace{0.5cm}  \\
\hspace{0.5cm} omera1@mail.tau.ac.il\hspace{0.5cm}   &\hspace{0.5cm}   \customsize{danielde@post.tau.ac.il}\hspace{0.5cm}  &\hspace{0.5cm}   \customsize{nafrost@ebay.com} 
\hspace{0.5cm} 
\end{tabular}\\\\
\begin{tabular}{cc}
\hspace{0.5cm}Ahmet Kara\hspace{0.5cm} &\hspace{0.5cm} Dan Olteanu\hspace{0.5cm}\\
\hspace{0.5cm}\customsize{University of Zurich}\hspace{0.5cm} &  \hspace{0.5cm}\customsize{University of Zurich}\hspace{0.5cm}  \\
\hspace{0.5cm}\customsize{kara@ifi.uzh.ch}\hspace{0.5cm} &  \hspace{0.5cm}\customsize{olteanu@ifi.uzh.ch}\hspace{0.5cm}  
\end{tabular}
}

\begin{document}
\maketitle

\begin{abstract}
 Quantifying the contribution of database facts to query answers has been studied as means of explanation. The Banzhaf value, originally developed in Game Theory, is a natural measure of fact contribution, yet its efficient computation for select-project-join-union queries is challenging.
In this paper, we introduce three algorithms to compute the Banzhaf value of database facts: an exact algorithm, an anytime deterministic approximation algorithm with relative error guarantees, and an algorithm for ranking and top-$k$. They have three  key building blocks: compilation of query lineage into an equivalent function that allows efficient Banzhaf value computation; dynamic programming computation of the Banzhaf values of variables in a Boolean function using the Banzhaf values for constituent functions; and a mechanism to compute efficiently lower and upper bounds on Banzhaf values for any positive DNF function.

We complement the algorithms with a dichotomy for the Banzhaf-based ranking problem: given two facts, deciding whether the Banzhaf value of one is greater than of the other is tractable for hierarchical queries and intractable for non-hierarchical queries.

We show experimentally that our algorithms significantly outperform exact and approximate algorithms from prior work, most times up to two orders of magnitude. Our algorithms can also cover challenging problem instances that are beyond reach for prior work.
\end{abstract}

\section{Introduction}
Explaining the answer to a relational query is a fundamental problem in data management \cite{DBS-006,DBLP:conf/pods/GreenT17,chapman2009not,DBLP:journals/pvldb/MeliouRS14,DBLP:journals/pvldb/LeeLG20,whyprovenance,DBLP:journals/vldb/HerschelDL17,cuiweinerwidom,DBLP:conf/sigmod/MiaoZGR19}.
One main approach to explanation is based on attribution, where each tuple from the input database is assigned a score reflecting its contribution to the query answer. A measure that quantifies the contribution of a fact to the query answer is the {\em Banzhaf} value~\cite{Penrose:Banzhaf:1946,Banzhaf:1965}. It has found applications in various domains. Most prominently, it is used as a measure of voting power in the analysis of voting in the Council of the European Union~\cite{Banzhaf:Voting:2012}. It was shown to provide more robust data valuation across subsequent runs of stochastic gradient descent than alternative scores such as the Shapley value~\cite{pmlr-v206-wang23e}. It is used for understanding feature importance in training tree ensemble models, where it is preferable over the Shapley value as it can be computed faster and it can be numerically more robust~\cite{Banzhaf:FeatureImportance:2022}. In Banzhaf random forests~\cite{Banzhaf:RandomForests:2018}, it is used to evaluate the importance of each feature across several possible feature sets used for training random forests. It is also used as a measure of risk analysis in terrorist networks~\cite{Banzhaf:Networks:2023}.

This paper starts a systematic investigation of both theoretical and practical facets of three computational problems for Banzhaf-based fact attribution in query answering: exact computation, approximation, and ranking. 
Our contribution is fourfold. 

\smallskip

\smallskip

{\em 1. Exact Banzhaf Computation.}
We introduce \textsc{ExaBan}, an algorithm that computes the exact Banzhaf scores for the contributions of facts in the answers to positive relational queries (Select-Project-Join-Union in SQL). Its input is the query lineage, which is a Boolean positive function whose variables are the database facts. Its output is the Banzhaf value of each variable. It relies on the compilation of the lineage into a d-tree, a data structure previously used for efficient computation in probabilistic databases~\cite{FinkHO:AdaProb:2013}. The compilation recursively decomposes the function into a disjunction or conjunction of (independent) functions over disjoint sets of variables, or into a disjunction of (mutually exclusive) functions with disjoint sets of satisfying variable assignments. Our use of d-tree is justified by the observation that if we have the Banzhaf values for independent or mutually exclusive functions, we can then compute the Banzhaf values for the conjunction or disjunction of these functions. 
In our experiments with over 300 queries and three widely-known datasets (TPC-H, IMDB, Academic), \textsc{ExaBan} consistently outperforms the state-of-the-art solution~\cite{DeutchFKM:SIGMOD:2022}, which we adapted to compute Banzhaf instead of Shapley values. The performance gap is up to two orders of magnitude on those workloads for which the prior work finishes within one hour, while \textsc{ExaBan} also succeeds to terminate within one hour for 41.7\%-99.2\% (for the different datasets) of the cases for which prior work failed. 

\smallskip

{\em 2. Anytime Deterministic Banzhaf Approximation.}
We also introduce \adaban, an algorithm that computes  approximate Banzhaf values of facts. \adaban\ is an {\em approximation algorithm} in the sense that it computes an interval $[\ell,u]$ that contains the exact Banzhaf value of a given fact. It is {\em deterministic} in the sense that the exact value is guaranteed to be contained in the approximation interval\footnote{This is in stark contrast to randomized approximation schemes, where the exact value is contained in the approximation interval with a probability $\delta\in(0,1)$.}. 
It is {\em anytime} in the sense that it can be stopped at any time and provides a correct approximation interval for the exact Banzhaf value. Each decomposition step  cannot enlarge the approximation interval. 
Given any error $\epsilon\in[0,1]$ and an approximation interval $[\ell,u]$ computed by \adaban, if $(1-\epsilon)u \leq (1+\epsilon)\ell$, then any value in the interval $[(1-\epsilon)u,(1+\epsilon)\ell]$ is a (relative) $\epsilon$-approximation of the exact Banzhaf value. \adaban\ provably reaches the desired approximation error\footnote{In contrast, the randomized approximation schemes cannot guarantee that by executing one more iteration step the approximation interval does not enlarge.} after a number of steps. {\em In the worst case}, any deterministic approximation algorithm needs exponentially many steps in the number of facts\footnote{Otherwise, it would contradict the hardness of exact Banzhaf value computation~\cite{LivshitsBKS:LMCS:2021} that is attained by \textsc{AdaBan} for $\epsilon=0$.}.
Yet in practical settings including our experiments, \adaban's behavior is much better than the theoretical worst case. For instance, \textsc{AdaBan} takes up to one order of magnitude less time than \textsc{ExaBan} to reach $\epsilon=0.1$.

\textsc{AdaBan} has two main ingredients:  (1) the incremental decomposition of the query lineage into a d-tree, and (2) a mechanism to compute lower and upper bounds on the Banzhaf value for a variable in any positive DNF function. 

The first ingredient builds on \textsc{ExaBan}. Unlike \textsc{ExaBan}, \textsc{AdaBan} does not exhaustively compile the lineage into a d-tree before computing the Banzhaf values. Instead, it intertwines the incremental compilation of the lineage with the computation of approximation intervals for the Banzhaf value.
If an interval reaches the desired approximation error, then \textsc{AdaBan} stops the computation; otherwise, it further expands the d-tree. Thus, it may finish after much fewer decomposition steps than \textsc{ExaBan}. This is the main reason behind \textsc{AdaBan}'s speedup over \textsc{ExaBan}, as reported in our experiments.

The second ingredient is the computation of approximation intervals. \textsc{AdaBan} can derive lower and upper bounds on the Banzhaf value for any variable in positive DNF functions at the leaves of a d-tree. While the bounds may be arbitrarily loose, they can be computed in time linear in the function size. Given approximation intervals at the leaves of a d-tree, \textsc{AdaBan} computes an approximation interval for the entire d-tree, and thus for the query lineage.

\smallskip

{\em 3. Banzhaf-based Ranking and Top-$k$ Facts.} 
We also introduce \textsc{IchiBan}, an algorithm that ranks facts and selects the top-$k$ facts based on their Banzhaf values. \textsc{IchiBan} is a natural generalization of \adaban: It incrementally refines the  approximation intervals for the Banzhaf values of all facts until the intervals are separated or become the same Banzhaf value. Two intervals are separated when the lower bound of one becomes larger than the upper bound of the other. \textsc{IchiBan} also supports approximate ranking, where the approximation intervals are ordered by their middle points.
\nop{Like \textsc{ExaBan} and \textsc{AdaBan}, \textsc{IchiBan} shares the function decomposition at each step across all facts and keeps track of approximation intervals for several facts at any  time.} 

The top-$k$ problem is to find $k$ facts whose Banzhaf values are the largest across all facts in the database. To obtain such top-$k$ facts, we proceed similarly to ranking. We start by incrementally tightening the approximation intervals for the Banzhaf values of all facts. Once the approximation interval for a fact is below the lower bound of at least $k$ other facts, we discard that fact from our computation. Alternatively, we can stop the execution when the overlapping approximation intervals reach a given error, at the cost of allowing approximate top-$k$.

Our experiments show that when \textsc{IchiBan} is prompted to produce approximate ranking or top-$k$ results, in practice it achieves near-perfect results. This is true even in cases where previous work ~\cite{DeutchFKM:SIGMOD:2022}, which gives no top-$k$ correctness guarantees, produces inaccurate results.  Furthermore, \textsc{IchiBan} is by up to an order of magnitude faster than computing the exact Banzhaf values.

\smallskip

{\em 4. Dichotomy for Banzhaf-based Ranking.} 
Our fourth contribution is a dichotomy for the complexity of the ranking problem in case of self-join-free Boolean conjunctive queries: Given two facts, deciding whether the Banzhaf value of one fact is greater than the Banzhaf value of the other fact is tractable (i.e., in polynomial time) for hierarchical queries and intractable (i.e., not in polynomial time) for non-hierarchical queries. This dichotomy coincides with the dichotomy for the exact computation of Banzhaf values~\cite{LivshitsBKS:LMCS:2021}. This is surprising, since ranking facts does not require in principle their exact Banzhaf values but just an approximation sufficient to rank them (as done in \textsc{IchiBan}). The tractability for ranking is implied by the tractability for exact computation (since we can first compute the exact Banzhaf values of all facts in polynomial time and then sort the facts by their Banzhaf values), yet the intractability for ranking is {\em not} implied by the intractability for exact computation. Our intractability result relies on the conjecture that an efficient (i.e., polynomial in the inverse of the error and in the graph size) approximation for counting the independent sets in a bipartite graph is not possible~\cite{DyerGGJ03,CurticapeanDFGL19}.

\smallskip

The paper is organized as follows. Sec.~\ref{sec:prelim} introduces the notions of Banzhaf value, Boolean functions, relational databases and queries, and query lineage. Sec.~\ref{sec:algorithm} introduces the algorithms for exact  and approximate computation of Banzhaf values. Sec.~\ref{sec:ranking_top_k} introduces our algorithm for Banzhaf-based top-$k$ and ranking and our dichotomy for  ranking. 
Sec.~\ref{sec:experiments} details our experimental findings. Sec.~\ref{sec:related} contrasts our contributions to prior work on approximate computation and attribution by Shapley values. Sec.~\ref{sec:conc} concludes. Full proofs of formal statements are deferred to the Appendix.

\section{Preliminaries}
\label{sec:prelim}
We denote by $\mathbb{N}$ the set of natural numbers 
including $0$. For $n \in \mathbb{N}$, we denote 
$[n]  \defeq  \{1,2, \ldots, n\}$. In case $n = 0$, we have 
$[n] = \emptyset$. 

\paragraph{Boolean Functions}
Given a set $\bm X$ of Boolean variables, 
a {\em Boolean function} over $\bm X$ is a function 
$\varphi: \bm X \rightarrow \{0,1\}$ defined recursively as: a variable in $\bm X$; a conjunction $\varphi_1\wedge\varphi_2$ or a disjunction $\varphi_1\vee\varphi_2$ of two Boolean functions $\varphi_1$ and $\varphi_2$; or a negation $\neg(\varphi_1)$ of a Boolean function $\varphi_1$.
A {\em literal} is a variable or its negation.
\nop{We also use $\vars(\varphi)=\bm X$ to denote 
the set of variables in $\varphi$.}
The size of $\varphi$, denoted 
by $|\varphi|$, is the number of symbols in $\varphi$. For a variable $x \in \bm X$ and
a constant $b \in \{0,1\}$, $\varphi[x:=b]$ denotes the function that results from replacing $x$ by $b$ in  $\varphi$. 
An {\em assignment} for $\varphi$ is a function 
$\theta: \bm X \rightarrow \{0,1\}$. 
We also denote an assignment $\theta$ by the set 
$\{x\mid \theta(x)=1\}$ of its variables mapped to $1$. 
The Boolean 
value of $\varphi$ under the assignment $\theta$
is denoted by $\varphi[\theta]$.
If $\varphi[\theta] =1$, then
$\theta$ is a {\em satisfying assignment} or 
{\em model} of $\varphi$. 
We denote  the number of 
models of $\varphi$ by $\#\varphi$.
A function is {\em positive} if its literals are positive. \nop{It is in  \pptwodnf if it is positive, in disjunctive normal form (DNF), and its set of variables is partitioned into two disjoint sets $\bm Y$ and $\bm Z$ such that each clause is the conjunction of a variable from $\bm Y$ and a variable from $\bm Z$.}

\begin{definition}[Banzhaf Value of Boolean Variable]
\label{def:banz}
Given a Boolean function $\varphi$ over $\bm X$, 
the {\em Banzhaf value} of a variable $x \in \bm X$ in $\varphi$ is:

\begin{align}
\label{eq:banzhaf_value}
\banz(\varphi, x) \defeq \sum_{\bm Y \subseteq \bm X\setminus \{x\}}
\varphi[\bm Y \cup \{x\}] - \varphi[\bm Y]
\end{align}
\end{definition}
{\em Normalized} versions of the Banzhaf value $\banz(\varphi, x)$ can be obtained by dividing it by (1) the number $2^{|\bm X|-1}$ of all possible assignments of the variables in $\bm X$ except $x$ ({\em Penrose–Banzhaf power}), or by  (2) the sum $\sum_{y \in \bm X} \banz(\varphi, y)$ of the Banzhaf values of all variables ({\em Penrose–Banzhaf index})~\cite{KirschL10}.
In this paper, we use the definition in Eq.~\eqref{eq:banzhaf_value}, but our results immediately apply 
to the normalized versions as well.

\begin{example}
\label{ex:basic_function_banzhaf}
\rm
Consider the Boolean function 
$\varphi = x_1 \vee (x_2 \wedge \neg x_3)$.
The following table shows all possible assignments
$\bm Y$ for $\varphi$ and the Boolean value 
of $\varphi$ under $\bm Y$. For simplicity,
we identify variables by their indices, e.g., $x_1$ is identified by $1$.
\begin{center}
\begin{tabular}{c|cccccccc}
$\theta$ & $\emptyset$ & $\{1\}$ & $\{2\}$ & $\{3\}$ & $\{1,2\}$ & $\{1,3\}$ &$\{2,3\}$ & $\{1,2,3\}$\\ 
\hline
$\varphi[\theta]$ & $0$ & $1$ & $1$ & $0$ & $1$ & $1$& $0$ & $1$ \\
\end{tabular}
\end{center}
Recall the set notation for an assignment; e.g., $\bm Y = \{2,3\}$ means that $x_2=x_3=1$ and $x_1=0$.
To compute the Banzhaf value of $x_1$, we sum up the differences 
$\varphi[\bm Y \cup \{x_1\}] - \varphi[\bm Y]$ for all 
$\bm Y \subseteq \{x_2,x_3\}$:
\begin{align*}
\banz(\varphi, x_1) =&\ \big(\varphi[\{1\}] - \varphi[\emptyset]\big) + 
\big(\varphi[\{1,2\}] - \varphi[\{2\}]\big) + \\ 
&\ \big(\varphi[\{1,3\}] - \varphi[\{3\}]\big) +
\varphi[\{1,2,3\}] - \varphi[\{2,3\}] \\
= &\ 1 + 0 + 1 + 1 = 3
\end{align*}
Similarly, $\banz(\varphi, x_2) = 1$ and $\banz(\varphi, x_3) = -1$.
The latter is negative, because $x_3$ appears  negated in $\varphi$. 
\end{example}

An alternative characterization of the Banzhaf value, adapted from prior work~\cite{LivshitsBKS:LMCS:2021}, is the difference between the numbers of the models of the function where $x$ is set to $1$ and respectively to $0$.

\begin{proposition}
\label{prop:alternative_banzhaf}
The following holds for any Boolean function $\varphi$ over $\bm X$ and variable $x \in \bm X$:
\begin{align}
\label{eq:alternative_banzhaf}
\banz(\varphi, x) = \#\varphi[x:=1] - \#\varphi[x:=0]
\end{align}
\end{proposition}

\begin{example}
\rm
Consider again the function $\varphi = x_1 \vee (x_2 \wedge \neg x_3)$
from Example~\ref{ex:basic_function_banzhaf}.
We compute the Banzhaf value of the variable $x_1$ using
Eq.~\eqref{eq:alternative_banzhaf}.
The function 
$\varphi[x_1:=1] = 1 \vee (x_2 \wedge \neg x_3)$
evaluates to $1$ under any assignment for the variables $x_2$ and $x_3$, 
hence $\#\varphi[x_1:=1] = 4$.
The only model of the function 
$\varphi[x_1:=0] = 0 \vee (x_2 \wedge \neg x_3)$
is $\{x_2\}$, hence $\#\varphi[x_1:=0] = 1$.
We obtain $\banz(\varphi, x_1) = 4 - 1 = 3$, which is the same 
as the value computed in 
Example~\ref{ex:basic_function_banzhaf}.
\end{example}

\paragraph{Databases}
Let a countably infinite set $\Dom$ of constants. 
A {\em database schema}  $\bm S$ is a finite 
set of  {\em relation symbols}, with each 
relation symbol $R$ having a fixed {\em arity}\nop{ $ar(R)$}. 
A database $D$ over ${\bm S}$ associates with each
relation symbol $R$ of arity $k$ a finite $k$-ary relation $R^D \subseteq \Dom^{k}$. 
We identify a database $D$ with
its finite set of {\em facts} $R(c_1, \ldots , c_k)$, stating that the 
$k$-ary relation  $R^D$
contains the tuple $(c_1, \ldots , c_k)$.
\nop{
A database $D$ over $\bm S$ is a finite set 
of {\em facts} of the form $R(c_1, \ldots, c_\ell)$
where $R$ is a relation symbol from $\bm S$
with arity $\ell$ 
and each $c_i$ is a value from $\Dom$.
}
Following prior work, we assume 
that the database is partitioned into a set $D_n$ 
of {\em endogenous}
and a set $D_x$ of {\em exogenous} facts \cite{LivshitsBKS:LMCS:2021}.

\paragraph{Queries}
A {\em conjunctive query} (CQ) over database schema 
$\bm S$ has the form:
%
$Q = \exists \bm Y \bigwedge_{j \in [m]} R_j(\bm Y_j)$, 
%
where $R_j$ is a relation symbol from $\bm S$,  
each $\bm Y_j$ is a tuple of \nop{$ar(R_j)$} variables and constants, and $\bm Y$ is a set of variables included in $\bigcup_{j \in [m]} \bm Y_j$.  
To distinguish variables in queries from variables in Boolean functions, we denote the query variables by uppercase letters and the function variables by lowercase letters.
All variables in $\bm Y$
are {\em bound}, whereas the variables included in 
the set $\bigcup_{j \in [m]}{\bm Y_j}$ but not in $\bm Y$
are {\em free}.
Each $R_j(\bm Y_j)$ is an {\em atom} of $Q$.
We denote by $at(X)$ the set of atoms with the query variable $X$. 
A {\em Boolean} query is a query without free variables.

A CQ is {\em hierarchical} if for any two variables $X$
and $Y$, one of the following conditions holds:
$at(X) \subset at(Y)$, $at(X) \supseteq at(Y)$, 
or $at(X) \cap at(Y) = \emptyset$. A CQ is
{\em self-join free} if there are no two atoms 
with the same relation symbol. 

\begin{example}
\label{ex:query}
\rm
The query $Q = \exists X,Y,Z,V,U\ R(X,Y,Z)\wedge  S(X,Y,V)$ $\wedge T(X,U)$ is hierarchical: $\at(V)\subset\at(Y)\subset\at(X)$, $\at(U)\subset\at(X)$, and $\at(U)\cap\at(Y)=\emptyset$.
The query  $Q = \exists X,Y\ R(X) \wedge  S(X,Y) \wedge T(Y)$ is non-hierarchical: the sets $\at(X)$ $=$ $\{R(X), S(X,Y)\}$ and $\at(Y)$ $=$ $\{T(Y), S(X,Y)\}$ are neither disjoint nor one
is included in the other.
\end{example}

A {\em union of conjunctive queries} (UCQ) has the form 
$Q=Q_1 \vee \cdots \vee Q_n$ where $Q_1,\ldots,Q_n$ are CQs.
The query $Q$ is Boolean if $Q_1,\ldots,Q_n$ are Boolean. 
Given a non-Boolean query $Q$ with free variables $X_1,\ldots,X_n$, a {\em residual} query of $Q$ is a Boolean query, where each free variable $X_i$ is replaced by a constant $a_i$ for $i\in[n]$. We denote this residual query by $Q[a_1/X_1,\ldots,a_n/X_n]$.

Selection conditions of the form $X\ \theta\  \texttt{const}$, where $X$ is a query variable, $\texttt{const}$ is a constant, and the comparison $\theta$ is any of $<, \leq, =, \neq, \geq, >, \geq$, are also supported for practical reasons. UCQs with selections correspond to select-project-join-union queries in SQL. 

\paragraph{Query Lineage}
Let a database $D = D_n \cup D_x$. Each endogenous fact $f$ in $D_n$ is associated with a propositional variable denoted by $v(f)$.
Given a Boolean UCQ $Q$ and a database $D$, the lineage of $Q$ over $D$, denoted by $\varphi_{Q,D}$, is a positive Boolean function in DNF over the variables $v(f)$ of facts $f$ in $D_n$. 
Each clause is a conjunction of $m$ variables, where $m$ is the number of atoms in $Q$. 
We define lineage recursively on the structure of $Q$ (we skip $D$ from the subscript):
\begin{align*}
     \varphi_{Q_1\wedge Q_2} \defeq &\ \varphi_{Q_1}\wedge \varphi_{Q_2}&
    \varphi_{Q_1\vee Q_2} \defeq &\ \varphi_{Q_1}\vee \varphi_{Q_2}\\[5pt]
    \varphi_{\exists X Q} \defeq &\ \bigvee_{a\in\Dom} 
    \varphi_{Q[a/X]} &
    \varphi_{R(t)} \defeq &\
    \begin{cases}
        v(R({\bm t})) & \text{if } R({\bm t})\in D_n\\
        1 & \text{if } R({\bm t})\in D_x\\
        0    & \text{otherwise}
    \end{cases}
\end{align*}
where  $Q[a/X]$ is $Q$ where the variable $X$
is set to the constant $a$. If $Q$ is the conjunction (disjunction)
of subqueries, the lineage of $Q$ is the conjunction (disjunction)
of the lineages of the subqueries. In case of an
existential quantifier $\exists X$, the lineage is the disjunction 
of the lineages of the residual queries obtained
by replacing $X$ with each value in the domain.
If $Q$ is an atom $R(\bm t)$ where all variables are already replaced 
by constants, we check whether $R(\bm t)$ is a fact in the database. 
If it is not, then the Boolean constant $0$ is added to the lineage. Otherwise, we have two cases.
If $R({\bm t})$ is an endogenous fact, then the variable $v(R({\bm t}))$
associated with $R({\bm t})$ is added to the lineage.
If $R({\bm t})$ is an exogenous fact, then the constant $1$ is added instead to the lineage. This means that exogenous facts are not in the lineage, even though they are used to create the lineage.

The lineage for any non-Boolean query $Q$ is defined using the case of Boolean queries. Each tuple in the result of $Q$ defines a residual query of $Q$, which is Boolean and for which we can compute the lineage as defined above. In other words, the lineage of $Q$ is given by the set of lineages of the tuples in the result of $Q$.

\begin{example}
\label{ex:lineage}
\rm
Reconsider the first query $Q$ from Example \ref{ex:query} and the database $D=\{R(1,2,3),$ $S(1,2,4),$ $S(1,2,5),$ $T(1,6)\}$,
where all facts are endogenous. There are two groundings of the query in the database, obtained by replacing $X,Y,Z,V,U$ with $1,2,3,4,6$ respectively or $1,2,3,5,6$ respectively. Each grounding is intuitively an alternative reason for the query satisfaction and yields a clause in the lineage. Thus, the lineage is $\varphi_{Q,D} = [v(R(1,2,3)) \wedge v(S(1,2,4)) \wedge v(T(1,6)) ]$
$ \vee [v(R(1,2,3)) \wedge v(S(1,2,5)) \wedge v(T(1,6))]$.

\end{example}

\paragraph{Banzhaf Values of Database Facts}
We use the Banzhaf value of an endogenous database fact $f$ as a measure of contribution of $f$ to the result of a given query.
An equivalent formulation is via the query lineage: We want the Banzhaf value of the variable $v(f)$ associated with $f$ in the lineage of the query. 

Consider a Boolean query $Q$, a database $D = (D_n,D_x)$, and an endogenous fact $f \in D_n$. Let $v(f)$ be the variable associated to $f$. We define:
\begin{align}
\label{eq:banz_fact}
\banz(Q,D,f) \defeq \banz(\varphi_{Q,D},v(f))
\end{align}
Since the function $\varphi_{Q,D}$ is positive,
it follows from Eq.~\eqref{eq:banzhaf_value}
that $\banz(Q,D,f)$ is the number of subsets 
$D' \subseteq (D_n \setminus \{f\})$
such that $Q(D' \cup D_x) = 0$ and 
$Q(D' \cup D_x \cup \{f\}) = 1$.

For a non-Boolean query $Q$ with free variables $\bm Z$,
the Banzhaf value of $f$ is defined with respect
to a tuple $\bm t$ in the result of $Q$:
%
$$\banz(Q,D,f,\bm t) \defeq \banz(Q[\bm t/\bm Z],D,f)$$
where $Q[\bm t/\bm Z]$ is the Boolean residual query of $Q$, where the tuple of free variables $\bm Z$ is replaced by the tuple $\bm t$ of constant values.

\begin{example}
\rm
Consider again the lineage 
$\varphi_{Q,D}$ from Example~\ref{ex:lineage}. 
We have 
$\varphi_{Q,D}[v(R(1,2,3)):= 1] - 
\varphi_{Q,D}[v(R(1,2,3)):= 0]
= 2-0 = 2$ and 
$\varphi_{Q,D}[v(S(1,2,4)):= 1] - 
\varphi_{Q,D}[v(S(1,2,4)):= 0]
= 2-1 = 1$. Hence, 
$\banz(\varphi_{Q,D}, v(R(1,2,3))) = 
\banz(Q,D,R(1,2,3)) = 2$
and 
$\banz(\varphi_{Q,D}, v(S(1,2,4))) = 
\banz(Q,D,S(1,2,4)) = 1$.
\end{example}
\section{Banzhaf Computation}
\label{sec:algorithm}

This section introduces our algorithmic framework for computing the exact or approximate Banzhaf value for a fact (variable) in a query lineage (Boolean positive DNF function). Sec.~\ref{sec:exact} gives our exact algorithm, which allows us to introduce the building blocks of decomposition trees and formulas for Banzhaf value computation that exploit the independence and mutual exclusion of functions. 
Then, Sec.~\ref{sec:approximate} extends the exact algorithm to an anytime deterministic approximation algorithm, which incrementally refines approximation intervals for the Banzhaf values until the desired error is reached.

\subsection{Exact Computation}
\label{sec:exact}

The main idea of our exact algorithm is as follows. Assume we have the Banzhaf value for a variable $x$ in a function $\varphi_1$. Then, we can compute efficiently the Banzhaf value for $x$ in a function $\varphi = \varphi_1\text{ op }\varphi_2$, where op is one of the logical connectors OR ($\vee$) or AND ($\wedge$) and in case the functions $\varphi_1$ and $\varphi_2$ are independent, i.e., they have no variable in common, or mutually exclusive, i.e., 
they have no satisfying assignment in common. The following formulas make this argument precise, where we keep track of both the Banzhaf value for $x$ in $\varphi$ and also of the model count $\#\varphi$ for $\varphi$:

\begin{itemize}
    \item If $\varphi = \varphi_1 \wedge \varphi_2$ and $\varphi_1$ and 
    $\varphi_2$ are independent, then:
\begin{align}
\#\varphi =&\ \#\varphi_1\cdot \#\varphi_2
\label{eq:ind_and_count} \\
\banz(\varphi,x) =&\ \banz(\varphi_1, x)\cdot \#\varphi_2
\label{eq:ind_and_banz}
\end{align}

\item If $\varphi = \varphi_1 \vee \varphi_2$ and $\varphi_1$ and $\varphi_2$ are independent, then:  
\begin{align}
\#\varphi &=  
 \#\varphi_1 \cdot 2^{n_2} +  2^{n_1} \cdot \#\varphi_2 - \#\varphi_1 \cdot \#\varphi_2
\label{eq:ind_or_count}
\\
\banz(\varphi, x) &=  \banz(\varphi_1,x) \cdot (2^{n_2}- \#\varphi_2),
\label{eq:ind_or_banz}
\end{align}
where $n_i$ is the number of variables in $\varphi_i$ for $i \in [2]$.

\item If $\varphi = \varphi_1 \vee \varphi_2$, and $\varphi_1$ and $\varphi_2$ are mutually exclusive and  over the same variables, then: 
\begin{align}
\#\varphi& = 
\#\varphi_1 +  
\#\varphi_2
\label{eq:mut_excl_count}\\
\banz(\varphi,x)& = 
\banz(\varphi_1, x) +  
\banz(\varphi_2, x)
\label{eq:mut_excl_banz}
\end{align}
\end{itemize}
The derivations of these formulas are given in Appendix~\ref{app:algorithm}.

For functions representing the lineage of hierarchical queries, it is known that they can be decomposed efficiently into independent functions down to trivial functions of one variable~\cite{olteanu2008using}. For such functions, Eq.~\eqref{eq:ind_and_count} to \eqref{eq:ind_or_banz} are then sufficient to compute efficiently the Banzhaf values. For non-hierarchical queries, however, this is not the case. A common general approach, which is widely used in probabilistic databases~\cite{suciu2011probabilistic} and exact Shapley computation~\cite{DeutchFKM:SIGMOD:2022}, and borrowed from knowledge compilation~\cite{darwiche2002knowledge}, is to decompose, or {\em compile}, the query lineage into an equivalent Boolean function, where all logical connectors are between functions that are either independent or mutually exclusive. While in the worst case this necessarily leads to a blow-up in the number of decomposition steps (unless P=NP), it turns out that in many practical cases (including our own experiments), this number remains reasonably small.

\nop{If the input function is a composition of independent and mutually exclusive functions, then we can compute the Banzhaf values as stated above. Remarkably, the lineage of {\em hierarchical} queries admits a linear-size representation as nested conjunctions and disjunctions of independent functions~\cite{olteanu2008using}, so Eq.~\eqref{eq:ind_and_count} to \eqref{eq:ind_or_banz} are sufficient to compute efficiently the Banzhaf values. This is however not applicable for the lineage of non-hierarchical queries.

The algorithm sketched above can be immediately extended to compute the Banzhaf values for all variables together, thereby sharing the cost of the decomposition. In Sec.~\ref{sec:approximate} we turn this exact algorithm into an approximate one.
}

In this paper, we compile the query lineage into a {\em decomposition tree}~\cite{FinkHO:AdaProb:2013}. Such trees have inner nodes that are the logical operators enhanced with information about independence and mutual exclusiveness of their children: $\otimes$ stands for independent-or, $\odot$ for independent-and, and $\oplus$ for mutual exclusion.

\begin{definition}~\cite{FinkHO:AdaProb:2013}
A {\em decomposition tree}, or d-tree for short, is defined recursively as follows:

\begin{itemize}

\item Every function $\varphi$ is a d-tree for $\varphi$.

\item If $T_\varphi$ and $T_\psi$ are d-trees for independent functions $\varphi$ and respectively $\psi$, then 
\begin{center}
\begin{minipage}{0.2\linewidth}
\tikz {
 \node at (3.6,-1)  (n4) {$\otimes$};
\node at (3.2,-1.75)  (n3) {$T_\varphi$} edge[-] (n4);
\node at (4.0,-1.75)  (n3) {$T_\psi$} edge[-] (n4);
}
\end{minipage}
\hspace{1em}
and
\hspace{1em}
\begin{minipage}{0.2\linewidth}
\tikz {
 \node at (3.6,-1)  (n4) {$\odot$};
\node at (3.2,-1.75)  (n3) {$T_\varphi$} edge[-] (n4);
\node at (4.0,-1.75)  (n3) {$T_\psi$} edge[-] (n4);
}
\end{minipage}
\end{center}
are d-trees for $\varphi \vee \psi$ and respectively $\varphi \wedge \psi$.

\item If $T_\varphi$ and $T_\psi$ are d-trees for mutually exclusive functions $\varphi$ and respectively $\psi$, then  
\begin{center}
\begin{minipage}{0.2\linewidth}
\tikz {
 \node at (3.6,-1)  (n4) {$\oplus$};
\node at (3.2,-1.75)  (n3) {$T_\varphi$} edge[-] (n4);
\node at (4.0,-1.75)  (n3) {$T_\psi$} edge[-] (n4);
}
\end{minipage}
\end{center}
is a d-tree for $\varphi \vee \psi$.
\end{itemize}
A d-tree, whose leaves are Boolean constants or literals, is {\em complete}.  
\end{definition}

Any Boolean function can be compiled into a complete d-tree by decomposing it into conjunctions or disjunctions of independent functions or into disjunctions of mutually exclusive functions. The latter is always possible via Shannon expansion: Given a function $\varphi$ and a variable $x$, $\varphi$ can be equivalently expressed as the disjunction of two mutually exclusive functions defined over the same variables as $\varphi$:
$\varphi = (x \wedge \varphi[x:=1]) \vee (\neg x \wedge \varphi[x:=0])$.
This expression yields the d-tree: 
$(x \odot \varphi[x:=1]) \oplus (\neg x \odot \varphi[x:=0])$.
The details of d-tree construction are given in prior work~\cite{FinkHO:AdaProb:2013}. In a nutshell, it first attempts to partition the function into independent functions using a standard algorithm for finding connected components in a graph representation of the function. If this fails, then it applies Shannon expansion on a variable that appears most often in the function (other heuristics are possible, e.g., pick variables whose conditioning allow for independence partitioning). The functions $\varphi[x:=1]$ and $\varphi[x:=0]$ are subject to standard simplifications for conjunctions and disjunctions with the constants $0$ and $1$. \nop{Shannon expansion is the most expensive step, as it may create two functions whose number of literals is that of the input function minus one.}
In the worst case, d-tree compilation may (unavoidably) require a number of Shannon expansion steps exponential in the number of variables.

\begin{example}
\label{ex:d-tree}
\rm
We construct a d-tree for the Boolean function 
$\varphi = (x \wedge y) \vee (x \wedge z)$.
We first observe that 
the two conjunctive clauses are not independent, so 
we apply Shannon expansion on $x$ and decompose the function into 
the two mutually exclusive functions 
$\varphi_1 = x \wedge \varphi[x:=1] = x\wedge (y\vee z)$ and 
$\varphi_0 = \neg x \wedge \varphi[x:=0] = 0$. 
The left branch representing $\varphi_1$ can be further decomposed into independent functions until we obtain a complete d-tree:

\begin{minipage}{\linewidth}
\begin{center}
\tikz {
 \node at (0,0)  (root) {$\oplus$};
\node at (-1.5,-0.75)  (c1) {$\odot$} edge[-] (root);
\node at (1.5,-0.75)  (c2) {$0$} edge[-] (root);
\node at (-3,-1.5)  (c21) {$x$} edge[-] (c1);
\node at (0,-1.5)  (c22) {$\otimes$} edge[-] (c1);
\node at (-1.5,-2.25)  (c31) {$y$} edge[-] (c22);
\node at (1.5,-2.25)  (c32) {$z$} edge[-] (c22);
}
\end{center}
\end{minipage}

\noindent
Alternatively, we can factor out $x$ to obtain the function $x \wedge (y \vee z)$, and compile it into the d-tree $x \odot (y \otimes z)$. Our algorithm computing d-trees does this whenever a variable occurs in all clauses. 
\end{example}

Fig.~\ref{alg:exactban} gives our algorithm $\textsc{ExaBan}$ that computes the exact Banzhaf value for any variable $x$ in an input function $\varphi$. It takes as input a complete d-tree for $\varphi$ and uses Eq.~\eqref{eq:ind_and_count} to~\eqref{eq:mut_excl_banz} to express the Banzhaf value of a variable $x$ in a function $\varphi$ represented by a d-tree $T_\varphi$ using the Banzhaf values of $x$ in sub-trees $T_{\varphi_1}$ and $T_{\varphi_2}$.

\begin{figure}[t]
\begin{center}
  \renewcommand{\arraystretch}{1.15}
  \begin{tabular}{@{\hskip 0.1in}l}
  \toprule
  $\textsc{ExaBan}$(d-tree $T_{\varphi}$ for function $\varphi$, variable $x$) \\ outputs $(\banz(\varphi,x), \#\varphi)$\\[0.2ex]
  \midrule
    $B$ := $0$; \TAB \# := $0$; \TAB \texttt{//initialization} \\
 \SWITCH $T_\varphi$\\
 \TAB\CASE $x$:  \STAB $B$ := $1$; $\#$ := $1$\\
 \TAB\CASE $\neg x$: $B$ := $-1$; $\#$ := $1$\\
 \TAB\CASE $1$ or a literal not $x$ nor $\neg x$: $B$ := $0$; $\#$ := $1$\\
 \TAB\CASE $0$: $B$ := $0$; $\#$ := $0$\\
 \TAB\CASE $T_{\varphi_1} \text { op } T_{\varphi_2}$:  \\ 
  \TAB\TAB $(B_i, \#_i)$ := $\textsc{ExaBan}(T_{\varphi_i}, x)$ for $i \in [2]$\\
  \TAB\TAB $n_i$ := number of variables in $T_{\varphi_i}$ for $i \in [2]$\\
  \TAB\TAB\SWITCH \text{ op }\\
  \TAB\TAB\TAB\CASE $\odot$: \TAB \texttt{//wlog if $x$ is in $\varphi$,then it is in $\varphi_1$} \\
  \TAB\TAB\TAB $B$ := $B_1 \cdot \#_2$; \TAB $\# $ := $\#_1 \cdot \#_2$ \\ 
 \TAB\TAB\TAB\CASE $\otimes$: \TAB \texttt{//wlog if $x$ is in $\varphi$,then it is in $\varphi_1$} \\
 \TAB\TAB\TAB\TAB $B$ := $B_1\cdot (2^{n_2} - \#_2)$; \TAB $\#$ := $\#_1\cdot 2^{n_2} + 2^{n_1} \cdot \#_2 - \#_1 \cdot \#_2$\\
 \TAB\TAB\TAB\CASE $\oplus$: \TAB \texttt{//wlog $\varphi_1$ and $\varphi_2$ have same variables}\\
  \TAB\TAB\TAB\TAB $B$ := $B_1+ B_2 $; \TAB $\#$ := $\#_1+ \#_2$\\
  \RETURN $(B,\#)$ \\
  \bottomrule
  \end{tabular}
\end{center}
  \caption{Computing the exact Banzhaf value for a variable $x$ and the model count over a complete d-tree.}
  \label{alg:exactban}
\end{figure}

\begin{proposition}
\label{prop:exaban_correct}
For any positive DNF function $\varphi$, complete d-tree $T_{\varphi}$ for $\varphi$, and variable $x$ in $\varphi$, 
it holds 
$$\textsc{ExaBan}(T_{\varphi}, x) = (\banz(\varphi,x), \#\varphi).$$
\end{proposition}

\begin{example}
\label{ex:d-tree_banzhaf}
\rm
We next show the trace of the computation of \textsc{ExaBan} for the input d-tree from Ex.~\ref{ex:d-tree} and the variable $x$. Each node of the d-tree is labelled by the pair of the Banzhaf value and the model count computed for the subtree rooted at that node:

\begin{minipage}{\linewidth}
\begin{center}
\tikz {
 \node at (0,0)  (root) {$\oplus$};
\node at (-1.5,-0.75)  (c1) {$\odot$} edge[-] (root);
\node at (1.5,-0.75)  (c2) {$0$} edge[-] (root);
\node at (-3,-1.5)  (c21) {$x$} edge[-] (c1);
\node at (0,-1.5)  (c22) {$\otimes$} edge[-] (c1);
\node at (-1.5,-2.25)  (c31) {$y$} edge[-] (c22);
\node at (1.5,-2.25)  (c32) {$z$} edge[-] (c22);

 \node at (0.5,0)  (x) {$(3,3)$};
\node at (-1,-0.75)  (x) {$(3,3)$};
\node at (2,-0.75)  (x) {$(0,0)$};
\node at (-2.5,-1.5)  (x) {$(1,1)$};
\node at (0.5,-1.5)  (x) {$(0,3)$};
\node at (-1,-2.25)  (x) {$(0,1)$};
\node at (2,-2.25)  (x) {$(0,1)$};
}
\end{center}
\end{minipage}

\noindent
The values $(3,3)$ at the left child node of the root are computed as follows. 
This node is an independent-and ($\odot$). The variable $x$ is in the left subtree. \textsc{ExaBan} computes the Banzhaf value $3$ of 
$x$ by multiplying the Banzhaf value $1$ at the left child node with the 
model count $3$ at the right child node.
The model count of $3$ is obtained by multiplying the model counts 
at the child nodes. 
The function represented by the tree rooted at this $\odot$-node is 
$\varphi_1 = x \wedge (y \vee z)$.
Indeed, every model of the function must satisfy $x$ and at least one 
of $y$ and $z$, which implies $\#\varphi_1 = 3$.
Using Eq.~\eqref{eq:alternative_banzhaf}, we have 
$\banz(\varphi_1,x)= \varphi_1[x:=1]- \varphi_1[x:=0] = 3-0 = 3$.
\end{example}

\textsc{ExaBan} can be immediately generalized to compute the Banzhaf values for any number of variables $x_1,\ldots,x_n$. For all variables, it uses the same d-tree and shares the computation of the counts $\#_i$. \nop{Fig.~\ref{alg:exactban_star} gives the generalization of the independent-and case. The other cases are treated similarly.}

\nop{
\begin{figure}
\begin{center}
  \renewcommand{\arraystretch}{1.15}
  \begin{tabular}{@{\hskip 0.1in}l}
  \toprule
  $\textsc{ExaBan}$(d-tree $T_\varphi$ for function $\varphi$, variables $x_1, \ldots, x_n$) \\ 
  outputs $(\banz(\varphi,x_1), \ldots , \banz(\varphi,x_n)$,  $\#\varphi$)\\[0.2ex]
  \midrule
$\ldots$ \\
\CASE $T_{\varphi_1} \text{ op } T_{\varphi_2}$: \\
  \TAB $(B^{(1)}_i, \ldots , B^{(n)}_i, \#_i)$ := $\textsc{ExaBan}(T_{\varphi_i}, x_1,\ldots,x_n)$ for $i \in [2]$\\
  \TAB\SWITCH \text{ op } \\
  \TAB\TAB\CASE $\odot$: \texttt{//wlog if any $x_j$ is in $\varphi$,then it is in $\varphi_1$}\\
  \TAB\TAB\TAB \FOREACH $j = 1, \ldots, n$ \DO $D^{(j)}$ := $B^{(j)}_1 \cdot \#_2$ \\
  \TAB\TAB\TAB $\#$ := $\#_1 \cdot \#_2$ \\
  $\ldots$\\
  \RETURN $(D^{(1)}, \ldots , D^{(n)}, \#)$ \\
  \bottomrule
  \end{tabular}
  \end{center}
\caption{Computing exact Banzhaf values for several variables and the model count over a complete d-tree.}
  \label{alg:exactban_star}
\end{figure}
}
\subsection{Anytime Deterministic Approximation}
\label{sec:approximate}

As explained in Sec.~\ref{sec:exact}, to obtain exact Banzhaf values
for the variables in a function, we first compile the function into a complete d-tree and then compute in a bottom-up traversal of the d-tree the exact Banzhaf values and model counts at each node of the d-tree. 
Approximate computation does not require in general a complete d-tree for the function. In this section, we introduce an anytime deterministic approximation algorithm, called \adaban, that {\em gradually} expands the d-tree and computes after each expansion step upper and lower bounds on the Banzhaf values and model counts for the new leaves. It then uses the bounds to compute an approximation interval for the partial d-tree. If the approximation interval meets the desired error, it stops. Otherwise, it continues with the function compilation and bounds computation at another leaf in the d-tree. Eventually, the approximation interval becomes tight enough to meet the allowed error.
Unlike \textsc{ExaBan}, \adaban\  merges the construction of the d-tree with the computation of the bounds so it can intertwine them at each expansion step.

Sec.~\ref{sec:bounds_banz_count} explains how to efficiently compute upper and lower bounds for positive DNF functions, albeit without any error guarantee. Sec.~\ref{sec:refine_bounds_d-tree} introduces \adaban, which uses such bounds to compute  approximation intervals and incrementally refine them.

\subsubsection{Efficient Computation of Lower and Upper Bounds for Positive DNF Functions}
\label{sec:bounds_banz_count}

We introduce two procedures $L$ (for lower bound) and $U$ (for upper bound) that map any positive DNF function $\varphi$ to positive DNF functions that enjoy the following four desirable properties: (1) $L(\varphi)$ and $U(\varphi)$ admit linear-time computation of model counting; (2) $L(\varphi)$ and $U(\varphi)$ can be synthesized from $\varphi$ in time linear in the size of $\varphi$; (3) the number of models of $L(\varphi)$ is less than or equal to the number of models of $\varphi$, which in turn is less than or equal to the number of models of $U(\varphi)$; and (4) lower and upper bounds on the Banzhaf value of $x$ in $\varphi$ can be obtained by applying $L$ and $U$ to the functions $\varphi[x:=0]$ and $\varphi[x:=1]$.
 
The co-domain of $L$ and $U$ is the class of iDNF functions~\cite{FinkHO:AdaProb:2013}, which are positive DNF functions where every variable occurs once. Whereas the first three aforementioned properties are already known to hold for iDNF functions~\cite{FinkHO:AdaProb:2013}, the fourth one is new and key to our approximation approach. 

For the first property, we note that since each variable in an iDNF function only occurs once, we can decompose the function in linear time into a complete d-tree with $\odot$ or $\otimes$ as inner nodes and literals or constants at leaves. Then, we can traverse the d-tree bottom up and
use Eq.~\eqref{eq:ind_and_count} and \eqref{eq:ind_or_count} to compute at 
each node the model count for the function represented by the subtree rooted at that node. Overall, model counting for iDNF functions takes linear time.

For the second property, we explain the procedures $L$ and $U$ for a given DNF function $\varphi$. The iDNF function $L(\varphi)$ is any subset of the clauses such that no two selected clauses share variables.
The iDNF function $U(\varphi)$ is a transformation of $\varphi$, where we keep one occurrence of each variable and eliminate all other occurrences.

The third and fourth properties follow by Prop.~\ref{prop:banz_count_bounds}:

\begin{proposition}
\label{prop:banz_count_bounds}
For any positive DNF function $\varphi$ and variable $x$ in $\varphi$, 
it holds: 
\begin{align*}
\#L(\varphi) &\leq \#\varphi \leq \#U(\varphi) \\
\#L(\varphi[x:=1]) - \#U(\varphi[x:=0]) &\ \leq  \banz(\varphi,x) \\
&\ \leq   \#U(\varphi[x:=1]) - \#L(\varphi[x:=0])
\end{align*}
\end{proposition}


\begin{figure}[t]
\begin{center}
  \renewcommand{\arraystretch}{1.15} 
 \begin{tabular}{@{\hskip 0.1in}l}
  \toprule
  $\textsc{bounds}$(d-tree $T_ {\varphi}$ for function $\varphi$, variable $x$) \\
  outputs lower and upper bounds for $\banz(\varphi,x)$ and $\#\varphi$\\[0.2ex]
  \midrule
  $(L_b,L_\#,U_b,U_\#) := (0,0,0,0)$ \texttt{// Initialize the bounds}\\
   \SWITCH $T_\varphi$\\
    \TAB\CASE  literal or constant $\ell$: \\
    \TAB\TAB $(L_b,L_{\#}):= (U_b,U_{\#})$ := $\textsc{ExaBan}(\ell,x)$ \\
    \TAB\CASE  non-trivial leaf $\psi$: \TAB 
    \texttt{//no literal nor constant}\\
    \TAB\TAB \texttt{//Compute bounds by Prop.~\ref{prop:banz_count_bounds}}\\
    \TAB\TAB $L_b$ := $\#L(\psi[x:=1]) - \#U(\psi[x:=0])$ \\
    \TAB\TAB $U_b$ := $\#U(\psi[x:=1]) - \#L(\psi[x:=0])$ \\
    \TAB\TAB $L_\#$ := $\#L(\psi)$;\TAB $U_\#$ := $\#U(\psi)$ \\
    \TAB\CASE $T_{\varphi_1}\text {op } T_{\varphi_2}$: \\
    \TAB\TAB $(L^{(i)}_b,L^{(i)}_\#,U^{(i)}_b,U^{(i)}_\#)$ := \textsc{bounds}($T_{\varphi_i}, x$), for $i\in[2]$ \\
    \TAB\TAB $n_i$ := number of variables in $\varphi_i$, for $i\in[2]$\\
    \TAB\TAB\SWITCH \text{op } \\
    \TAB\TAB\TAB\CASE $\odot$: \texttt{//wlog if $x$ is in 
$\varphi$, then it is in $\varphi_1$}\\
    \TAB\TAB\TAB\TAB $L_b$ := $L^{(1)}_b\cdot L^{(2)}_\#$; \TAB $U_b$ := $U^{(1)}_b\cdot U^{(2)}_\#$\\
    \TAB\TAB\TAB\TAB $L_\#$ := $L^{(1)}_\#\cdot L^{(2)}_\#$; \TAB $U_\#$ := $U^{(1)}_\#\cdot U^{(2)}_\#$\\
    \TAB\TAB\TAB\CASE $\otimes$: \texttt{//wlog if $x$ is in 
$\varphi$, then it is in $\varphi_1$}\\    
    \TAB\TAB\TAB\TAB $L_b$ := $L^{(1)}_b\cdot (2^{n_2} - U^{(2)}_\#)$;
     $U_b$ := $U^{(1)}_b\cdot (2^{n_2} - L^{(2)}_\#)$\\
    \TAB\TAB\TAB\TAB $L_\#$ := $L^{(1)}_\#\cdot 2^{n_2} + L^{(2)}_\#\cdot 2^{n_1} - L^{(1)}_\#\cdot L^{(2)}_\#$\\
    \TAB\TAB\TAB\TAB $U_\#$ := $U^{(1)}_\#\cdot 2^{n_2} + U^{(2)}_\#\cdot 2^{n_1} - U^{(1)}_\#\cdot U^{(2)}_\#$\\  
    \TAB\TAB\TAB\CASE $\oplus$: \texttt{//wlog $\varphi_1$ and $\varphi_2$ have same variables}\\
    \TAB\TAB\TAB\TAB $L_b$ := $L^{(1)}_b + L^{(2)}_b$;\TAB $U_b$ := $U^{(1)}_b + U^{(2)}_b$ \\
    \TAB\TAB\TAB\TAB $L_\#$ := $L^{(1)}_\# + L^{(2)}_\#$; \TAB $U_\#$ := $U^{(1)}_\# + U^{(2)}_\#$ \\
    \RETURN $(L_b,L_\#,U_b,U_\#)$ \\
  \bottomrule
  \end{tabular}
  \end{center}
  \caption{Computation of bounds for the Banzhaf value $Banzhaf(\varphi,x)$ and model count $\#\varphi$, given a (possibly partial) d-tree $T_\varphi$ for the function $\varphi$ and a variable $x$.}
  \label{alg:adaban-bounds}
\end{figure}


\begin{example}
\label{ex:construct_lower_upper_bound}
\rm
Consider the DNF function 
$\varphi = 
(x \wedge y) \vee (x \wedge z) \vee u$.
The function is a disjunction of two independent functions 
$\varphi_1 = (x \wedge y) \vee (x \wedge z)$
and 
$\varphi_2 = u$. Since $\varphi_1$ is the
function analyzed in Ex.~\ref{ex:d-tree_banzhaf}, 
we know that $\banz(\varphi_1,x) = \#\varphi_1 = 3$.
Also, it holds
$\banz(\varphi_2,x) = 0$ and $\#\varphi_2 = 1$.
Using Eq.~\eqref{eq:ind_or_count} and~\ref{eq:ind_or_banz},
we obtain
\begin{align*}
&\banz(\varphi,x) = \banz(\varphi_1,x) \cdot (2^1-1)= 3\cdot 1  = 3 \\
&\#\varphi = \#\varphi_1\cdot\#\varphi_2 +\#\varphi_1\cdot(2^1-1)+(2^3-\#\varphi_1)\cdot\#\varphi_2  = 3 + 3 + 5 = 11.
\end{align*}
The functions 
$\varphi[x:=0] = (0 \wedge y) \vee (0 \wedge z) \vee u$ and 
$\varphi[x:=1] = (1 \wedge y) \vee (1\wedge z) \vee u = y\vee z\vee u$
are in iDNF, so it holds
$L(\varphi[x:=0]) = U(\varphi[x:=0]) = \varphi[x:=0]$ and
$L(\varphi[x:=1]) = U(\varphi[x:=1]) = \varphi[x:=1]$. Note that 
$\varphi[x:=0] = u$, yet the function  is defined over three variables, which is important for computing its correct model count.

We may also obtain the following iDNF functions: 
$L(\varphi) = (x \wedge y) \vee u$
by skipping the clause $(x \wedge z)$ in $\varphi$;
and  
$U(\varphi) = (x \wedge y) \vee z \vee u$
by removing $x$ from the second clause of $\varphi$. 
Using Eq.~\eqref{eq:ind_and_count} and \eqref{eq:ind_or_count}:
\begin{align*}
\#L(\varphi[x:=0]) =&\ \#U(\varphi[x:=0]) = 4,\\ 
\#L(\varphi[x:=1]) =&\ \#U(\varphi[x:=1]) = 7, \\
\#L(\varphi) = &\ 5,  \text{ and } \#U(\varphi) =  13.
\end{align*}
Hence, it indeed holds that
$\#L(\varphi) = 5 \leq \#\varphi = 11 \leq \#U_\varphi = 13$
and 
$\#L(\varphi[x:=1])- \#U(\varphi[x:=0]) = 3 \leq \banz(\varphi,x) = 3 \leq 
\#U(\varphi[x:=1])- \#L(\varphi[x:=0]) = 3$.
\end{example}

\subsubsection{Efficient Computation of Lower and Upper Bounds for D-trees}
\label{sec:bounds_banz_count_dtree}
The procedure $\textsc{bounds}$ in Fig.~\ref{alg:adaban-bounds} computes lower and upper bounds on the Banzhaf value and model count for any d-tree, whose leaves are positive DNF functions, (possibly negated) literals, or constants. It does so in linear time in one bottom-up pass over the d-tree.

The procedure takes as input a d-tree $T_\varphi$ for a function $\varphi$
and a variable $x$ for which we want to compute the Banzhaf value.
At a leaf $\ell$ of $T_\varphi$ that is a literal or a constant, it calls
$\textsc{ExaBan}(\ell, x)$ to compute the exact Banzhaf value 
and model count for $\ell$.
At a leaf $\psi$ that is not a literal nor a constant,
the algorithm first computes the iDNF functions  
$L(\psi)$, $U(\psi)$, $L(\psi[x:=b])$, and 
$U(\psi[x:=b])$ for $b \in \{0,1\}$.
By Prop.~\ref{prop:banz_count_bounds}, these functions can be used to derive lower and upper bounds on $\banz(\psi,x)$ and $\#\psi$. 
If $T_{\varphi}$ has children, then it recursively computes bounds on them and then combines them into bounds for itself.
We next discuss the lower bound for the Banzhaf value of $x$ in case $\varphi$ is a disjunction of independent functions $\varphi_1$ and $\varphi_2$. The other cases are handled analogously. 
By Eq.~\eqref{eq:ind_or_banz}, the formula for the exact Banzhaf value is $\banz(\varphi, x) =  \banz(\varphi_1,x) \cdot (2^{n_2}- \#\varphi_2)$.
 To obtain a lower bound on $\banz(\varphi, x)$, we replace the term 
 $\banz(\varphi_1,x)$ by its lower bound and the term 
 $\#\varphi_2$ by its upper bound. The reason for using the upper bound is that the term occurs negatively. 

\begin{example}
\label{ex:d-tree_bounds}
\rm
Consider the following partial d-tree representing a function $\varphi$.
Each node is assigned a quadruple of bounds for the Banzhaf value 
of some variable $x$ and the model count for the d-tree rooted at that node. 
Following the notation in the procedure $\textsc{bounds}$ in Fig.~\ref{alg:adaban-bounds}, 
the first and the third entry in a quadruple are the lower and respectively upper bound 
for the Banzhaf value; the second and the fourth entry are the lower and respectively upper 
bound for the model count. For the computation of the bounds at the 
node $\otimes$ assume that each of the functions $\psi_i$ 
has four variables.

\begin{minipage}{\linewidth}
\begin{center}
\tikz {
 \node at (0,0)  (root) {$\oplus$};
\node at (-2,-0.75)  (c1) {$\otimes$} edge[-] (root);
\node at (2,-0.75)  (c2) {$\odot$} edge[-] (root);
\node at (-3,-1.5)  (c21) {$\psi_1$} edge[-] (c1);
\node at (-1,-1.5)  (c22) {$\psi_2$} edge[-] (c1);
\node at (1,-1.5)  (c23) {$\varphi_1$} edge[-] (c2);
\node at (3,-1.5)  (c24) {$\varphi_2$} edge[-] (c2);
 
 \node at (0,0.5)  (x) {
 \begin{tabular}{cccc}
 $L_b$& $L_{\#}$ & $U_b$ & $U_{\#}$ \\
 ($43,$ & $219,$ & $136,$ & $374$)
 \end{tabular}
 };
 
 
\node at (-2.8,-0.5)  (x) {$(18,184,64, 214)$};
\node at (2.7,-0.5)  (x) {$(25,35,72, 160)$};
\node at (-3,-1.9)  (x) {$(3,7,8,9)$};
\node at (-1,-1.9)  (x) {$(0,8,0,10)$};
\node at (1,-1.9)  (x) {$(5,7,9,20)$};
\node at (3,-1.9)  (x) {$(0,5,0,8)$};
}
\end{center}
\end{minipage}

Assume we have already computed the bounds for the leaves
of the d-tree.
We explain how the procedure $\textsc{bounds}$ uses these bounds 
to derive bounds for the Banzhaf values at the nodes $\odot$ and $\oplus$. 
Assume that the variable $x$ appears in
$\varphi_1$ but not in  $\varphi_2$.
At the node $\odot$, the lower bound for the Banzhaf value is 
$5 \cdot 5 = 25$ and its upper bound is $9\cdot 8= 72$.
Similarly, at the node $\oplus$, the lower and upper bounds 
for the Banzhaf value   
are $L_b = 18 + 25 = 43$ and respectively $U_b = 64+72= 136$.

We cannot use the bounds $L_b$ and $U_b$ to derive a 
$0.5$-approximation for the Banzhaf value, since 
$(1- 0.5) \cdot U_b = 68$ is larger than  $(1+0.5) \cdot L_b = 64.5$.
However, every value within the interval from 
$(1- 0.6) \cdot U_b = 14.4$ to 
$(1+ 0.6) \cdot L_b = 68.8$ is a $0.6$-approximation. 
For instance, it holds that $20 \geq (1- 0.6) \cdot U_b \geq (1- 0.6) \cdot \banz(\varphi,x)$
and 
$20 \leq (1+0.6) \cdot L_b \leq (1+ 0.6) \cdot \banz(\varphi,x)$.

\nop{
The values $(3,3)$ at the left child node of the root are computed as follows. 
This node is an independent-and ($\odot$). The variable $x$ is in the left subtree. \textsc{ExaBan} computes the Banzhaf value $3$ of 
$x$ by multiplying the Banzhaf value $1$ at the left child node with the 
model count $3$ at the right child node.
The model count of $3$ is obtained by multiplying the model counts 
at the child nodes. 
The function represented by the tree rooted at this $\odot$-node is 
$\varphi_1 = x \wedge (y \vee z)$.
Indeed, every model of the function must satisfy $x$ and at least one 
of $y$ and $z$, which implies $\#\varphi_1 = 3$.
Using Eq.~\ref{eq:alternative_banzhaf}, we have 
$\banz(\varphi_1,x)= \varphi_1[x:=1]- \varphi_1[x:=0] = 3-0 = 3$.
}
\end{example}

Eq.~\eqref{eq:ind_and_count} to \eqref{eq:mut_excl_banz}
and Prop.~\ref{prop:banz_count_bounds} imply:

\begin{proposition}
\label{prop:bounds_correct}
For any positive DNF function $\varphi$, d-tree $T_{\varphi}$ for $\varphi$,  
and variable $x$ in $\varphi$, it holds $\textsc{bounds}(T_{\varphi}, x) = (L_b, L_{\#}, U_b, U_{\#})$
such that 
$L_b \leq \banz(\varphi,x) \leq U_b$ 
and 
$L_{\#} \leq \#\varphi \leq U_{\#}$.
\end{proposition}

\subsubsection{Refining Bounds for D-Trees}
\label{sec:refine_bounds_d-tree}


\begin{figure}[t]
\begin{center}
  \renewcommand{\arraystretch}{1.15}
  \begin{tabular}{@{\hskip 0.1in}l}
  \toprule
  $\textsc{\adaban}$(d-tree $T_ {\varphi}$, variable $x$, error $\epsilon$, bounds $[L,U]$)\\
    outputs bounds for $\banz(\varphi,x)$ satisfying relative error $\epsilon$ \\[0.2ex]
  \midrule
  $(L_b\ ,\ \cdot\ \ ,\ U_b\ ,\ \cdot\ )$ := $\textsc{bounds}(T_{\varphi},x)$ \TAB\texttt{//get bounds on $T_{\varphi}$}\\
  $\ell := u := 0$ \TAB\TAB\TAB\texttt{//initialize the bounds to return}\\
  $L$ := $\max\{L,L_b\}$; $U$ := $\min\{U,U_b\}$ \TAB\texttt{//update bounds}\\
  \IF $(1-\epsilon)\cdot U - (1+\epsilon) \cdot L \leq 0$ \TAB\TAB\STAB\texttt{//error satisfied}\\
  \TAB $\ell$ := $(1-\epsilon)\cdot U$; $u$ := $(1+\epsilon)\cdot L$\\
  \ELSE \\
  \TAB pick a non-trivial leaf $\psi$ of $T_{\varphi}$ \TAB\texttt{//no literal/constant}\\ 
  \TAB \SWITCH $\psi$\\
  \TAB\TAB \CASE $\psi_1 \wedge \psi_2$ for independent $\psi_1$ and $\psi_2$:\\
   \TAB\TAB\TAB replace $\psi$ by $\psi_1 \odot \psi_2$ in $T_{\varphi}$\\
  \TAB\TAB \CASE $\psi_1 \vee \psi_2$ for independent $\psi_1$ and $\psi_2$:\\
   \TAB\TAB\TAB replace $\psi$ by $\psi_1 \otimes \psi_2$ in $T_{\varphi}$\\
   \TAB\TAB \DEFAULT: \\
   \TAB\TAB\TAB pick a variable $y$ in $\psi$ \\
   \TAB\TAB\TAB replace $\psi$ by $(y \odot \psi[y:=1]) \oplus 
   (\neg y \odot \psi[y:=0])$ in $T_{\varphi}$ \\
 \TAB $[\ell,u]:= \textsc{\adaban}(T_ {\varphi}, x, \epsilon, [L,U])$\\
 \RETURN $[\ell,u]$\\
  \bottomrule
  \end{tabular}
\end{center}
  \caption{Computing approximate Banzhaf values with relative error $\epsilon$ using incremental decomposition and bound refinement.}
  \label{alg:adaban}
\end{figure}

Fig.~\ref{alg:adaban} introduces our approximation algorithm $\adaban$. It takes as input a partial d-tree $T_{\varphi}$, a variable $x$, a relative error $\epsilon$, and initial trivial bounds $[0,2^{n-1}]$ on $\banz(\varphi,x)$, where $n$ is the number of variables in $\varphi$. 
It then computes an interval of $\epsilon$-approxima\-tions for $\banz(\varphi,x)$.
\nop{It uses recursion to gradually improve the approximation interval.}
First, it calls the procedure \textsc{bounds} from 
Fig.~\ref{alg:adaban-bounds} to obtain 
a lower bound $L_b$ and an upper bound $U_b$
for $\banz(\varphi,x)$
based on the current partial d-tree  $T_{\varphi}$. It then updates the best lower bound $L$ and upper bound $U$ seen so far.
If $(1-\epsilon)\cdot U - (1+\epsilon) \cdot L \leq 0$,
then it returns the interval 
$[(1-\epsilon)\cdot U, (1+\epsilon)\cdot L]$.
For any value  $B$ in this non-empty interval, it holds 
$B \geq (1-\epsilon)\cdot U \geq (1-\epsilon)\cdot \banz(\varphi, x)$
and 
$B \leq (1+\epsilon)\cdot L \leq (1+\epsilon)\cdot \banz(\varphi, x)$,
i.e., $B$ is a relative $\epsilon$-approximation 
for $\banz(\varphi, x)$.
If the condition does not hold,  it picks a non-trivial (no literal/constant) leaf $\psi$, decomposes it,
and checks again whether the new bounds 
are satisfactory. Such a leaf $\psi$ always exists 
unless $T_{\varphi}$ is complete, in which case 
$U = L$. The decomposition of $\psi$ replaces $\psi$ 
by $\psi_1 \text{ op }  \psi_2$ where \text{ op }
represents independent-and ($\odot$), independent-or ($\otimes$), or mutual exclusion ($\oplus$).
The decomposition of $\psi$ into mutually exclusive functions $\psi_1$ and $\psi_2$ is always possible using  Shannon expansion. 

\begin{proposition}
\label{prop:adaban_correct}
For any positive DNF function $\varphi$,
d-tree $T_{\varphi}$ for $\varphi$,  variable $x$ in $\varphi$, 
error $\epsilon$, 
and bounds $L \leq \banz(\varphi,x) \leq U$,
it holds 
$\textsc{AdaBan}(T_{\varphi}, x, \epsilon, [L,U]) = [\ell, u]$ such that every value in $[\ell, u]$ is an $\epsilon$-approximation 
of $\banz(\varphi, x)$.
\end{proposition}

\subsubsection{Optimizations}
\label{sec:optimizatons}
The algorithms $\textsc{\adaban}$ and $\textsc{bounds}$
presented in Figs.~\ref{alg:adaban-bounds} and ~\ref{alg:adaban} are subject to four key optimizations implemented in our prototype.

(1)~Instead of {\em eagerly} recomputing the bounds for a partial d-tree after each decomposition step, we follow a {\em lazy} approach that does not recompute the bounds after independence partitioning steps and instead only recomputes them after Shannon expansion steps.

(2)~To avoid recomputation of bounds for subtrees whose 
leaves have not changed, we cache the bounds for each subtree. 
Hence, whenever a new bound is calculated for some leaf, it suffices
to propagate the bound along the path to the root of the d-tree. 

(3)~To approximate the Banzhaf values for several variables, we do not compute bounds for each variable after each expansion step. 
Instead, we compute the approximation for one variable at a time. 
After having achieved a satisfying approximation for one variable,
we reuse the partial d-tree constructed so far to obtain a
desired approximation for the next variable. This reduces the number of \textsc{bounds} calls and improves the overall runtime of \adaban.

(4)~Instead of computing bounds for $\#\varphi[x:=1]$ and $\#\varphi[x:=0]$, as done in \textsf{bounds}, it suffices to compute bounds for $\#\varphi$ and $\#\varphi[x:=0]$ for each variable $x$.
This is justified by the following insight:
\begin{align*}
&\banz(\varphi, x) =\ \#\varphi[x:=1] - \#\varphi[x:=0] \\
 &= \ \#\varphi[x:=1] + \#\varphi[x:=0] - 2\cdot \#\varphi[x:=0] 
=\ \#\varphi - 2\cdot \#\varphi[x:=0], 
\end{align*}
where the first equality is by the characterization of the Banzhaf value in Eq.~\eqref{eq:alternative_banzhaf} and the last equality 
states that the set of models of $\varphi$ is the disjoint union of the set of models where $x$ is $0$ and the set of models where $x$ is set to $1$.
In many practical scenarios, the lower bound for $\banz(\varphi, x)$ computed using bounds for $\#\varphi$ and $\#\varphi[x:=0]$ is tighter than the lower bound computed by $\textsc{\adaban}$.

\section{Banzhaf-based Ranking and top-$k$}
\label{sec:ranking_top_k}

Common uses of fact attribution in query answering and explanations are to identify the $k$ most influential facts and to rank the facts by their influence to the query result. Our anytime approximation of Banzhaf values lends itself naturally to fast ranking and computation of top-$k$ facts, as follows.

\subsection{The Algorithm \textsc{IchiBan}}

We introduce a new algorithm called \textsc{IchiBan}, that uses \textsc{AdaBan} to find the variables in a given function with the top-$k$ Banzhaf values. It starts by running \adaban\ for all variables at the same time. Whenever \adaban\ computes the bounds for the Banzhaf values of the variables, \textsc{IchiBan} identifies those variables whose upper bounds are smaller than the lower bounds of at least $k$ other variables. These former variables are not in top-$k$ and are discarded. It then resumes \adaban\ for the remaining variables and repeats the selection process using the refined bounds. Eventually, it obtains the variables with the top-$k$ Banzhaf values. For ranking, \textsc{IchiBan} runs until the approximation intervals for the variables do not overlap or collapse to the same Banzhaf value. 

\textsc{IchiBan} may also be executed with a parameter $\epsilon\in[0,1]$. In this case, it may finish as soon as 
each approximation interval reaches a relative error $\epsilon$.
\textsc{IchiBan} then ranks the facts based on the order of the mid-points of their respective intervals.

\subsection{A Dichotomy Result} 
The time complexity of \textsc{IchiBan} is exponential in the worst case. We next analyze in further depth the complexity of the ranking problem and show a dichotomy in the complexity of Banzhaf-based ranking of database facts. We first formalize the following ranking problem, parameterized by a Boolean 
CQ $Q$:

\begin{center}
\fbox{%
    \parbox{0.6\linewidth}{%
    \begin{tabular}{ll}
   Problem: & $\compbanz_Q$ \\
   Description: &  \textit{Banzhaf-based ranking of database facts} \\
   Parameter: &  Boolean CQ $Q$ \\    
        Input: & Database $D=(D_n,D_x)$ and facts $f_1, f_2 \in D_n$ \\
        Question: & Is $\banz(Q,D,f_1) \leq \banz(Q,D,f_2)$?
    \end{tabular}
    }}%
\end{center}

We now state the dichotomy and then explain it.
\begin{theorem}
\label{theo:dichotomy}
For any Boolean CQ $Q$ without self-joins, it holds:

\begin{itemize}
\item If $Q$ is hierarchical, then $\compbanz_Q$ can be solved in polynomial time.
\item If $Q$ is not hierarchical, then $\compbanz_Q$ cannot be solved in polynomial time, 
unless there is an $\fptas$ for $\countbis$.
\end{itemize}
\end{theorem}

The tractability part of our dichotomy follows from prior 
work: In case of hierarchical queries, 
{\em exact} Banzhaf values of database facts  can be computed in 
polynomial time~\cite{LivshitsBKS:LMCS:2021}.
Hence,  we can first compute
the exact Banzhaf values and then rank the facts.
Showing the intractability part of our dichotomy is more involved and requires novel development. 
It is based on the widely accepted conjecture that 
there is no polynomial-time approximation scheme (\fptas) for counting independent 
sets in bipartite graphs ($\countbis$)~\cite{DyerGGJ03,CurticapeanDFGL19}. In the following, we make these notions more precise.

A bipartite graph is an undirected graph $G = (V, E)$
where the set $V$ of nodes  is partitioned into 
two disjoint sets $U$ and $W$ and the edges $E \subseteq U \times W$
connect nodes from $U$ with nodes from $W$.   
An independent set $V'$ of $G$ is a subset of $V$ 
such that no two nodes in $V'$ are connected by an edge.  
The problem $\countbis$ is defined as:

\begin{center}
\fbox{%
    \parbox{0.6\linewidth}{%
    \begin{tabular}{ll}
   Problem: & $\countbis$ \\
   Description: &  \textit{Counting independent sets in bipartite graphs} \\
    Input: & Bipartite graph $G$  \\
    Compute: & Number of independent sets of $G$
    \end{tabular}
    }}%
    \end{center}

An algorithm $A$ for a numeric function $g$ is a
{\em fully polynomial-time approximation scheme} (\fptas) for $g$ if 
for any error $0<\epsilon < 1$ and input $x$, $A$ computes,  
in time polynomial in the size of $x$ and in $\epsilon^{-1}$, a value $A(x)$ such that $(1 -\epsilon)g(x) \leq A(x) \leq (1 + \epsilon) g(x)$.

The hardness result in Theorem~\ref{theo:dichotomy} assumes the widely accepted conjecture that there is 
no \fptas for $\countbis$~\cite{DyerGGJ03,CurticapeanDFGL19}. We next outline our proof strategy, which is visualized by the following diagram; the proof details are deferred to Appendix~\ref{app:ranking}. 

\begin{center}
\begin{tikzpicture}

\node (A) at (0, 0) {Polynomial-time algorithm for $\compbanz_Q$};
\node (A) at (0, -0.4) {for any non-hierarchical query $Q$};

\draw[rounded corners=0.1cm] (-3.5,-0.7) rectangle (3.5,0.4);

\node (B) at (-0.8, -1.4) {no \fptas}; 
\node (B) at (1, -1.4) {$\countbis$};

\draw (A) (1, -1.4) ellipse (0.7 and 0.3);

\node (B) at (-0.8, -2.4) {no \fptas}; 
\node (B) at (1, -2.4) {$\countnsat$};
\draw (A) (1, -2.4) ellipse (0.7 and 0.3);

\node (D) at (0, -4) {\fptas for $\countnsat$};
\draw[rounded corners=0.1cm] (-3.5,-4.3) rectangle (3.5,-3.7);

\draw[<-](-3.6,-3.8) to[out=120,in=240] (-3.6,-0.4);
\node (call) at (-4.8,-2.2) {\textit{implies}};

\draw[<-](-1.8,-2.4) to[out=120,in=240] (-1.8,-1.5);
\node (imply) at (-2.6,-1.9) {\textit{implies}};

\draw[<-](1.8,-2.4) to[out=40,in=320] (1.8,-1.5);
\node (imply) at (3.1,-1.5) {\textit{polynomial}};
\node (imply) at (3.1,-1.9) {\textit{parsimonious}};
\node (imply) at (3.1,-2.2) {\textit{reduction}};

\draw[red, decorate, decoration=snake](-0.8,-3.6) to (-0.8,-2.6);
\node (contradiction) at (0.1, -3.1) {\textit{\color{red}contradicts}};

\end{tikzpicture}
\end{center}

We use the intermediate problem $\countnsat$:  
Given a positive bipartite DNF function, compute the number of its non-satisfying assignments. We first give a parsimonious polynomial-time reduction from $\countbis$ to $\countnsat$, i.e., a polynomial-time reduction that also preserves the output; this means that the number of non-satisfying assignments equals the number of independent sets.
Assuming that there is no \fptas for \countbis, 
this reduction implies that there is no \fptas for \countnsat. 
Yet, given a polynomial-time algorithm $A$ for $\compbanz_Q$
for any non-hierarchical query $Q$, we can design an $\fptas$
for $\countnsat$.  
This contradicts the assumption that there is no \fptas
for $\countnsat$. Consequently, there cannot be any polynomial-time 
algorithm for $\compbanz_Q$ for non-hierarchical queries $Q$. 

\section{Experiments}
\label{sec:experiments}

This section details our experimental setup and results. 

\subsection{Experimental Setup and Benchmarks}

We implemented all algorithms in Python 3.9 and performed experiments on a Linux Debian 14.04 machine with 1TB of RAM and an Intel(R) Xeon(R) Gold 6252 CPU @ 2.10GHz processor. We set a timeout for each run of an algorithm to one hour.

\paragraph*{Algorithms} We benchmarked our algorithms \textsc{ExaBan}, \textsc{AdaBan}, and \textsc{IchiBan} against the following three competitors: \textsc{Sig22}, for exact computation using an off-the-shelf knowledge compilation package~\cite{DeutchFKM:SIGMOD:2022}; \textsc{MC}, a Monte Carlo-based randomized approximation~\cite{DBLP:conf/icdt/LivshitsBKS20}; and \textsc{CNFProxy}, an heuristic for ranking facts based on their contribution~\cite{DeutchFKM:SIGMOD:2022}.  
These competitors were originally developed for Shapley value. We adapted them to compute Banzhaf values (see Sec.~\ref{sec:related}). \textsc{AdaBan}, \textsc{MC}, and \textsc{IchiBan} expect as input: the error bound, the number of samples, and respectively the number of top results to retrieve. We use the notation \textsc{Algo{\bf X}} to denote the execution of an algorithm \textsc{Algo} with parameter value ${\bf X}$.

\paragraph*{Datasets}  We tested the algorithms using 301 queries evaluated over three datasets: Academic, IMDB and TPC-H (SF1). The workload is based on previous work on  Shapley values for query answering ~\cite{DeutchFKM:SIGMOD:2022,arad2022learnshapley}: as in ~\cite{DeutchFKM:SIGMOD:2022}, for TPC-H we used all queries without nested subqueries and with aggregates removed, so expressible as SPJU queries. For IMDB and Academic, we used all queries from \cite{arad2022learnshapley} (Academic was not used in \cite{DeutchFKM:SIGMOD:2022}). We constructed lineage for all output tuples of these queries using ProvSQL~\cite{senellart2018provsql}. The resulting set of nearly 1M lineage expressions is the most extensive collection for which attribution in query answering has been assessed in academic papers. Table \ref{tab:stat} includes statistics on the datasets.

\paragraph*{Measurements} We measure the execution time of all algorithms and the accuracy of \textsc{AdaBan} and \textsc{MC}. 
We define an instance as the (exact, approximate or top-$k$) computation of the Banzhaf values for all variables in a lineage of an output tuple of a query over one dataset. We report failure in case an algorithm did not terminate an instance within one hour. We also report the success rate of each algorithm and statistics of its execution times across all instances (average, median, maximal execution time, and percentiles). The  p$X$ columns in the following tables show the execution times for the $X$-th percentile of the considered instances.

\nop{We now turn to the experimental results, starting with the problem of exact computation.}

\begin{table}
\begin{center}
\begin{tabular}{| c | c | c |c| c|}
    \hline
        {\textbf{Dataset}} &  {\makecell{\textbf{\# Queries}}} &
         \textbf{\# Lineages}  & \textbf{\# Vars (avg/max)} & \textbf{\# Clauses (avg/max)} \\
         
         \hline
    \texttt{Academic} & 92 & 7,865 & 79 / 6,027 & 74 / 6,025\\
    \hline
    \texttt{IMDB} & 197 & 986,030 & 25 / 27,993 & 15 / 13,800 \\
    \hline
    \texttt{TPC-H} & 12 & 165 & 1,918 / 139,095 & 863 / 75,983\\
    \hline
    
\end{tabular}
\end{center}
\caption{Statistics of the datasets used in the experiments.}
\label{tab:stat}
\end{table}

\subsection{Exact Banzhaf computation}
\label{sec:exp-exact}

We first compare the two exact algorithms: $\textsc{ExaBan}$ and $\textsc{\textsc{Sig22}}$.

\paragraph*{Success Rate} 
Table \ref{tab:successrate} gives the success rate of $\textsc{ExaBan}$ and $\textsc{Sig22}$ for each dataset. 
$\textsc{ExaBan}$ succeeded for far more queries and lineages than $\textsc{Sig22}$. For Academic and IMDB, both algorithms succeeded for the majority of instances; a breakdown based on queries shows that whenever $\textsc{Sig22}$ failed for a query, it actually failed for all lineages (output tuples) of this query. $\textsc{ExaBan}$ succeeds for 15\% and 17\% more queries for Academic and respectively IMDB.
For TPC-H, the query success rate is significantly lower for both algorithms, even though $\textsc{ExaBan}$ failed for only 9\% of the  queries ($\textsc{Sig22}$ failed for 14\%).

\begin{table}
\begin{center}
\begin{tabular}{| c | c| c|c|}
\hline
\textbf{Dataset} & \textbf{Algorithm} & \textbf{Query Success Rate}
& \textbf{Lineage Success Rate} \\
\hline
{\multirow{4}{*}{{\texttt{Academic}}}}  & $\textsc{ExaBan}$  & 98.91\% & 99.99\%\\
 & $\textsc{Sig22}$ & 83.91\% & 98.40\% \\
\cdashline{2-4}
& $\textsc{AdaBan0.1}$ & 98.91\% & 99.99\%  \\
 & $\textsc{MC50\#vars}$ & 96.74\% & 98.83\%  \\
\hline
{\multirow{4}{*}{{\texttt{IMDB}}}} & $\textsc{ExaBan}$ & 82.23\% & 99.63\%  \\
 & $\textsc{Sig22}$ & 65.48\% & 98.35\%  \\
\cdashline{2-4}
 & $\textsc{AdaBan0.1}$ & 88.32\% & 99.81\% \\
 & $\textsc{MC50\#vars}$ & 83.76\% & 99.74\% \\

\hline
{\multirow{4}{*}{{\texttt{TPC-H}}}} & $\textsc{ExaBan}$ & 58.33\% & 91.52\%  \\
& $\textsc{Sig22}$ & 50.00\% & 85.46\%  \\
\cdashline{2-4}
& $\textsc{AdaBan0.1}$ & 75.00\% & 92.73\% \\
& $\textsc{MC50\#vars}$ & 50.00\% & 85.46\% \\

\hline

\end{tabular}
\end{center}
\caption{Query success rate: Percentage of queries for which the algorithms finish for all instances of a query within one hour. Lineage success rate:  Percentage of instances (over all queries in each dataset) for which the algorithms finish within one hour.}
\label{tab:successrate}
\end{table}

\paragraph*{Runtime Performance.} 
We first analyze the instances for which both algorithms succeed.
There are also instances for which $\textsc{Sig22}$ fails and $\textsc{ExaBan}$ succeeds. There are however no instances for which $\textsc{Sig22}$ succeeds and $\textsc{ExaBan}$  fails.
Table~\ref{tab:success_exact} shows that $\textsc{ExaBan}$ clearly outperforms $\textsc{\textsc{Sig22}}$: Whenever both succeed for Academic and TPC-H, they are very fast, bar a few outliers for $\textsc{Sig22}$. $\textsc{ExaBan}$ needs less than 0.4 and respectively 0.95 seconds for each instance. For instances that are hard for $\textsc{Sig22}$, $\textsc{ExaBan}$ achieves a speedup of up to 166x (229x) for TPC-H (Academic). For IMDB, $\textsc{ExaBan}$'s speedup over $\textsc{Sig22}$ is already visible for simple instances, with a speedup of 25x for the 95-th  percentiles. $\textsc{ExaBan}$ also has a few performance outliers for IMDB.

\paragraph*{Runtime Performance  of $\textsc{ExaBan}$ when $\textsc{Sig22}$ fails.}
$\textsc{Sig22}$ fails for 126 instances in Academic, 16239 instances in IMDB, and 24 instances in TPC-H. Table~\ref{tab:failed_exact} summarizes the success rate and runtime performance of $\textsc{ExaBan}$ for these instances. For Academic, $\textsc{ExaBan}$ achieves near-perfect success and finishes in less than ten minutes for all these instances. For IMDB, $\textsc{ExaBan}$ succeeds in $77.4\%$ of these instances. For 95\% of these success cases, $\textsc{ExaBan}$ finishes in under ten minutes. For TPC-H, $\textsc{ExaBan}$ succeeds in 41.7\% of these instances; whenever it succeeds, its computation time is just over one minute.   
To summarize, the algorithm $\textsc{ExaBan}$ is generally faster and more robust than $\textsc{\textsc{Sig22}}$. One reason is that, in contrast to $\textsc{ExaBan}$,  $\textsc{\textsc{Sig22}}$ requires to turn the lineage into a CNF representation, which may increase its size and complexity.

\begin{table}
    \centering
    
    \begin{tabular}{| c | c | c c c c c c c|}
        \hline
         \multirow{2}{*}{\textbf{Dataset}} & \multirow{2}{*}{\textbf{Algorithm}} &
         \multicolumn{7}{c |}{\textbf{Execution times [sec]}} \\
         & &\textbf{Mean} &  \textbf{p50} & \textbf{p75} & \textbf{p90} & \textbf{p95} & \textbf{p99} & \textbf{Max} \\
         \hline
         {\multirow{2}{*}{{\texttt{Academic}}}} 

        & $\textsc{ExaBan}$ & 0.004 & 0.001 & 0.002 &
        0.003 & 0.004 & 0.080 & 
        0.356\\
        
        & \textsc{Sig22} & 0.290 & 0.124 & 0.134 &
        0.303 & 0.537 & 2.433 & 
        81.54\\
        
        \hline
        
        {\multirow{2}{*}{{\texttt{IMDB}}}} & $\textsc{ExaBan}$ & 0.323 & 0.002 & 0.008 &
        0.066 & 0.231 & 2.174 & 
        1793\\  
        
        & \textsc{Sig22} & 2.840 & 0.146 & 0.365 &
        1.710 & 5.909 & 54.63 & 
        2271\\

        \hline
        
        {\multirow{2}{*}{{\texttt{TPC-H}}}} & $\textsc{ExaBan}$ & 0.713 & 0.892 & 0.905 &
        0.935 & 0.935 & 0.941 & 
        0.941\\  
        
         & \textsc{Sig22} & 1.217 & 0.080 & 0.140 & 0.200 & 0.260 &  1.450 & 157.3\\
        \hline
    \end{tabular}%
    \caption{Runtime performance for exact Banzhaf computation in instances for which \textsc{Sig22} succeeds. \nop{p$X$ stands for the execution time for the $X$-th percentile instances from the set of successful instances.}}
    \label{tab:success_exact}
\end{table}

\begin{table}
    \centering
    \begin{tabular}{| c | c| c c c c c c c|}
        \hline
         \multirow{2}{*}{\textbf{Dataset}} & \multirow{2}{*} {\makecell{\textbf{Success} \\ \textbf{rate}}} &
         \multicolumn{7}{c |}{\textbf{Execution times [sec]}} \\
         & &\textbf{Mean} &  \textbf{p50} & \textbf{p75} & \textbf{p90} & \textbf{p95} & \textbf{p99} & \textbf{Max} \\
         \hline

        \revision{\texttt{Academic}} & \revision{99.2\%} & \revision{128.9} & \revision{168.4} &   
        \revision{172.0} & \revision{174.4} &
        \revision{175.0} & \revision{189.0} & \revision{563.5}\\

        \revision{\texttt{IMDB}} & \revision{77.4\%} & \revision{111.9} & \revision{24.10} & \revision{95.95} &
        \revision{348.8} & \revision{597.1} & \revision{1055} & 
        \revision{1381}\\

        \revision{\texttt{TPC-H}} & \revision{41.7\%} & \revision{53.77} & \revision{56.44} & \revision{60.24} &
        \revision{63.27} & \revision{66.23} & \revision{68.59} & 
        \revision{69.18}\\  
        
        \hline
    \end{tabular}%
    \caption{\textsc{ExaBan}'s runtime performance for instances on which \textsc{Sig22} fails. \nop{p$X$ is the execution time for the $X$-th percentile instances from the set of instances on which \textsc{Sig22} failed.}}
    \label{tab:failed_exact}   
\end{table}

 \begin{figure}[]
     \centering
     \begin{subfigure}{\linewidth}
     \includegraphics[width=0.5\linewidth]{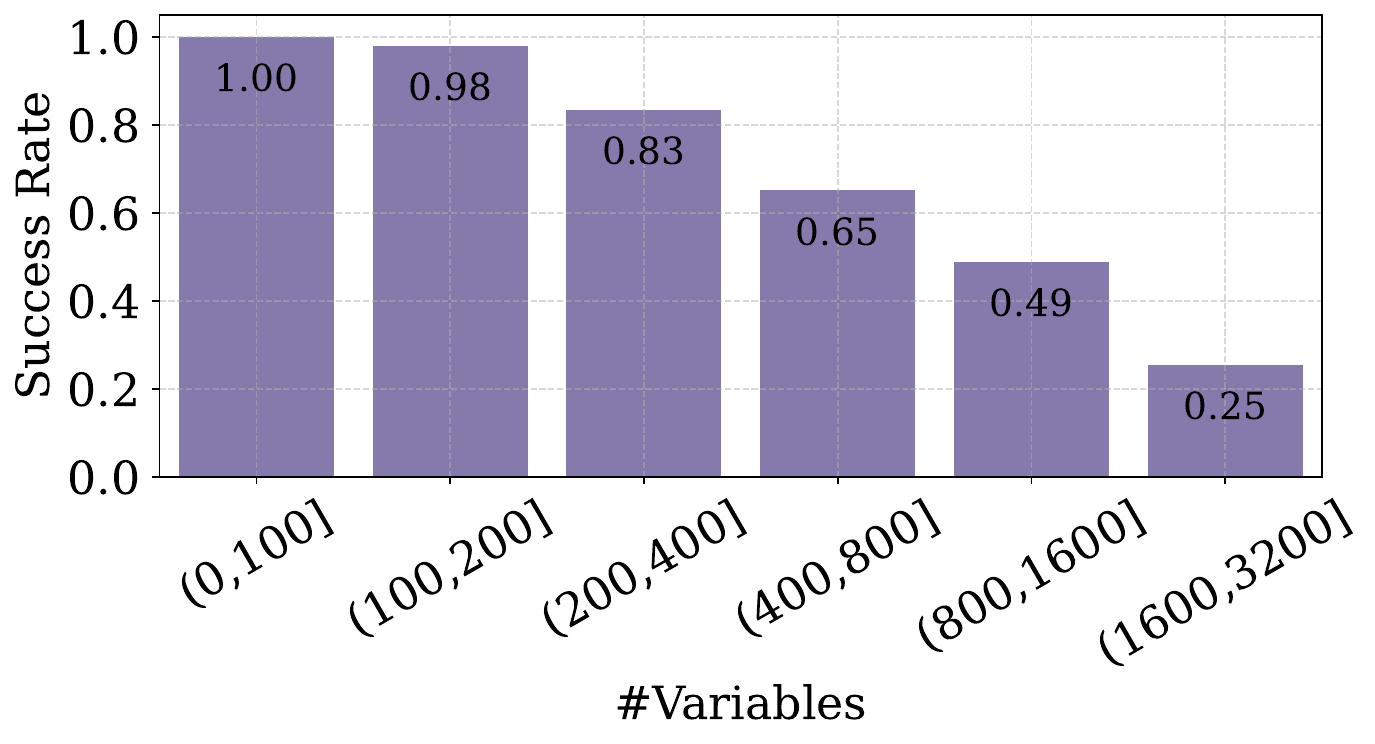}
     \includegraphics[width=0.5\linewidth]{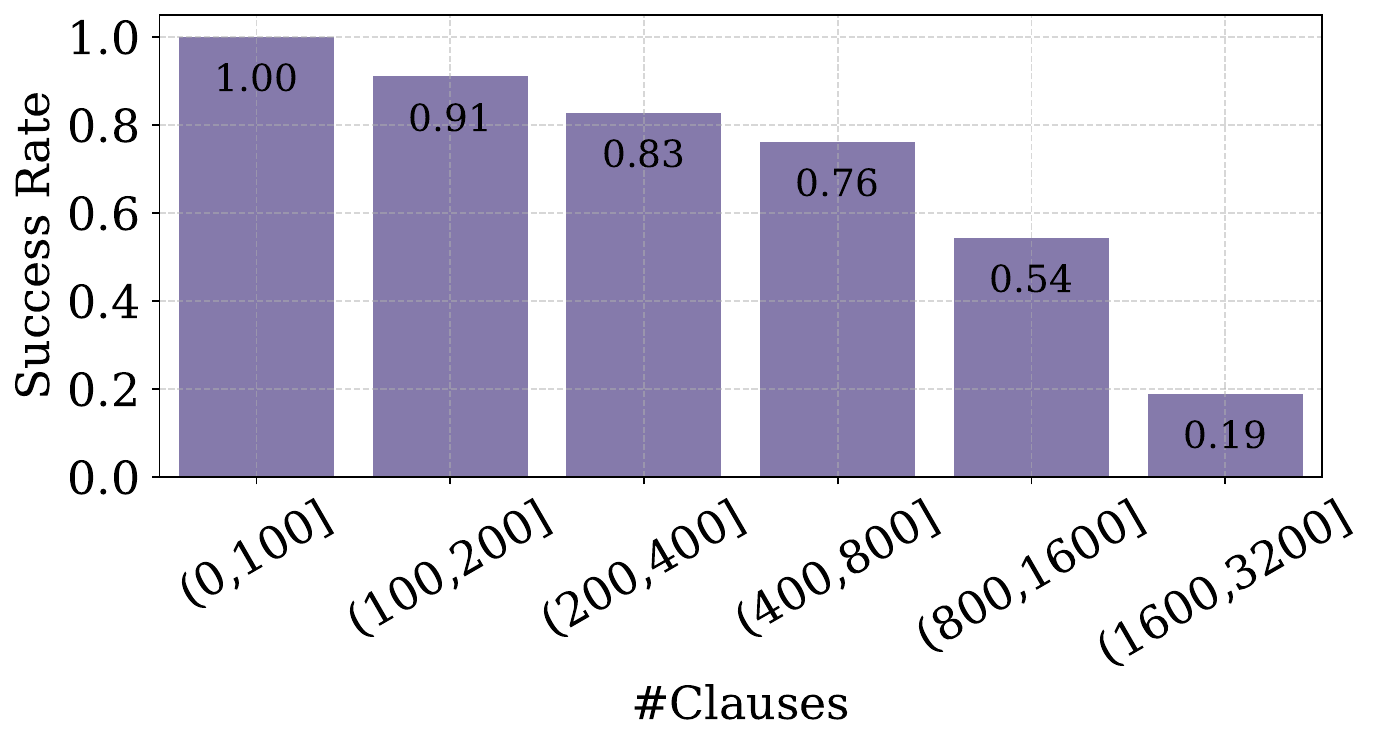}
    
     \caption{\normalsize Success rate (average over all instances in each group)}
     \label{fig:success_rate_exact_tuples}
     \end{subfigure}
     \hfill
     \vspace{1cm}
    
     \begin{subfigure}{\linewidth}
         \includegraphics[width=0.5\linewidth]{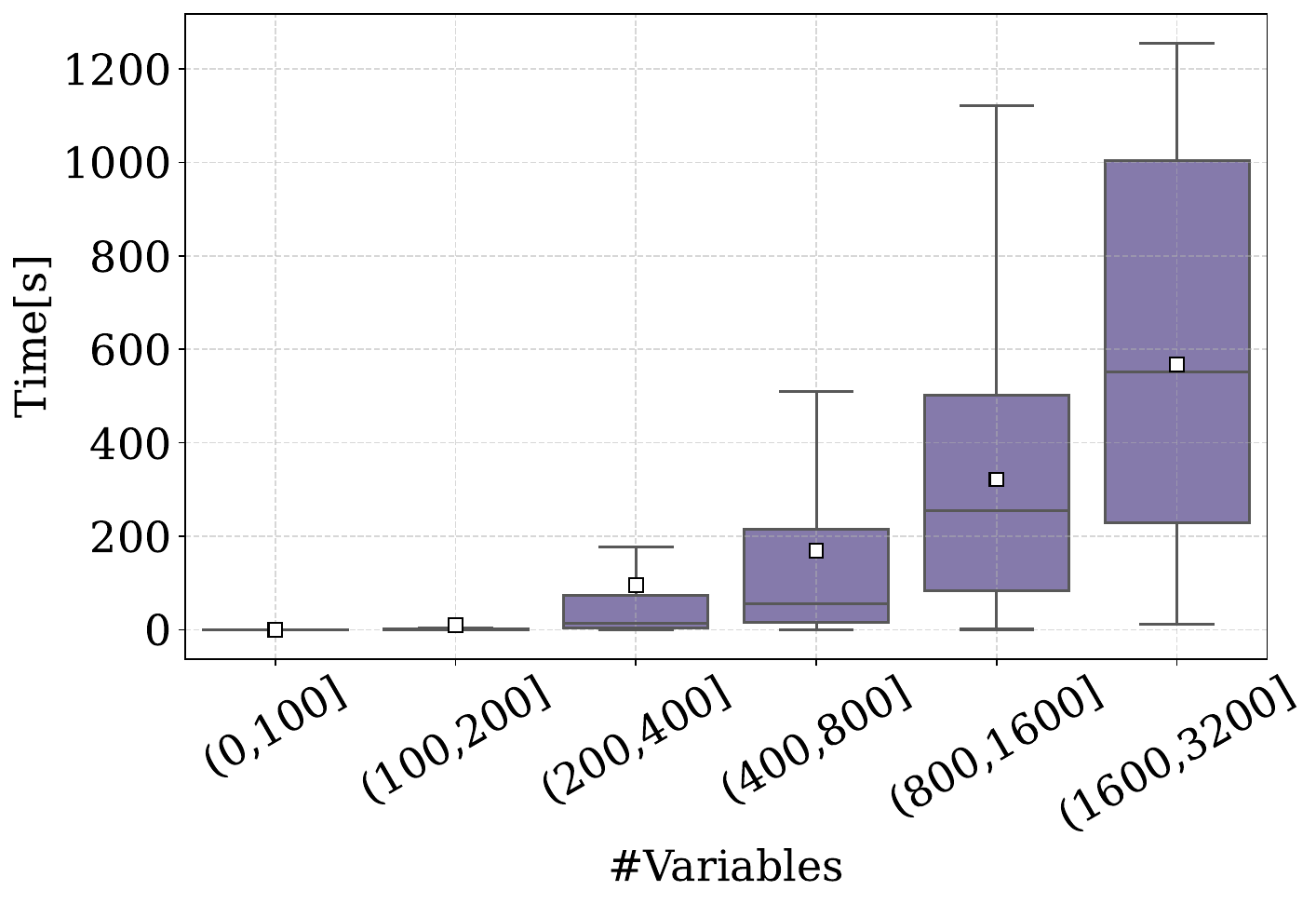}
         \includegraphics[width=0.5\linewidth]{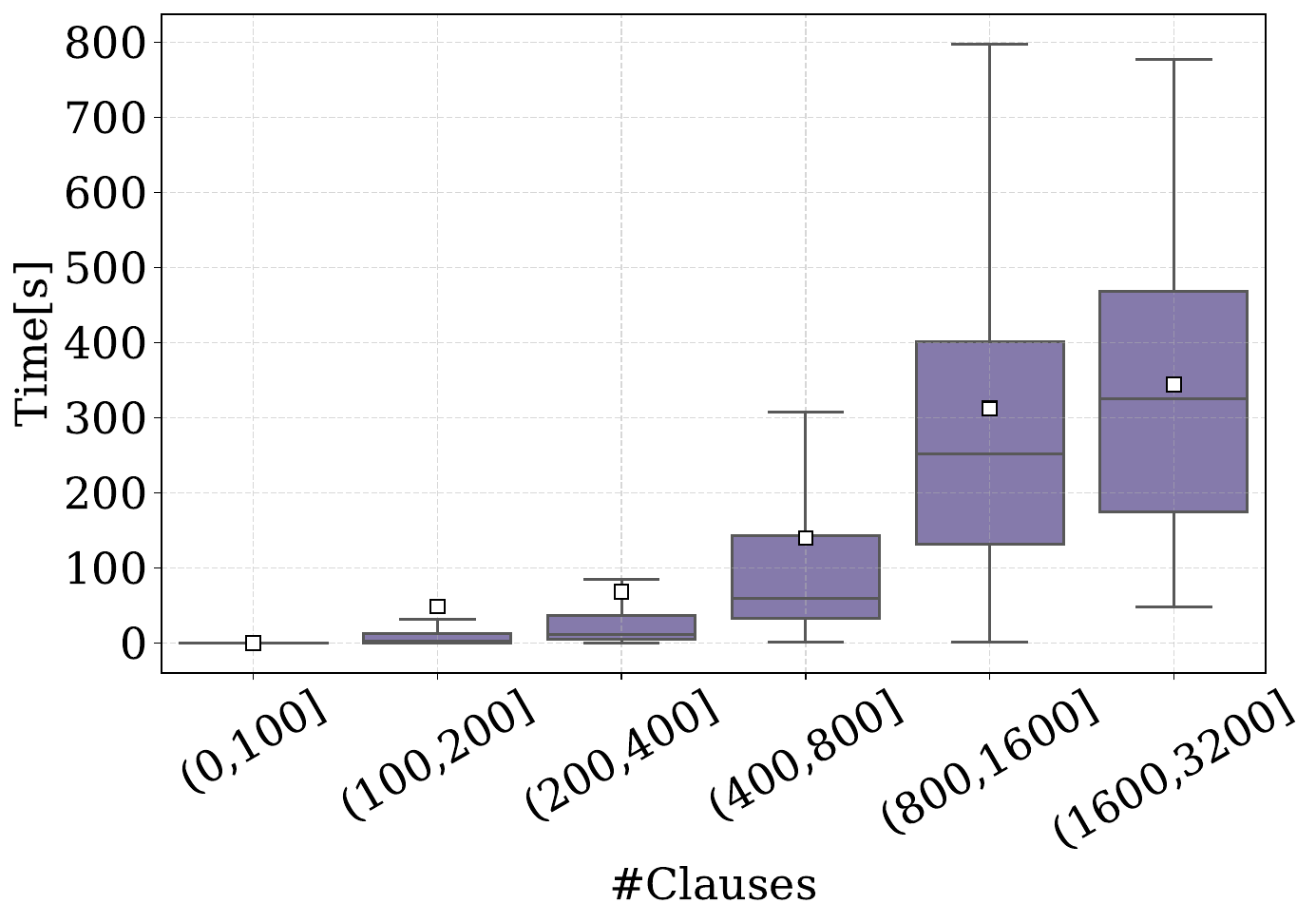}

     \caption{\revision{\normalsize  Execution time (ranges over all instances in each group)}}
     \label{fig:run_time_exact_tuples}
     \end{subfigure}
    
     \caption{\revision{Success rate and execution time of $\textsc{ExaBan}$ across all database and queries, grouped by the number of variables (clauses) in the lineage. An interval $[i,j]$ on the x-axis represents the set of lineages with \#vars (\# clauses) between $i$ and $j$.}}
     \label{fig:success_rate_exact_clauses_ratio}
 \end{figure}

\paragraph*{The effect of lineage size and structure} 
Figure~\ref{fig:success_rate_exact_clauses_ratio} gives a breakdown analysis of the performance of $\textsc{ExaBan}$, grouped by the number of variables or clauses. $\textsc{ExaBan}$ achieves near-perfect success rates and terminates in under a few seconds for instances with less than 200 variables or less than 100 clauses. $\textsc{ExaBan}$ is successful in $25 \%$ ($18\%$) of the instances with 1600-3200 variables (clauses).


\subsection{Approximate Banzhaf Computation}
\label{sec:exp-approximation}

We next examine the performance of $\textsc{AdaBan}0.1$ (i.e., $\textsc{AdaBan}$ with relative error 0.1) compared to $\textsc{ExaBan}$ and $\textsc{MC}$. 

\paragraph*{Success Rate} Table~\ref{tab:successrate} shows that $\textsc{AdaBan0.1}$'s success rate is higher than that of \textsc{ExaBan}. Indeed,  the former succeeds at least for all instances for which the latter also succeeds. For Academic, where the success rate of $\textsc{ExaBan}$ is already near perfect, there is no further improvement brought by $\textsc{AdaBan0.1}$. For IMDB and TPC-H, however, $\textsc{AdaBan0.1}$ succeeds for 88.32\% and respectively 75\% of queries, a significant increase relative to  $\textsc{ExaBan}$, which only succeeds for 82.23 \% and respectively 58.33 \% of queries.  In particular, we observe that \textsc{Adaban0.1} achieves a success rate of 74\% (68 \%) even for lineages with 1600-3200 variables (clauses), a significant improvement compared to the success rate of \textsc{ExaBan} for these cases. \textsc{MC50\#vars}'s success rate is comparable to that of \textsc{ExaBan} (but see the discussion below on execution time). 

\begin{table}
    \centering
    \begin{tabular}{| c | c | c c c c c c c|}
        \hline
         \multirow{2}{*}{\textbf{Dataset}} & \multirow{2}{*}{\textbf{Algorithm}} &
         \multicolumn{7}{c |}{\textbf{Execution times [sec]}} \\
         & &\textbf{Mean} &  \textbf{p50} & \textbf{p75} & \textbf{p90} & \textbf{p95} & \textbf{p99} & \textbf{Max} \\
         \hline
         {\multirow{3}{*}{{\texttt{Academic}}}} 
         & \textsc{AdaBan0.1} & 0.761 & 0.001 & 0.002 & 0.007 & 0.048 & 60.05 & 173.7\\
         
         & \textsc{ExaBan} & 2.065 & 0.001 & 0.002 & 0.012 & 0.197 & 164.5 & 563.5\\
        
         & \textsc{MC50\#vars} & >42.77 & 0.003 & 0.013 & 0.072 & 0.239 & >3600 &  >3600\\
        
         \hline

         {\multirow{2}{*}{{\texttt{IMDB}}}} &  $\textsc{AdaBan0.1}$ & 0.624 & 0.001 & 0.003 & 0.014 & 0.044 &  4.740 & 984.9 \\
         
          & $\textsc{ExaBan}$ & 1.579 &  0.002 & 0.003 & 0.009 &  0.077 &  10.374 & 1793\\ 
          
          & \textsc{MC50\#vars} &  >13.99 & 0.012 & 0.039 & 0.386 & 2.613 & 257.1 & >3600 \\
          
          \hline
        
        {\multirow{2}{*}{{\texttt{TPC-H}}}} & $\textsc{AdaBan0.1}$ &  0.198 & 0.003 &  0.005 & 0.013 & 2.590 &  3.421 & 3.460\\  
        
        & \textsc{ExaBan} & 4.227 & 0.895 & 0.931 & 0.938 & 51.05 & 61.98 & 
        69.18 \\
        
        & \textsc{MC50\#vars} & >260.7 & 0.003 & 0.009 & 0.066 & >3600 & >3600 & >3600\\
        
        \hline
    \end{tabular}%
    \caption{Approximate versus exact Banzhaf computation for instances on which \textsc{ExaBan} succeeds.}
    \label{tab:approx-time-when-exact-succeeds}
\end{table}

\begin{table}
    \centering

    \begin{tabular}{| c | c| c c c c c c c|}
        \hline
         \multirow{2}{*}{\textbf{Dataset}} & \multirow{2}{*} {\makecell{\textbf{Success} \\ \textbf{rate}}} &
         \multicolumn{7}{c |}{\textbf{Execution times [sec]}} \\
         & &\textbf{Mean} &  \textbf{p50} & \textbf{p75} & \textbf{p90} & \textbf{p95} & \textbf{p99} & \textbf{Max} \\
         \hline

        \revision{\texttt{IMDB}} & \revision{49.53\%} & \revision{644.1} & \revision{575.3} & \revision{847.0} &
        \revision{1105} & \revision{1247} & \revision{1584} & 
        \revision{1802}\\

        \revision{\texttt{TPC-H}} & \revision{15.39\%} & \revision{166.3} & \revision{166.3} & \revision{166.4} &
        \revision{166.4} & \revision{166.4} & \revision{166.4} & 
        \revision{166.4}\\
        
        \hline
    \end{tabular}%
    
    \caption{\textsc{AdaBan0.1} runtime performance and success rate for instances on which \textsc{ExaBan} fails. \nop{Academic dataset is excluded since \textsc{ExaBan} only fails on one lineage for this dataset.} }
    \label{tab:approximate computationtimewhenexactfails}
\end{table}

\paragraph*{Runtime Performance} Table~\ref{tab:approx-time-when-exact-succeeds} focuses on the instances on which $\textsc{ExaBan}$ (and also $\textsc{AdaBan0.1}$) succeeds. $\textsc{AdaBan0.1}$ consistently outperforms both $\textsc{ExaBan}$ and $\textsc{MC50\#vars}$. The gains in the average runtime over $\textsc{ExaBan}$ range from a factor of 3 for Academic to 20 for TPC-H. We further observe that $\textsc{MC50\#vars}$ is slower than $\textsc{ExaBan}$ for over 99\% of the examined instances, and even fails for some of the instances for which $\textsc{ExaBan}$ succeeds. Running \textsc{MC} with a larger number of samples to improve its accuracy (see below) is only going to take more time.

\begin{table}
    \centering
      
          \resizebox{\textwidth}{!}{%
    \begin{tabular}{|c|c| c c c c c c c|}
        \hline
         \textbf{Dataset} & \textbf{Algorithm} & \textbf{Mean} & \textbf{p50} & \textbf{p75} & \textbf{p90} & \textbf{p95} & \textbf{p99} & \textbf{Max} \\
         \hline
         
      {\multirow{2}{*}{{\texttt{Academic}}}} & \textsc{AdaBan0.1} & 5.24E-05 & 0 & 0 & 0 & 0 & 1.18E-03 &  2.09E-02\\  
      & \textsc{MC50\#vars} & 0.60 &  0.56 & 0.78 & 1.00 & 1.30 &  1.34 & 1.67\\  

    \hline
    {\multirow{2}{*}{{\texttt{IMDB}}}} & \textsc{AdaBan0.1} & 1.35E-04 & 0 & 0 & 0 & 7.77E-04 & 3.34E-03 & 1.92E-02 \\  
    & \revision{\textsc{MC50\#vars}} & \revision{0.56} & \revision{0.51} & \revision{0.67} &
        \revision{0.87} & \revision{1.00} & \revision{1.20} & 
        \revision{1.71}\\  
    \hline    
    {\multirow{2}{*}{{\texttt{TPC-H}}}} & \revision{\textsc{AdaBan0.1}} & \revision{9.04E-18} & \revision{0} & \revision{0} &
        \revision{0} & \revision{1.24E-24} & \revision{3.23E-23} & 
        \revision{1.37E-15}\\  
        & \revision{\textsc{MC50\#vars}} & \revision{0.50} & \revision{0.44} & \revision{0.67} &
        \revision{1.00} & \revision{1.34} & \revision{1.34} & 
        \revision{1.34}\\  
        \hline

    {\multirow{2}{*}{{\texttt{Hard}}}} 
         & \revision{\textsc{AdaBan0.1}} & \revision{3.96E-04} & \revision{2.40E-05} & \revision{3.61E-04} &
    \revision{1.19E-03} & \revision{2.06E-03} & \revision{4.21E-03} & 
    \revision{1.65E-02}\\
    & \revision{\textsc{MC50\#vars}} & \revision{0.312} & \revision{0.303} & \revision{0.383} &
        \revision{0.4.65} & \revision{0.516} & \revision{0.64} & 
        \revision{1.13}\\  
        \hline
    \end{tabular}%
    }
     \caption{Observed error ratio as $\ell_1$ distance between the vectors of algorithm's output and of the exact normalized Banzhaf values for instances on which \textsc{ExaBan} succeeded. \nop{pX stands for the execution time for Xth percentile variable over all lineages in this set and all variables in each lineage.} } 
\label{tab:l1}
\end{table}

\paragraph*{Runtime Performance and Success Rate of $\textsc{AdaBan0.1}$ when other Algorithms fail}
Table~\ref{tab:approximate computationtimewhenexactfails} shows 
that, when only considering the instances on which $\textsc{ExaBan}$ fails, $\textsc{AdaBan0.1}$ succeeds in nearly 50\% (15\%) of these instances for IMDB (TPC-H). Both $\textsc{ExaBan}$ and 
$\textsc{AdaBan0.1}$ fail for just one instance in Academic (not shown).

\paragraph*{Approximation Quality} $\textsc{AdaBan0.1}$ guarantees a deterministic relative error of $0.1$. $\textsc{MC50\#vars}$ only guarantees a probabilistic absolute error, where the number of required samples depends quadratically on the inverse of the error. Table~\ref{tab:l1} compares the observed approximation quality of $\textsc{AdaBan0.1}$ and $\textsc{MC50\#vars}$. These are measured as the $\ell_1$ distance between the vectors of estimated Banzhaf values computed by each algorithm, compared to the ground truth exact Banzhaf values as computed by $\textsc{ExaBan}$. The results are shown for all instances for which $\textsc{ExaBan}$ succeeds, and separately for the "Hard" instances for which $\textsc{ExaBan}$ took at least five seconds. For all these instances, $\textsc{AdaBan0.1}$'s approximation is consistently closer to the ground truth than $\textsc{MC50\#vars}$'s approximation by several orders of magnitude.

\paragraph*{Approximation Error as a Function of Time} 
Figure \ref{fig:errorevolution} presents, for several instances, the evolution of the observed error for $\textsc{AdaBan}$ and $\textsc{MC}$ over time. These instances appear in \cite{Article_git} and were selected, for illustration, from the set of ``hard" lineages for which \textsc{ExaBan} needs longer than 200 seconds to compute the Banzhaf values of all variables (and then individual variables appearing in these lineages were selected at random).
The error of $\textsc{AdaBan}$ shown in Figure \ref{fig:errorevolution} decreases exponentially and consistently over time, reaching a very small error within a few seconds. This is consistent with our observation that a small error ($\epsilon=0.1$) is typically reached very quickly.
In contrast, the behavior of \textsc{MC} is erratic and for some instances it may not even converge within two hundred seconds.

\begin{figure*}[t]{}
\centering
    \begin{subfigure}{0.48\linewidth}
    \includegraphics[width=\linewidth]{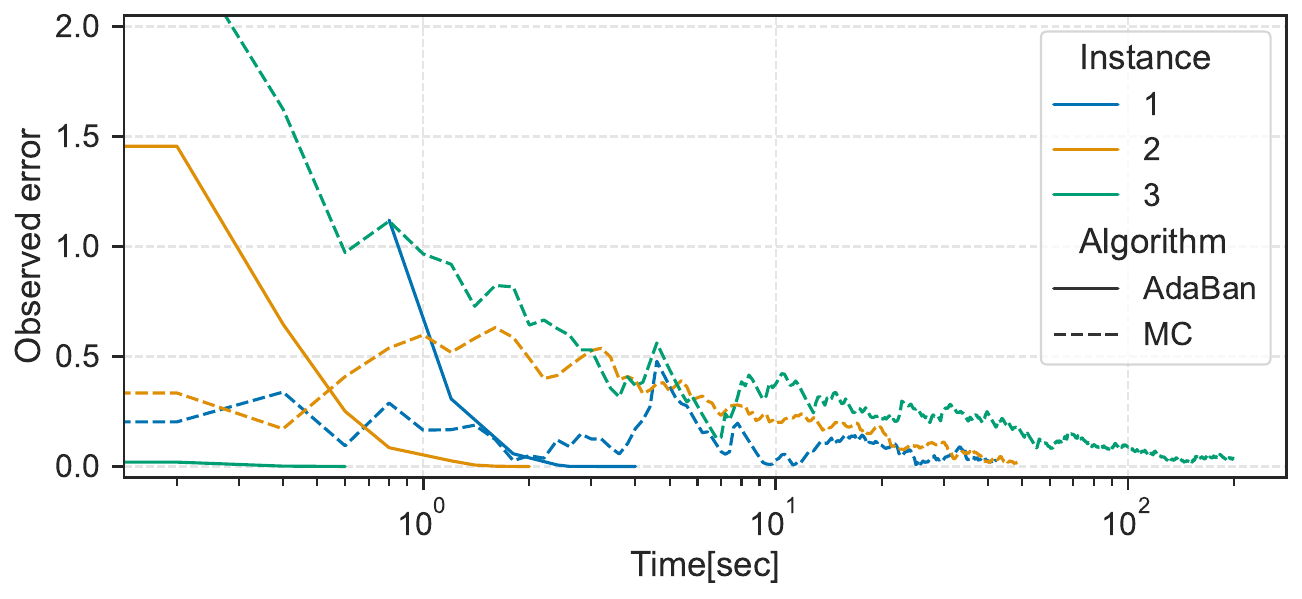}
    \caption{\normalsize  Three instances for which \textsc{MC} converged to the Banzhaf value.}
    \label{fig:evolution_simple}
    \end{subfigure}
    \hspace{0.4cm}
    \begin{subfigure}{0.48\linewidth}
    \includegraphics[width=\linewidth]{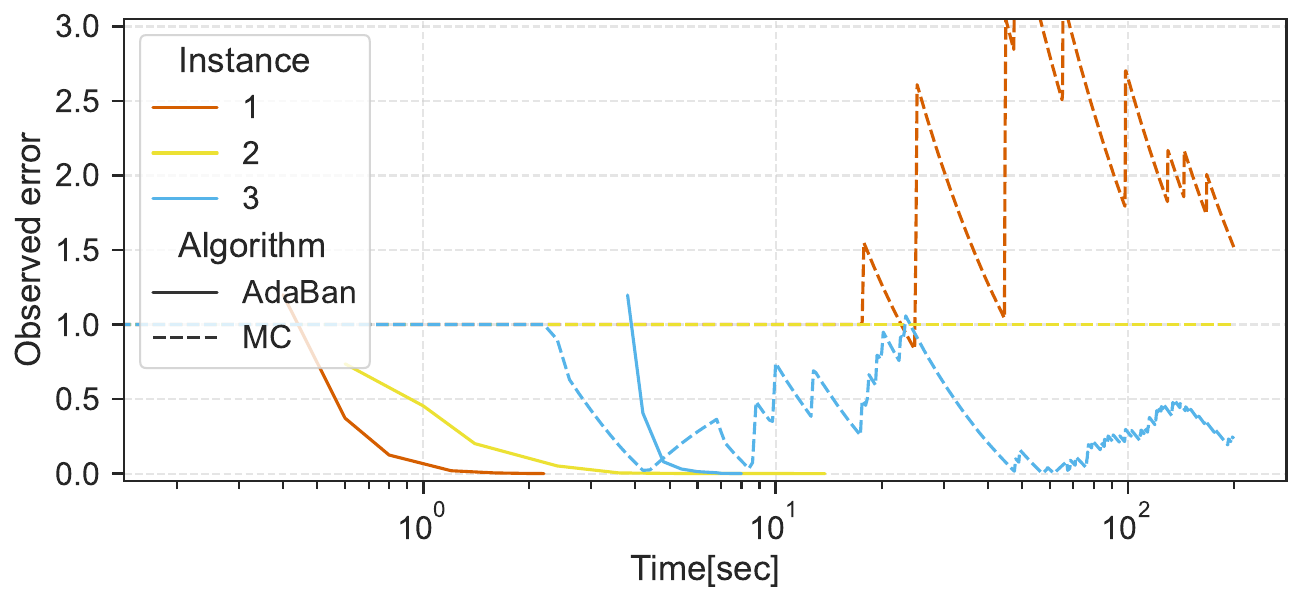}
    \caption{\normalsize Three instances for which \textsc{MC} did not converge to the Banzhaf value.}
    \label{fig:evolution_hard}
    \end{subfigure}
   \caption{Convergence rate of approximate Banzhaf value $\hat{v}$ to the  exact Banzhaf value $v$ as a function of time, for representative instances. The observed error on the y-axis is calculated as $\frac{|v - \hat{v}|}{v}$. \textsc{AdaBan} is stopped as soon as it reaches the exact Banzhaf value.}
   \label{fig:errorevolution}
\end{figure*}

\begin{table}
    \centering  
    \resizebox{\textwidth}{!}{%
    \begin{tabular}{|c|c| c c c c c c c|}
        \hline
         \textbf{Dataset} & \textbf{Algorithm} & \textbf{Mean} & \textbf{p50} & \textbf{p75} & \textbf{p90} & \textbf{p95} & \textbf{p99} & \textbf{Min} \\
         \hline

      {\multirow{3}{*}{{\texttt{Academic}}}} & \textsc{IchiBan0.1} & 1 / 1 & 1 / 1 & 1 / 1 & 1 / 1 &  1 / 1 & 1 / 1 & 0.9 / 1 \\ 
      
      & \textsc{MC50\#vars} & 0.87 / 0.90 & 0.9 / 1 & 0.8 / 0.8 & 0.7 / 0.6 & 0.5 / 0.6 & 0.3 / 0.4 & 0.2 / 0.2\\
      
      & \textsc{CNF Proxy} & 0.87 / 0.95 & 0.9 / 1 & 0.8 / 1 & 0.7 / 0.8 & 0.6 / 0.8 & 0.5 / 0.6 & 0.3 / 0.4 \\ 
      \hline
      
      {\multirow{3}{*}{{\texttt{IMDB}}}} & \textsc{IchiBan0.1} & 1 / 1 & 1 / 1 & 1 / 1 & 1 / 1 & 1 / 1 & 1 / 1 & 0.6 / 0.4\\

      & \textsc{MC50\#vars} &  0.90 / 0.87 & 0.9 / 1 & 0.8 / 0.8 & 0.7 / 0.6 & 0.6 / 0.6 & 0.5 / 0.4 & 0 / 0\\
      
      & \textsc{CNF Proxy} & 0.93 / 0.98 & 1 / 1 & 0.9 / 1 & 0.8 / 1 & 0.7 / 0.8 & 0.6 / 0.6 & 0.2 / 0.2\\ 
      \hline
      
      {\multirow{3}{*}{{\texttt{TPC-H}}}} 
      & \textsc{IchiBan0.1} & 1 / 1 & 1 / 1 & 1 / 1 & 1 / 1 & 1 / 1 & 1 / 1 &  1 / 1 \\
      
      & \textsc{MC50\#vars} & 0.34 / 0.84 & 0.1 / 1 & 0.1 / 1 & 0.1 / 0.2 & 0.1 / 0.2 & 0.1 / 0.11 & 0.1 / 0\\
      
      & \textsc{CNF Proxy} & 0.88 / 0.97 & 0.9 / 1 & 0.8 / 1 & 0.7 / 0.8 & 0.7 / 0.8 & 0.7 / 0.6 & 0.7 / 0.6 \\
      \hline
    \end{tabular}%
    }
    \caption{Observed precision@10 / precision@5 for instances for which \textsc{ExaBan} succeeds. \nop{pX stands for the observed value for Xth percentile variable over all lineages in this set.}}
    \label{tab:precisiontopk}  
\end{table}

\subsection{Top-$k$ Computation}
\label{sec:exp-topk}
We evaluate the accuracy of $\textsc{IchiBan0.1}$, which allows a relative error of up to 0.1, $\textsc{MC50\#vars}$, and \textsc{CNF Proxy}
using the standard measure of precision@k, which is the fraction of reported top-$k$ tuples that are in the ground truth top-$k$ set. 
Table~\ref{tab:precisiontopk} gives the distribution of precision@k values observed for different instances and  $k\in\{5,10\}$. 
With the exception of some outliers for IMDB, $\textsc{IchiBan0.1}$ achieves near perfect precision@k, while $\textsc{MC50\#vars}$ is much less stable and consistently inferior. \textsc{CNF Proxy} is more accurate than $\textsc{MC50\#vars}$, but is also consistently outperformed by $\textsc{IchiBan0.1}$. The results for $k=1,3$ are omitted for lack of space: for $k=1$, all algorithms achieve high success rates; for $k=3$ the observed trends are similar to those in the table.
The execution time of $\textsc{IchiBan0.1}$ is essentially the same as reported for $\textsc{AdaBan0.1}$, i.e. typically an order of magnitude better than \textsc{ExaBan}.

We further run the variant of $\textsc{IchiBan}$ that decides the top-$k$ results with certainty (deferred to Appendix~\ref{app:topk}):
for top-$1$, it is extremely fast; in many of the considered instances, there is a clear top-$1$ fact, whose
Banzhaf value is much greater than of the others. For top-$3$ and top-$5$, it achieved better performance over IMDB than both \textsc{ExaBan} and \textsc{AdaBan0.1}. This was however not the case for TPC-H, where separating the top-$3$ or top-$5$ facts from the rest took longer than \textsc{ExaBan}. We attribute this to a large number of ties in the Banzhaf values of facts for the TPC-H workload, whose lineages are more symmetric in the variables. \textsc{IchiBan0.1} is a good alternative for such instances.

\subsection{Summary of Experimental Findings}

Our experimental findings lead to the following main conclusions:

(1) {\em \textsc{ExaBan} consistently outperforms \textsc{Sig22} for exact computation.} 
Sec.~\ref{sec:exp-exact} shows that \textsc{ExaBan} not only outperforms \textsc{Sig22} on the  workloads previously used for \textsc{Sig22}, but it also
succeeds in many
cases where \textsc{Sig22} times out (41.7\%-99.2\% of these cases for the different datasets).  

(2) {\em \textsc{AdaBan} outperforms \textsc{ExaBan} already for small relative errors.} Sec.~\ref{sec:exp-approximation} shows that \textsc{AdaBan} is up to an order of magnitude, and on average three times faster than \textsc{ExaBan}  for relative error 0.1.

(3) {\em The accuracy of \textsc{MC} can be orders of magnitude worse than that of \textsc{AdaBan}.} Sec.~\ref{sec:exp-approximation} shows that if we only run \textsc{MC} for a sufficiently small number of steps so that its runtime remains competitive to \textsc{AdaBan}, its accuracy can be up to four orders of magnitude worse than that of \textsc{AdaBan}. On the other hand, if we were to run \textsc{MC} sufficiently many steps to achieve a comparable accuracy, then its runtime becomes infeasible.

(4) {\em \textsc{IchiBan} can quickly identify the top-$k$ facts.} Sec.~\ref{sec:exp-topk} shows that \textsc{IchiBan} quickly and accurately separates the approximation intervals of the first $k$ Banzhaf values (demonstrated for $k$ up to $10$) from the remaining values. It is significantly more accurate than previous approaches based on \textsc{MC} or \textsc{CNF Proxy}.

\section{Related Work}
\label{sec:related}

We compare our work to multiple lines of related work. 

\paragraph{Shapley value} Recent work~\cite{DBLP:conf/icdt/LivshitsBKS20,ReshefKL20,LivshitsBKS:LMCS:2021,LivshitsBKS:SIGMODREC:2021,DeutchFKM:SIGMOD:2022,DavidsonDFKKM:SIGMOD:2022,KhalilK23, bertossishapley} investigated the use of the Shapley value~\cite{shapley1953value} to define attribution scores in query answering, with particular focus on algorithms and the complexity of computing exact and approximate Shapley values for database facts.
The Banzhaf value~\cite{Penrose:Banzhaf:1946,Banzhaf:1965} is very closely related to the Shapley value, and both have been extensively investigated in Game Theory \cite{lehrer1988axiomatization,feltkamp1995alternative,van1998axiomatizations,strafiin1988shapley}. They have the same formula up to combinatorial coefficients that are present in the Shapley value formula and missing in the Banzhaf value formula; different coefficients need to be computed for each size of variable set, and are multiplied by the number of sets of this size.  Computationally, we have empirically shown advantages of the approach presented here over prior work. Furthermore, our algorithmic and theoretical contributions do not have a parallel in the literature on Shapley or Banzhaf values for query answering\nop{ (to our knowledge)}. Specifically, ours are the first deterministic approximation and ranking algorithms with provable guarantees, whereas approximation in previous works is based on Monte Carlo and with absolute error guarantees~\cite{DeutchFKM:SIGMOD:2022,LivshitsBKS:LMCS:2021}, while ranking is only heuristic and can be arbitrarily off the true ranking~\cite{DeutchFKM:SIGMOD:2022} (see also the discussion below on approximation algorithms). 

Banzhaf-based ranking and Shapley-based ranking can differ already for the simple query $Q(X) = R(X) \wedge S(X,Y) \wedge T(X,Z)$ (details in Appendix~\ref{app:shapley}).
 Our dichotomy result establishes that Banzhaf-based ranking is tractable precisely for the same class of hierarchical queries for which the exact computation of the Banzhaf (and even Shapley) value~\cite{LivshitsBKS:LMCS:2021} is also tractable. The hierarchical property led to further dichotomies, e.g., for probabilistic query evaluation~\cite{dalvi2007efficient}, incremental view maintenance~\cite{BerkholzKS:PODS:2017}, one-step streaming evaluation in the finite cursor model~\cite{GroheGLSTB:TCS:2009}, and readability of query lineage~\cite{OlteanuZ:ICDT:2012}.

\paragraph{Hardness of exact Banzhaf computation}
Prior work shows that for non-hierarchical self-join free CQs,
computing exact Banzhaf values of database facts is $\fpsharpp$-hard~\cite{LivshitsBKS:LMCS:2021}.
The proof is by a reduction from the $\fpsharpp$-hard problem of evaluating non-hierarchical 
queries over probabilistic databases~\cite{DalviS12}.
Our argument for the hardness of Banzhaf-based ranking 
is different. It relies on the conjecture that there is no polynomial-time 
approximation for counting the independent sets in a bipartite graph \cite{DyerGGJ03,CurticapeanDFGL19}.

\paragraph*{Further attribution measures in query answering}
Causality-based methods focus on uncovering the causal responsibility of database facts for a query outcome~\cite{MeliouGMS:Causality:2011, DBLP:journals/pvldb/MeliouRS14, DBLP:conf/tapp/SalimiBSB16}. 
The causal responsibility of a fact $f$ 
is a score proportional to the largest fact set such that including $f$ in the set turns the query answer from false to true.
Furthermore, recent work has empirically evaluated various attribution methods for the problem of credit distribution \cite{DossoDS:CreditDistribution:2022}. Their study compares game theory-based methods with approaches based on causal responsibility and simpler methods like fact frequency counting in the provenance. They highlight both the similarities and differences among these attribution approaches.

\paragraph*{Attribution in machine learning} 
The SHAP (SHapley Additive exPlanations) score attributes feature importance in machine learning models~\cite{lundberg2017unified}. It builds upon the Shapley value, but differs in that it models missing ``players'' (feature values in the context of machine learning) according to their expectation. A recent line of work studies the computational complexity of the SHAP score~\cite{ arenas2021tractability, DBLP:journals/jmlr/ArenasBBM23}: Under commonly accepted complexity assumptions, there is no polynomial-time algorithm for ranking based on SHAP scores, even for monotone DNF functions. This hardness result uses a different technique from our work. It is open whether Banzhaf-based ranking is computationally cheaper than SHAP-based ranking.

\paragraph*{Approximation algorithms} Our work relies on the anytime deterministic approximation framework originally introduced for (ranked) query evaluation in probabilistic databases~\cite{OlteanuHK:AdaProb:2010,OlteanuW:AdaProbRank:2012,FinkHO:AdaProb:2013}. In particular, it  uses an incremental shared compilation of query lineage into partial d-trees for approximate computation, ranking, and top-$k$. Besides the general approximation framework, our work  differs significantly from this prior work as it is tailored at Banzhaf value computation and Banzhaf-based ranking as opposed to probability computation. In particular, \adaban\ uses lower and upper bounds for the Banzhaf values 
in functions represented (1) in independent DNF and (2) by disjunctions and conjunctions of mutually exclusive or independent functions. These bounds need also be computed for each variable in the function rather than for the entire function.  

Prior work~\cite{LivshitsBKS:LMCS:2021} gives a polynomial time randomized absolute approximation scheme for Shapley (and Banzhaf) values based on Monte Carlo sampling. Sec.~\ref{sec:experiments} shows experimentally that  \adaban\  significantly outperforms this randomized approach. As also shown for ranking in probabilistic databases~\cite{OlteanuW:AdaProbRank:2012}, randomized approximations based on Monte Carlo sampling have three important limitations, which are not shared by our deterministic approximation \adaban: (1) the achieved ranking is only a probabilistic approximation of the correct one; (2) running one more Monte Carlo step does not necessarily lead to a refinement of the approximation interval, and hence the approximation is not truly incremental; (3) The sampling approach sees the input functions as black boxes and does not exploit their structure. Sec.~\ref{sec:experiments} also reports on experiments with the CNF Proxy heuristic~\cite{DeutchFKM:SIGMOD:2022}, which efficiently rank facts based on a proxy value; though it has no theoretical guarantees, the obtained ranking is often similar to the Shapley-based ranking, even though the proxy values are typically {\em not} similar to the Shapley values. Sec.~\ref{sec:experiments} shows that our algorithm also outperforms CNF Proxy in terms of accuracy.

\section{Conclusion}
\label{sec:conc}

In this paper we introduced effective algorithms for the exact and anytime deterministic approximate computation of the Banzhaf values that quantify the contribution of database facts to the answers of select-project-join-union queries. We also showed the use of these algorithms for Banzhaf-based ranking and gave a dichotomy in the complexity of ranking. We showed experimentally that our algorithms outperform prior work in both runtime and accuracy for a wide range of problem instances.

There are several exciting directions for future work. First, we would like to extend our algorithmic framework to more expressive queries that also have aggregates and negation. There is also a host of possible optimizations that can improve the scalability and efficiency of our algorithms. Finally, we would like to generalize our algorithms to further fact attribution measures, such as the Shapley value, the SHAP score, and the causality-based measures highlighted in Sec.~\ref{sec:related}.

\paragraph{Acknowledgements.}
Ahmet Kara and Dan Olteanu would like to acknowledge 
the UZH Global Strategy and Partnerships Funding Scheme. 
The research of Daniel Deutch and Omer Abramovich has been partially funded by the European Research Council (ERC) under the European Union’s Horizon
2020 research and innovation programme (Grant agreement No. 804302).

\bibliographystyle{plain}
\bibliography{bib}

\appendix
\section{Missing Details in Section \ref{sec:prelim}}
\label{app:prelim}
\subsection{Proof of Proposition~\ref{prop:alternative_banzhaf}} 
\textsc{Proposition} \ref{prop:alternative_banzhaf}.
{\it The following holds for any Boolean function $\varphi$ over $\bm X$ and variable $x \in \bm X$:
\begin{align*}
\banz(\varphi, x) = \#\varphi[x:=1] - \#\varphi[x:=0]
\end{align*}
}

\medskip
The proposition follows from the following simple equalities: 
\begin{align*}
  \banz(\varphi,x) \overset{(a)}{=}& \  \sum_{\bm Y \subseteq \bm X\setminus \{x\}}
\varphi[\bm Y \cup \{x\}] - \varphi[\bm Y]   \\
  =&\  \sum_{\bm Y \subseteq \bm X\setminus \{x\}}
\varphi[\bm Y \cup \{x\}] - \sum_{\bm Y \subseteq \bm X\setminus \{x\}} \varphi[\bm Y]   \\
  \overset{(b)}{=} &\ 
\#\varphi[x:=1] - \#\varphi[x:=0]
\end{align*}
Eq. $(a)$ holds by definition.
To obtain Eq. (b), we observe that 
for any subset 
$\bm Y \subseteq \bm X\setminus \{x\}$, it holds:
$\bm Y \cup \{x\}$ is a model of 
$\varphi$ 
if and only if $\bm Y$ 
is a model of $\varphi[x:= 1]$;  
$\bm Y$ is a model of 
$\varphi$ 
if and only if $\bm Y$ 
is a model of $\varphi[x:= 0]$.

\section{Missing Details in Section~\ref{sec:algorithm}} 
\label{app:algorithm}
\subsection{Explanations of Eq. \eqref{eq:ind_and_count} to \eqref{eq:mut_excl_banz}}
We explain Eq. \eqref{eq:ind_and_count} to \eqref{eq:mut_excl_banz}.
We consider a function $\varphi$ of the form $\varphi_1 \text{ op } \varphi_2$
and assume, without  loss of generality, that the variable $x$ is contained in 
$\varphi_1$.

We start with the case that 
$\varphi = \varphi_1 \wedge \varphi_2$ and $\varphi_1$ and 
    $\varphi_2$ are independent. In this case, we need to show the equalities:
\begin{align*}
\#\varphi =&\ \#\varphi_1 \cdot \#\varphi_2 &\hspace{2em} \eqref{eq:ind_and_count}  \\
\banz(\varphi,x) =&\ \banz(\varphi_1, x)\cdot \#\varphi_2 & \eqref{eq:ind_and_banz}
\end{align*}
Eq.~\eqref{eq:ind_and_count} holds because any pair $\theta_1$ and $\theta_2$ of models for 
$\varphi_1$ and respectively $\varphi_2$ can be combined into a model 
for $\varphi$.

Eq. \eqref{eq:ind_and_banz} can be derived as follows:
\begin{align*}
\banz(\varphi, x) \overset{(a)}{=} &\  \#\varphi[x = 1] - \#\varphi[x = 0]  \\
\overset{(b)}{=} &\ \#\varphi_1[x = 1] \cdot \#\varphi_2 - \#\varphi_1[x = 0] \cdot \#\varphi_2 \\
= &\ \big(\#\varphi_1[x = 1] -\#\varphi_1[x = 0]\big)\cdot \#\varphi_2 \\
\overset{(c)}{=} &\ \banz(\varphi_1, x) \cdot \#\varphi_2
\end{align*}
Eq. $(a)$ and $(c)$ hold by the characterization 
of the Banzhaf value given in 
Eq.~\eqref{eq:alternative_banzhaf}. 
Eq. $(b)$ follows from Eq.~\eqref{eq:ind_and_count}
and the relationship 
$\#\varphi_2[x:=0] = \#\varphi_2[x:=1] = \#\varphi_2$, which relies on the fact that 
$\varphi_2$ does not contain $x$.

Now, we consider the case that $\varphi = \varphi_1 \vee \varphi_2$ and $\varphi_1$ and $\varphi_2$ are independent. We show how to derive the following equalities:
\begin{align*}
\#\varphi = & \  
 \#\varphi_1 \cdot 2^{n_2} + 2^{n_1} \cdot \#\varphi_2 - \#\varphi_1 \cdot \#\varphi_2 & \eqref{eq:ind_or_count} \\
  \banz(\varphi, x) = &\  \banz(\varphi_1,x) \cdot (2^{n_2}- \#\varphi_2) &  \eqref{eq:ind_or_banz}\\
\end{align*}
where $n_i$ is the number of variables in $\varphi_i$, for $i \in [2]$.

We derive Eq.~\eqref{eq:ind_or_count}:
\begin{align*}
\#\varphi \overset{(a)}{=} &\  
\#\varphi_1 \cdot \#\varphi_2 + 
\#\varphi_1 \cdot (2^{n_2}- \#\varphi_2)
+ (2^{n_1} - \#\varphi_1) \cdot \#\varphi_2 \\
= &\ 
\#\varphi_1 \cdot \#\varphi_2 + 
\#\varphi_1 \cdot 2^{n_2} - \#\varphi_1 \cdot \#\varphi_2 + 
 2^{n_1}\cdot \#\varphi_2 - \#\varphi_1 \cdot \#\varphi_2 \\
= &\  
\#\varphi_1 \cdot 2^{n_2}  + 
 2^{n_1}\cdot \#\varphi_2 - \#\varphi_1 \cdot \#\varphi_2
\end{align*}
Eq. (a) holds because each model 
of $\varphi$ is either a model of both $\varphi_1$ and $\varphi_2$
or a model of exactly one of them.
The other equalities use the distributivity of multiplication over 
summation.

Eq. \eqref{eq:ind_or_banz} is implied by the following 
equations:
\begin{align*}
\banz(\varphi, x) \overset{(a)}{=} &\  
\#\varphi[x = 1] - \#\varphi[x = 0] \\
\overset{(b)}{=} &\ \Big[\#\varphi_1[x = 1] \cdot \#\varphi_2 +  \#\varphi_1[x = 1] \cdot (2^{n_2}- \#\varphi_2) \ + 
\\
&\ (2^{n_1-1} - \#\varphi_1[x = 1]) \cdot \#\varphi_2
\Big] \ -
\\
&\ \Big[\#\varphi_1[x = 0] \cdot \#\varphi_2 +  \#\varphi_1[x = 0] \cdot (2^{n_2}- \#\varphi_2) 
\ +
\\
&\ (2^{n_1-1} - \#\varphi_1[x = 0]) \cdot \#\varphi_2\Big] 
\\
= &\ \big(\#\varphi_1[x = 1]- \#\varphi_1[x = 0]\big) \cdot   \#\varphi_2 \ + \\
&\ \big(\#\varphi_1[x = 1]- \#\varphi_1[x = 0]\big) \cdot (2^{n_2}- \#\varphi_2) \ + \\
&\  \big(\#\varphi_1[x = 0] - \#\varphi_1[x = 1] \big) \cdot \#\varphi_2\\
= &\ 
 \big(\#\varphi_1[x = 1]- \#\varphi_1[x = 0]\big) \cdot (2^{n_2}- \#\varphi_2) \\
\overset{(c)}{=} &\ 
\banz(\varphi_1,x) \cdot (2^{n_2}- \#\varphi_2) 
\end{align*}
Eq. $(a)$ and $(c)$ follow from Eq.~\eqref{eq:alternative_banzhaf}.
Eq. $(b)$ follows from Eq.~\eqref{eq:ind_or_count}
and the equalities 
$\#\varphi_2[x:=0] = \#\varphi_2[x:=1] = \#\varphi_2$, which hold because 
$\varphi_2$ does not contain $x$.

Finally, we consider the case that $\varphi = \varphi_1 \vee \varphi_2$ and $\varphi_1$ and $\varphi_2$ are over the same variables but mutually exclusive. We explain the following equalities:
\begin{align*}
\#\varphi = & \  
 \#\varphi_1 +  \#\varphi_2 & \eqref{eq:mut_excl_count} \\
  \banz(\varphi, x) = &\  \banz(\varphi_1, x) +  \banz(\varphi_2, x) &  \eqref{eq:mut_excl_banz}\\
\end{align*}
Eq.~\eqref{eq:mut_excl_count} holds because every model of $\varphi$
is either a model of $\varphi_1$ or a model of $\varphi_2$.

Eq. \eqref{eq:mut_excl_banz} holds because:
\begin{align*}
\banz(\varphi, x) \overset{(a)}{=} &\  \#\varphi[x = 1] - \#\varphi[x = 0]  \\
\overset{(b)}{=} &\ \Big[\#\varphi_1[x = 1] + \#\varphi_2[x = 1]\Big]\ - \Big[\#\varphi_1[x = 0] +
\#\varphi_2[x = 0]\Big] \\
= &\ \Big[\#\varphi_1[x = 1] - \#\varphi_1[x = 0]\Big] \ + 
 \Big[\#\varphi_2[x = 1] - \#\varphi_2[x = 0]\Big] \\
\overset{(c)}{=} &\ \banz(\varphi_1, x) + 
\banz(\varphi_2, x)
\end{align*}
Eq. $(a)$ and $(c)$ follow from  Eq.~\eqref{eq:alternative_banzhaf}.
Eq. $(b)$ is implied by Eq.~\eqref{eq:mut_excl_count}.

\subsection{Proof of Proposition~\ref{prop:exaban_correct}}
\textsc{Proposition}
\ref{prop:exaban_correct}.
\textit{
For any positive DNF function $\varphi$, complete d-tree $T_{\varphi}$ for $\varphi$, and variable $x$ in $\varphi$, 
it holds 
$$\textsc{ExaBan}(T_{\varphi}, x) = (\banz(\varphi,x), \#\varphi).$$
}

\medskip

Proposition~\ref{prop:exaban_correct} is implied by the following lemma, which states that \textsc{ExaBan} computes the correct Banzhaf value and model count 
for each subtree of its input d-tree:
\begin{lemma}
\label{lem:exaban_correct}
For any positive DNF function $\varphi$, complete d-tree $T_{\varphi}$ for $\varphi$, subtree $T_{\xi}$ of $T_{\varphi}$ for some function $\xi$,
and variable $x$ in $\varphi$, 
it holds 
$$\textsc{ExaBan}(T_{\xi}, x) = (\banz(\xi,x), \#\xi).$$
\end{lemma}

\begin{proof}
Consider a positive DNF function $\varphi$, a complete d-tree $T_{\varphi}$ for $\varphi$, a subtree $T_{\xi}$ of $T_{\varphi}$ for some function $\xi$,
and a variable $x$ in $\varphi$. To prove Lemma~\ref{lem:exaban_correct},
we show by induction over the structure of $T_{\xi}$ that it holds   
$\textsc{ExaBan}(T_{\xi}, x) = (\banz(\xi,x), \#\xi)$.

\paragraph{Base Case of the Induction}
Assume that $T_{\xi}$ consists of the single node $\xi$. 
We analyze all 
cases for $\xi$.

\begin{itemize}
\item In case $\xi$ is $x$, $\textsc{ExaBan}$ returns $(1,1)$.
By Eq.~\eqref{eq:alternative_banzhaf}, we have 
$\banz(x,x) = \#x[x:=1] - \#x[x:=0] = 
\#1 - \#0 = 1-0 = 1$. We obtain the last equality by observing that  
the empty set is the only  model of the constant $1$.
It also holds 
that $\#x = 1$, since the assignment that maps $x$ to $1$ is the only 
 model of the function $x$.
It follows that the pair $(1,1)$ returned by  
\textsc{ExaBan} is correct in this case.

\item In case $\xi$ is $\neg x$, $\textsc{ExaBan}$ returns $(-1,1)$.
By Eq.~\eqref{eq:alternative_banzhaf}, it holds 
$\banz(\neg x,x) = \#\neg x[x:=1] - \#\neg x[x:=0] = 
\#0 - \#1 = 0-1 = -1$.
It also holds $\#\neg x = 1$, since the assignment that maps $x$ to $0$ is the only 
 model of $\neg x$.
We conclude that the pair $(-1,1)$ returned by  
\textsc{ExaBan} is correct.

\item In case $\xi$ is $1$ 
or a literal different from 
$x$ and $\neg x$, \textsc{ExaBan} returns $(0,1)$.
By Eq.~\eqref{eq:alternative_banzhaf}, it holds 
$\banz(\xi,x) = \#\xi[x:=1] - \#\xi[x:=0] = \#\xi - \#\xi = 0$.
We also observe that $\#\xi = 1$, because: 
if $\xi = 1$, the empty set is the only model of $\xi$;
if $\xi = y$ for a variable $y$, $\{y \mapsto 1\}$ is the only model of $\xi$;
if $\xi = \neg y$, $\{y \mapsto 0\}$ is the only model of $\xi$.
This implies that the pair $(0,1)$ returned by  
the procedure \textsc{ExaBan} is correct.

\item In case $\xi$ is $0$,   
\textsc{ExaBan} returns $(0,0)$.
By Eq.~\eqref{eq:alternative_banzhaf},
it holds $\banz(0,x) = 
\#0[x:=1] - \#0[x:=0] = \#0 - \#0 = 0$.
The constant $0$ cannot be satisfied by any assignment. 
Thus, the pair $(0,0)$ returned by  
\textsc{ExaBan} is correct.
\end{itemize}

\paragraph{Induction Step}
Assume that $T_{\xi}$
is of the form $T_{\xi_1} \text{ op }
T_{\xi_2}$. 
The procedure \textsc{ExaBan} first computes 
$(B_i,\#_i) \defeq \textsc{ExaBan}(T_{\xi_i},x)$ for $i \in [2]$.
The induction hypothesis is:
\begin{align}
\label{eq:ind_hypo_exaban}
\textsc{ExaBan}(T_{\xi_i},x) \defeq (B_i,\#_i) = (\banz(\xi_i,x),\#\xi_i)
\end{align}
for $i \in [2]$. We show that 
$\textsc{ExaBan}(T_{\xi},x)= (\banz(\xi,x),\#\xi)$.
This follows from 
Eq.~\eqref{eq:ind_and_count} to \eqref{eq:mut_excl_banz}. 
We analyze the case for $\text{op} = \odot$ in detail.
The cases for $\text{op} = \otimes$ and $\text{op} = \oplus$ are analogous. 

The procedure \textsc{ExaBan} 
returns the pair $(B_i \cdot \#_2, \#_1 \cdot \#_2)$.
By Eq.~\eqref{eq:ind_and_count}, it holds 
$\#\xi = \#\xi_1 \cdot \#\xi_2$.
Due to the induction hypothesis in Eq.~\eqref{eq:ind_hypo_exaban}, this implies $\#\xi = \#_1 \cdot \#_2$.
Hence, the model count computed by \textsc{ExaBan}
is correct. It remains to show that 
$B_1 \cdot \#_2 = \banz(\xi,x)$.

First, we consider the case that $x$ is not included in 
$\xi$. 
By Eq.~\eqref{eq:alternative_banzhaf}, it holds 
$\banz(\xi_1,x) = \#\xi_1[x:=1] - \#\xi_1[x:=0] =
\#\xi_1 - \#\xi_1 = 0$
and 
$\banz(\xi,x) = \#\xi[x:=1] - \#\xi[x:=0] = \#\xi - \#\xi = 0$.
By the induction hypothesis in Eq.~\eqref{eq:ind_hypo_exaban}, 
$B_1$ must be $0$. Hence, 
$B_1 \cdot \#_2 = 0 = \banz(\xi,x)$.
This means that the Banzhaf value computed by 
\textsc{ExaBan} is correct. 

Now, we consider the case that $x$ is in $\xi$.
Without loss generality, 
we assume that $x$ is in $\xi_1$.
By Eq.~\eqref{eq:ind_and_banz}, it holds 
$\banz(\xi,x) = \banz(\xi_1,x) \cdot \#\xi_2$.
By the induction hypothesis in Eq.~\eqref{eq:ind_hypo_exaban}, we obtain
$\banz(\xi,x) = B_1 \cdot \#_2$.
This means that the Banzhaf value computed by \textsc{ExaBan}
is correct. 
This completes the induction step 
for $\text{ op } = \odot$.
\end{proof}

\subsection{Proof of Proposition~\ref{prop:banz_count_bounds}}
\textsc{Proposition} \ref{prop:banz_count_bounds}.
{\textit For any positive DNF function $\varphi$ and variable $x$ in $\varphi$, 
it holds: 
\begin{align*}
\#L(\varphi) &\leq \#\varphi \leq \#U(\varphi) \\
\#L(\varphi[x:=1]) - \#U(\varphi[x:=0]) &\ \leq  \banz(\varphi,x) \\
&\ \leq   \#U(\varphi[x:=1]) - \#L(\varphi[x:=0])
\end{align*}
}

\medskip

We first prove the bounds on $\#\varphi$.
Consider a model $\theta$ for $L(\varphi)$.
The model must satisfy at least one clause $C$ 
in $L(\varphi)$. By construction, 
$C$ is included in $\varphi$. Let $\theta'$ be an assignment for $\varphi$ that results from $\theta$ by mapping all variables that appear in $\varphi$ but not in $L(\varphi)$ to $1$. Since $\theta'$ satisfies $C$, 
it is a model of $\varphi$. Observe that for two 
distinct models $\theta_1$ and $\theta_2$ for $L(\varphi)$,
the resulting models $\theta_1'$ and $\theta_2'$ must be distinct as well. This implies 
$\#L(\varphi) \leq \#\varphi$.

Consider now a model $\theta$ for $\varphi$.
The function $\varphi$ must contain at least one clause 
$C$ such that $\theta$ satisfies all literals in $C$. 
By construction, $U(\varphi)$ has the same variables as 
$\varphi$ and contains a clause $C'$ that results from $C$ by skipping variables. This means that $\theta$ satisfies $C'$, hence it  is a model of $U(\varphi)$. This implies 
$\#\varphi \leq \#U(\varphi)$.

The bounds on $\banz(\varphi,x)$ follow immediately from the bounds on the model counts and 
the alternative characterization of the Banzhaf value given in Eq.~\eqref{eq:alternative_banzhaf}:

\begin{align*}
\banz(\varphi,x) =&\ \#\varphi[x:=1] - \#\varphi[x:=0] \\
\geq &\ \#L(\varphi[x:=1]) - \#U(\varphi[x:=0]) 
\end{align*}

\begin{align*}
\banz(\varphi,x) =&\ \#\varphi[x:=1] - \#\varphi[x:=0] \\
\leq &\ \#U(\varphi[x:=1]) - \#L(\varphi[x:=0]) 
\end{align*}

\subsection{Proof of Proposition~\ref{prop:bounds_correct}}
\textsc{Proposition} \ref{prop:bounds_correct}.
\textit{For any positive DNF function $\varphi$, d-tree $T_{\varphi}$ for $\varphi$,  
and variable $x$ in $\varphi$, it holds 
$\textsc{bounds}(T_{\varphi}, x) = (L_b, L_{\#}, U_b, U_{\#})$
such that 
$L_b \leq \banz(\varphi,x) \leq U_b$ 
and 
$L_{\#} \leq \#\varphi \leq U_{\#}$.
}

\medskip

Proposition~\ref{prop:bounds_correct} is implied 
by the following lemma:
\begin{lemma}
\label{lem:bounds}
For any positive DNF function $\varphi$, d-tree $T_{\varphi}$ for $\varphi$,
subtree $T_{\xi}$ of $T_{\varphi}$ for some function $\xi$,
and variable $x$ in $\varphi$, it holds 
$\textsc{bounds}(T_{\xi}, x) = (L_b, L_{\#}, U_b, U_{\#})$
such that 
$L_b \leq \banz(\xi,x) \leq U_b$ 
and 
$L_{\#} \leq \#\xi \leq U_{\#}$.
\end{lemma}

\begin{proof}
Consider a positive DNF function $\varphi$, a complete d-tree $T_{\varphi}$ for $\varphi$, a subtree $T_{\xi}$ of $T_{\varphi}$ for some function $\xi$,
and a variable $x$ in $\varphi$. The proof of Lemma~\ref{lem:exaban_correct}
is by induction over the structure of $T_{\xi}$.

\paragraph{Base Case of the Induction}
Assume that $T_{\xi}$ consists of the single node $\xi$. 
We consider the cases that $\xi$ is a literal, a constant, or 
a function that is not a literal nor a constant.

\begin{itemize}
\item If $\xi$ is a literal or a constant, 
the procedure \textsc{bounds} calls the procedure $\textsc{ExaBan}(\xi, x)$ from Figure~\ref{alg:exactban}, which computes the exact 
values $\banz(\xi,x)$ and $\#\xi$ (Lemma~\ref{lem:exaban_correct}). 
Hence, the output of \textsc{bounds}
is correct in this case. 

\item Consider the case that $\xi$ is not a literal nor a constant. Since $\varphi$ is a positive DNF function, also
$\xi$ must be a positive DNF function.  
The procedure $\textsc{bounds}$ sets
\begin{align*}
L_\# \defeq &\ \#L(\xi),  \\
U_\# \defeq &\ \#U(\xi), \\
L_b \defeq &\  \#L(\xi[x:=1]) - \#U(\xi[x:=0]),  \text{ and }\\ 
U_b \defeq &\ \#U(\xi[x:=1]) - \#L(\xi[x:=0]).
\end{align*}
By Proposition~\ref{prop:banz_count_bounds}, it holds 
\begin{align*}
L_\# \defeq &\ \#L(\xi) \leq \#\xi \leq \#U(\xi) \defeq U_\# \text{ and } \\
L_b \defeq &\  \#L(\xi[x:=1]) - \#U(\xi[x:=0]) \\
\leq &\  \banz(\xi,x) \\
\leq &\ \#U(\xi[x:=1]) - \#L(\xi[x:=0])
\defeq U_b. 
\end{align*}
Thus, also in this case the output of \textsc{bounds} is correct. 
\end{itemize}

\paragraph{Induction Step}
Assume that $T_{\xi}$ is of the form $T_{\xi_1} \text{ op }T_{\xi_2}$.
The procedure \textsc{bounds} computes 
$(L^{(i)}_b,L^{(i)}_\#,U^{(i)}_b,U^{(i)}_\#) \defeq 
\textsc{bounds}(T_{\xi_i}, x)$, for $i\in[2]$.
The induction hypothesis states that the following inequalities hold: 
\begin{align*}
L_{\#}^{(1)} \leq &\ \#\xi_1 \leq U_{\#}^{(1)}, \\
L_{\#}^{(2)} \leq &\ \#\xi_2 \leq U_{\#}^{(2)}, \\
L_b^{(1)} \leq &\ \banz(\xi_1,x) \leq U_b^{(1)}, \text{ and } \\
L_b^{(2)} \leq &\ \banz(\xi_2,x) \leq U_b^{(2)}.
\end{align*}
We consider the case that $\text{ op } = \otimes$ and show that the following quantities 
$L_{\#}$ and  $L_b$ computed by \textsc{bounds} are indeed lower bounds 
for $\#\xi$ and respectively $\banz(\xi,x)$. 
\begin{align*}
L_\# \defeq &\ L^{(1)}_\#\cdot 2^{n_2} + L^{(2)}_\#\cdot 2^{n_1} - L^{(1)}_\#\cdot L^{(2)}_\# \text{ and } \\
L_b \defeq &\  L^{(1)}_b\cdot (2^{n_2} - U^{(2)}_\#).
\end{align*}
The other cases are handled analogously. 

Without loss of generality, assume that $x$ is in $\xi_1$ if it is in $\xi$.
First, we show that $L_\# \leq \#\xi$. This is implied by 
the following (in)equalities, where $n_i$ is the number of variables 
in $\xi_i$ for $i \in [2]$.
\begin{align*}
\#\xi - L_\# \overset{(a)}{=} &\ \#\xi_1 \cdot 2^{n_2} + \#\xi_2 \cdot 2^{n_1} - 
\#\xi_1 \cdot \#\xi_2 - \\
&\ (L^{(1)}_\#\cdot 2^{n_2} + L^{(2)}_\#\cdot 2^{n_1} - L^{(1)}_\#\cdot L^{(2)}_\#) \\
\overset{(b)}{=}&\ (\#\xi_1-L^{(1)}_\#) \cdot 2^{n_2} + (\#\xi_2-L^{(2)}_\#) \cdot 2^{n_1}  - 
\#\xi_1 \cdot \#\xi_2 + L^{(1)}_\#\cdot L^{(2)}_\#  \\
\overset{(c)}{\geq} &\ (\#\xi_1-L^{(1)}_\#) \cdot \#\xi_2 + (\#\xi_2-L^{(2)}_\#) \cdot \#\xi_1  - 
\#\xi_1 \cdot \#\xi_2 + L^{(1)}_\#\cdot L^{(2)}_\#  \\
\overset{(d)}{=} &\ \#\xi_1\cdot \#\xi_2 -  L^{(1)}_\# \cdot \#\xi_2 + \#\xi_1 \cdot \#\xi_2 - 
L^{(2)}_\# \cdot \#\xi_1  - 
\#\xi_1 \cdot \#\xi_2 + L^{(1)}_\#\cdot L^{(2)}_\#  \\
= &\ (L^{(1)}_\# \cdot L^{(2)}_\# + \#\xi_1\cdot \#\xi_2) -
(L^{(1)}_\# \cdot \#\xi_2 + \#\xi_1 \cdot L^{(2)}_\#)
\overset{(e)}{\geq} 0
\end{align*}
Eq. $(a)$ follows from Eq.~\eqref{eq:ind_or_count} and the definition of $L_{\#}$. 
We obtain Eq.~$(b)$ and~$(d)$ using the distributivity of multiplication over addition. 
Ineq. $(c)$ holds because the number of models of $\xi_i$ can be at most $2^{n_i}$, for $i \in [2]$.
For Ineq. $(e)$, it suffices to show: 
$$(L^{(1)}_\# \cdot \#\xi_2 + \#\xi_1 \cdot L^{(2)}_\#) \leq 
(L^{(1)}_\# \cdot L^{(2)}_\# + \#\xi_1\cdot \#\xi_2).$$
To show the latter inequality, we first observe that 
$L^{(i)}_\# \leq \#\xi_i$ for $i \in [2]$, by induction hypothesis. 
Then, we use the rearrangement inequality~\cite{hardy1952inequalities}.

Now, we show $L_b \leq \banz(\xi,x)$. This holds, because:
\begin{align*}
\banz(\xi,x) \overset{(a)}{=} &\ \banz(\xi_1,x) \cdot (2^{n_2} - \#\xi_2) \\
\overset{(b)}{\geq} &\ L^{(1)}_b\cdot (2^{n_2} - U^{(2)}_\#) \defeq L_b
\end{align*}
Eq. (a) holds due to Eq.~\eqref{eq:ind_or_banz}.
Observe that in case $x$ is not included in $\xi$, we have 
$\banz(\xi,x) = \banz(\xi_1,x) = 0$.
Eq. (b) follows from the induction hypothesis saying that 
$L^{(1)}_b \leq \banz(\xi_1,x)$ and 
$\#\xi_2 \leq U^{(2)}_\#$.
\end{proof}

We close this section with an auxiliary lemma that 
will be useful in the proof of 
Proposition~\ref{prop:adaban_correct}.
It states that \textsc{bounds} computes the
exact Banzhaf value in case the input d-tree is complete.

\begin{lemma}
\label{lem:bounds_exact_complete}
For any positive DNF function $\varphi$, complete 
d-tree $T_{\varphi}$ for $\varphi$,  
and variable $x$ in $\varphi$, it holds 
$\textsc{bounds}(T_{\varphi}, x) = (L_b, \cdot , U_b, \cdot )$
such that 
$L_b \leq \banz(\varphi,x) \leq U_b$. 
\end{lemma}

\begin{proof}
The main observation is as follows. Each leaf of $T_{\varphi}$ is either  
a literal or a constant. For each such leaf $\ell$, the procedure \textsc{bounds} calls $\textsc{ExaBan}(\ell, x)$, which, by Lemma~\ref{lem:exaban_correct}, computes $\banz(\ell,x)$ exactly. Then, the lemma follows from a simple structural induction as in the proof of Lemma~\ref{lem:bounds}. 
\end{proof}

\subsection{Proof of Proposition~\ref{prop:adaban_correct}}

\textsc{Proposition} \ref{prop:adaban_correct}.
\textit{
For any positive DNF function $\varphi$,
d-tree $T_{\varphi}$ for $\varphi$,  variable $x$ in $\varphi$, 
error $\epsilon$, 
and bounds $L \leq \banz(\varphi,x) \leq U$,
it holds 
$\textsc{AdaBan}(T_{\varphi}, x, \epsilon, [L,U]) = [\ell, u]$ such that every value in $[\ell, u]$ is an $\epsilon$-approximation 
of $\banz(\varphi, x)$.
}

\medskip
The procedure \textsc{AdaBan} first calls $\textsc{bounds}(T_{\varphi}, x)$
to compute a lower bound $L_b$ and an upper bound 
$U_b$ for $\banz(\varphi,x)$ (Proposition~\ref{prop:bounds_correct}). 
Then, it updates the bounds $L$ and $U$
by setting $L \defeq \max\{L,L_b\}$
and $U \defeq \min\{U,U_b\}$ and
checks whether 
\begin{align}
\label{eq:stopping_condition}
(1-\epsilon)\cdot U - (1+\epsilon)\cdot L \leq 0.
\end{align}
If this holds, it returns
the interval $[(1- \epsilon) \cdot U, (1 + \epsilon) \cdot L]$. Otherwise, it picks a node 
in $T_{\varphi}$ that is not a literal nor a constant, decomposes 
it into independent or mutually exclusive functions, and repeats the above 
steps. 

First, we explain that the procedure \textsc{AdaBan} reaches a state 
where Condition~\eqref{eq:stopping_condition} holds. 
Then, we show that this condition implies that each value 
in the interval $[(1- \epsilon) \cdot U, (1 + \epsilon) \cdot L]$ is a relative $\epsilon$-approximation
of $\banz(\varphi,x)$.
 
In case $T_{\varphi}$ is a complete d-tree,  
$\textsc{bounds}(T_{\varphi}, x)$ computes the
$\banz(\varphi,x)$ exactly (Lemma~\ref{lem:bounds_exact_complete}), which means that 
$L$ and $U$ are set to $\banz(\varphi,x)$.
This implies
\begin{align*}
&\ (1-\epsilon)\cdot U - (1+\epsilon)\cdot L \\
= &\ (1-\epsilon)\cdot \banz(\varphi,x) - (1+\epsilon)\cdot \banz(\varphi,x) \\
= &\  - 2\epsilon \cdot \banz(\varphi,x) \leq 0,
\end{align*}
which means that, at the latest when $T_{\varphi}$ is complete, Condition~\eqref{eq:stopping_condition} is satisfied.

Assume now that $L$ and $U$ are a lower and respectively an upper bound
of $\banz(\varphi,x)$ such that 
Condition~\eqref{eq:stopping_condition} is satisfied.
The condition implies
$(1-\epsilon)\cdot U \leq  (1+\epsilon)\cdot L$.
Consider now an arbitrary value $B$ in the interval 
$[(1- \epsilon) \cdot U, (1 + \epsilon) \cdot L]$.
It holds:
\begin{align*}
B \geq&\ (1-\epsilon)\cdot U \geq  (1-\epsilon)\cdot \banz(\varphi, x) \text{ and } \\
B \leq&\ (1+\epsilon)\cdot L \leq  (1+\epsilon)\cdot \banz(\varphi, x)
\end{align*}
This means that $B$ is a relative $\epsilon$-approximation 
for $\banz(\varphi, x)$.
\section{Missing Details in Section~\ref{sec:ranking_top_k}}
\label{app:ranking}
In this section, we prove
the intractability part of Theorem~\ref{theo:dichotomy}:

\begin{proposition}
\label{prop:comp_banz_hard}
For any non-hierarchical Boolean CQ $Q$ without self-joins, 
the problem $\compbanz_Q$
cannot be solved in polynomial time,
unless there is an $\fptas$ for $\countbis$.
\end{proposition}

We prove Proposition~\ref{prop:comp_banz_hard}
in two steps. 
In Sec.~\ref{sec:intrac_basic_non_hierarchical}, 
we show intractability of 
$\compbanz_Q$ for the basic 
non-hierarchical CQ:
\begin{align}
\label{eq:basic_non_hierarchical}
Q_{nh} = \exists X\exists Y\ R(X) \wedge S(X,Y) \wedge  T(Y)
\end{align}
In Sec.~\ref{sec:intrac_arbit_non_hierarchical},
we extend the intractability  result to 
arbitrary self-join-free non-hierarchical Boolean CQs. 

\subsection{Intractability for the Basic Non-Hierarchical CQ}
\label{sec:intrac_basic_non_hierarchical}
We say that a Boolean function is in  \pptwodnf if it is positive, in disjunctive normal form (DNF), and its set of 
variables is partitioned into two disjoint sets $\bm Y$ and $\bm Z$ such that each clause is the conjunction of a variable from $\bm Y$ and a variable from $\bm Z$.

To simplify the following reasoning, we introduce the 
problem $\countnsat$ of counting non-satisfying assignments 
of \pptwodnf functions and state some auxiliary lemmas. 

\begin{center}
\fbox{%
    \parbox{0.73\linewidth}{%
    \begin{tabular}{ll}
   Problem: & $\countnsat$ \\
   Description: &  \textit{Counting non-satisfying assignments of $\pptwodnf$ functions} \\
    Input: & $\pptwodnf$ function $\varphi$  \\
    Compute: & Number of non-satisfying assignments 
    of $\varphi$.
    \end{tabular}
    }}%
\end{center}

The impossibility of an 
\fptas for $\countbis$ implies 
the impossibility of an \fptas for 
$\countnsat$:

\begin{lemma}
\label{cl:countbunsat_not_fptas}
There is no $\fptas$ for $\countnsat$, if  
there is no $\fptas$ for $\countbis$.
\end{lemma}

\begin{proof}
We give a polynomial parsimonious reduction from  
$\countbis$ to $\countnsat$. That is, given a bipartite graph 
$G$, we construct a $\pptwodnf$ function 
$\varphi_G$ such that $\countbis(G) = \countnsat(\varphi_G)$.
Then, any \fptas $A$ 
for \countnsat can easily be turned into an \fptas for 
$\countbis$ as follows: Given $0 < \epsilon < 1$ and 
an input graph $G$, we convert $G$ 
into $\varphi_G$ and compute $A(\varphi_G)$.
Due to the parsimonious reduction, it holds 
$(1-\epsilon)\cdot \countbis(G) \leq A(\varphi_G) \leq (1+\epsilon)\cdot \countbis(G)$.

We now explain the reduction.
Given a bipartite graph $G = (V,E)$ with
node set $V = U \cup W$ for disjoint sets $U$ and $V$ and 
edge relation $E \subseteq U \times W$, 
we construct the 
\pptwodnf function
$\varphi_G = \bigvee_{(u,v) \in E} (x_u \wedge x_v)$.
A set $V' \subseteq V$ is an independent set
of $G$ if and only if $\{x_w \mid w \in V' \}$ 
is a non-satisfying assignment for $\varphi$. 
This implies 
$\countbis(G) = \countnsat(\varphi)$.
\end{proof}

Prior work shows how to construct from each 
\pptwodnf function $\varphi$ a database $D$
such that 
$\varphi_{Q_{nh},D}=\varphi$, where 
$Q_{nh}$ is the non-hierarchical CQ 
given in Eq.~\eqref{eq:basic_non_hierarchical}
and $\varphi_{Q_{nh},D}$ is the lineage of $Q$ over $D$~\cite{DalviS12}. 
For the sake of completeness, we give here the construction.

\begin{lemma}
\label{lemma:provIsBipartiteDNF}
For any $\pptwodnf$ function $\varphi$,
one can construct in time linear in $|\varphi|$ a 
database $D$ such that $\varphi_{Q_{nh},D}=\varphi$. 
\end{lemma}

\begin{proof}
Consider a  $\pptwodnf$ function $\varphi$ over 
disjoint variable sets 
$\bm X$ and $\bm Y$. 
We construct a database $D$ that consists of the relations 
$R= \{a_x \mid x \in \bm X\}$, $T= \{a_y \mid y \in \bm Y\}$, and 
$S = \{(a_x,a_y)\mid (x \wedge y) \text{ is a clause in } \varphi\}$. 
We set all facts in $R$ and $T$ to be endogenous and all
facts in $S$ to be exogenous. 
We associate each fact $a_x$ in $R$ ($b_y$ in $T$) with 
the variable $x$ ($y$). By construction, 
$\varphi_{Q_{nh},D}=\varphi$. The construction of $D$ requires 
a single pass over $\varphi$, hence the construction time is linear in $|\varphi|$.  
\end{proof}

The following lemma establishes the basic building 
block of a polynomial-time  approximation 
scheme for $\countnsat$.

\begin{lemma}
\label{lem:decide_non_sat}
Assume there is a polynomial-time algorithm 
for the problem $\compbanz_{Q_{nh}}$.
Given a $\pptwodnf$ function $\varphi$
over disjoint variable sets $\bm X$ and $\bm Y$
and $m \in \mathbb{N}$, we can decide in polynomial time in $|\varphi|$ and $m$ whether 
$\countnsat(\varphi) \leq 2^{|\bm X|} \cdot (\frac{3}{2})^m$.
\end{lemma}

\begin{proof}
We first introduce some notation.    
Given  a  $\pptwodnf$ function $\psi$ over disjoint variable sets 
$\bm X$ and $\bm Y$ and a fresh variable 
$z \not\in (\bm X \cup \bm Y)$,
we denote by $z \hat{\wedge} \psi$
the $\pptwodnf$ function 
$\psi \vee \bigvee_{y \in \bm Y} z \wedge y$.

Consider a $\pptwodnf$ function $\varphi$
over disjoint variable sets $\bm X$ and $\bm Y$
and an $m \in \mathbb{N}$. 
We denote by $\psi_m$ the 
$\pptwodnf$ function $(z^1_1 \wedge z^2_1)
\vee \cdots \vee (z^1_m \wedge z^2_m)$ such that the variables $z_i^j$ do not occur in $\varphi$.
Let $x$ and $y$ be fresh variables not contained in 
$\varphi$ nor in $\psi_m$. 
Consider the \pptwodnf function $\xi = (x \hat{\wedge} \varphi) \vee (y \hat{\wedge} \psi_m)$
whose clauses are visualized in the following figure. 
The variables in $\varphi$ are represented as bullets and each edge between two variables
symbolizes a conjunction between them.

\medskip

\begin{minipage}{\linewidth}
\begin{center}
\tikz {
\node at (0,0)  (11) {$\bullet$};s
\node at (0,-0.5)  (12) {$\bullet$};
\node at (0,-1)  (13) {$\bullet$};

\node at (1,0)  (21) {$\bullet$} edge[-] (12);
\node at (1,-0.5)  (22) {$\bullet$} edge[-] (11);
\node at (1,-1)  (23) {$\bullet$} edge[-] (13);
\node at (1,-1.5)  (24) {$\bullet$} edge[-] (11);
\draw[-] (24) -- (12);

\node at (0,-2)  (14) {$x$};

\draw[-] (14) -- (21);
\draw[-] (14) -- (22);
\draw[-] (14) -- (23);
\draw[-] (14) -- (24);

\node at (1.4,-0.7)  (phi) {$\varphi$};

\node at (2.1,-0.7)  (phi) {$x \hat{\wedge} \varphi$};

\draw[rounded corners=0.1cm, dashed] (-0.2,-1.7) rectangle (1.6,0.2);
\draw[rounded corners=0.1cm, dashed] (-0.4,-2.2) rectangle (2.5,0.4);

\node at (4,0)  (31) {$z_1^1$};
\node at (4,-0.5)  (32) {$z_2^1$};

\node at (4,-0.9)  (.) {$\cdot$};
\node at (4,-1)  (.) {$\cdot$};
\node at (4,-1.1)  (.) {$\cdot$};

\node at (4,-1.4)  (33) {$z_m^1$};

\node at (5,0)  (41) {$z_1^2$} edge[-] (31);
\node at (5,-0.5)  (42) {$z_2^2$} edge[-] (32);

\node at (5,-0.9)  (.) {$\cdot$};
\node at (5,-1)  (.) {$\cdot$};
\node at (5,-1.1)  (.) {$\cdot$};

\node at (5,-1.4)  (43) {$z_m^2$} edge[-] (33);

\node at (4,-2)  (34) {$y$};

\draw[-] (34) -- (41);
\draw[-] (34) -- (42);
\draw[-] (34) -- (43);

\node at (5.6,-0.7)  (phi) {$\psi_m$};

\node at (6.5,-0.7)  (phi) {$y \hat{\wedge} \psi_m$};

\draw[rounded corners=0.1cm, dashed] (3.7,-1.7) rectangle (6,0.3);
\draw[rounded corners=0.1cm, dashed] (3.5,-2.2) rectangle (7,0.4);
  }
  \end{center}
\end{minipage}

\medskip

The size of $\xi$ is linear in $|\varphi|$ and $m$.
Using Lemma~\ref{lemma:provIsBipartiteDNF}, we create
in time linear in $|\varphi|$ and $m$ a database $D_m$ such that 
$\varphi_{Q_{nh},D_m} = \xi$. 

Let $f_x$ and $f_y$ be the facts in $D_m$ associated with the variable $x$ and respectively $y$.
We first compute $\banz(Q_{nh},D_m,f_x)$.
This is equal to the number of sets $\bm Z$ of variables of $\xi)$
such that (1) $\bm Z$ does not include $x$, (2) $\bm Z$ does not 
satisfy $\xi$, but (3) $\bm Z \cup \{x\}$ satisfies $\xi$.
Each such set must include at least 
one variable from $\bm Y$.
The number of non-satisfying assignments of $\varphi$
containing at least one variable from 
$\bm Y$ is 
$\countnsat(\varphi) - 2^{|\bm X|}$. 
The number of non-satisfying assignments 
of $y \hat{\wedge} \psi^m$ that do not include 
$y$ is $3^m$ and the 
number of those that {\em do} include 
$y$ is $2^m$. 
Hence, the overall number of non-satisfying assignments 
of $y \hat{\wedge} \psi^m$ is $3^m + 2^m$.
This implies that 
$\banz(Q_{nh},D_m,f_x) = 
(\countnsat(\varphi) - 2^{|\bm X|})\cdot (3^m + 2^m)$.
Analogously, we
compute $\banz(Q_{nh},D_m,f_y)$.
This is equal to the number of sets $\bm Z$ of variables of $\xi$
such that (1) $\bm Z$ does not include $y$, (2) $\bm Z$ does not satisfy $\xi$, but 
(3) $\bm Z \cup \{y\}$ satisfies $\xi$.
Each such set must include at least 
one  $z^2_k$ with $k \in [m]$.
The number of non-satisfying assignments of $\psi_m$
containing at least one variable 
$z^2_k$ is $3^m- 2^m$.
The number of non-satisfying assignments 
of $x \hat{\wedge} \varphi$ that do not include 
$x$ is $\countnsat(\varphi)$ and the 
number of those that {\em do} include 
$x$ is $2^{|\bm X|}$. 
This means that  number of non-satisfying assignments 
of $x \hat{\wedge} \varphi$ 
is $\countnsat(\varphi) + 2^{|\bm X|}$.  
Hence, 
$\banz(Q_{nh},D_m,f_y)  = 
(3^m - 2^m) \cdot (\countnsat(\varphi) + 2^{|\bm X|})$.

Using these quantities, we obtain:

\begin{align*}
&\ \banz(Q_{nh},D_m,f_x) \leq \banz(Q_{nh},D_m,f_y) \\
\Leftrightarrow &\ (\countnsat(\varphi) - 2^{|\bm X|}) (3^m + 2^m) \leq 
(3^m - 2^m) (\countnsat(\varphi) + 2^{|\bm X|}) \\
\overset{(a)}{\Leftrightarrow} &\ 
\countnsat(\varphi) \cdot 3^m +
\countnsat(\varphi) \cdot 2^m -
2^{|\bm X|} \cdot 3^m -
2^{|\bm X|} \cdot 2^m \leq \\
&\ 
\countnsat(\varphi)\cdot  3^m -
\countnsat(\varphi) \cdot 2^m +
2^{|\bm X|} \cdot 3^m -
2^{|\bm X|} \cdot 2^m\\
\overset{(b)}{\Leftrightarrow} &\ 
\countnsat(\varphi)\cdot 2^m - 2^{|\bm X|} \cdot 3^m \leq 
2^{|\bm X|} \cdot 3^m - \countnsat(\varphi) \cdot 2^m\\
\Leftrightarrow &\ 
2 \cdot \countnsat(\varphi) \cdot 2^m  
\leq 2 \cdot 2^{|\bm X|}\cdot 3^m \\ 
\Leftrightarrow &\ 
\countnsat(\varphi)  
\leq 2^{|\bm X|}\cdot (\frac{3}{2})^m 
\end{align*}
Equivalence $(a)$ follows from the 
distributivity of addition and subtraction 
over product. We obtain Equivalence $(b)$
by subtracting $\countnsat(\varphi) \cdot 3^m$
and adding $2^{|\bm X|} \cdot 2^m$ on both sides of the 
inequality. 

We conclude that, given a polynomial-time algorithm for the problem  
$\compbanz_{Q_{nh}}$,
we can decide in polynomial time in $|\varphi|$ and $m$ whether 
$\countnsat(\phi) \leq 2^{|\bm X|} \cdot (\frac{3}{2})^m$.
\end{proof}

We say that an algorithm $A$ 
is an approximation algorithm  
for $\countnsat$ with {\em upper} approximation error $\frac{1}{2}$, 
if for each \pptwodnf function  
$\varphi$, it returns a value $A(\varphi)$ with 
$\countnsat(\varphi) \leq  A(\varphi) \leq \frac{3}{2}\cdot\countnsat(\varphi)$.
Using Lemma~\ref{lem:decide_non_sat}, we can 
design an 
approximation algorithm  
for $\countnsat$ with upper approximation error $0.5$
that runs in polynomial time.

\begin{lemma}
\label{lem:1.5-approximation}
Given a polynomial-time algorithm 
for $\compbanz_{Q_{nh}}$,
one can design a polynomial-time approximation algorithm 
for $\countnsat$ with upper approximation error 
$\frac{1}{2}$.
\end{lemma}

\begin{proof}
Assume that we have a polynomial-time algorithm for 
$\compbanz_{Q_{nh}}$.
The following is
a polynomial-time  approximation algorithm 
for $\countnsat$ with upper approximation error 
$\frac{1}{2}$.

  \begin{center}
  \renewcommand{\arraystretch}{1.15}
  \begin{tabular}{@{\hskip 0.1in}l}
  \toprule
  \textsc{Approx\#NSat}(\pptwodnf function $\varphi$)\\
  outputs value $v$ with  
  $\countnsat(\varphi) \leq  v \leq \frac{3}{2}\cdot\countnsat(\varphi)$\\
  \midrule
  \LET $\varphi$ be over the disjoint variable sets $\bm X$ and $\bm Y$\\
  $n:=$ the number of variables in $\varphi$\\
  $v:=0$ // \texttt{initialization} \\
  \FOREACH $i = 1 , \ldots , 2n$ \\
  \TAB \IF $\countnsat(\varphi) \leq 
  (\frac{3}{2})^{i} \cdot 2^{|\bm X|}$ and $v = 0$ \\ 
   \TAB \TAB $v:= (\frac{3}{2})^{i} \cdot 2^{|\bm X|}$ \\
\RETURN  $v$ \\
  \bottomrule
  \end{tabular}
  \end{center}

  The algorithm 
  returns
  $(\frac{3}{2})^{i} \cdot 2^{|\bm X|}$
  for the smallest   $i \in \{1, \ldots , 2n\}$ 
  such $\countnsat(\varphi) \leq 
  (\frac{3}{2})^{i} \cdot 2^{|\bm X|}$
  (and returns $0$ if no such $i$ exists).
  \nop{
  searches for the smallest 
  $i \in \{1, \ldots , 2n\}$
  and returns 
  $2n$ and checks for each $i$ whether 
  $\countnsat(\varphi) \leq 
  (\frac{3}{2})^{i} \cdot 2^{|\bm X|}$. 
  If this holds, it stops and returns
  $(\frac{3}{2})^{i} \cdot 2^{|\bm X|}$.
}

  \paragraph{Running time}
  The variable $i$ iterates over linearly many 
  values. 
  Each of these values is linear in $|\varphi|$.
  By Lemma~\ref{lem:decide_non_sat}, 
we can check   the condition $\countnsat(\varphi) \leq 
  (\frac{3}{2})^{i} \cdot 2^{|\bm X|}$
in polynomial time,
given a polynomial-time algorithm for $\compbanz_{Q_{nh}}$.

\paragraph{Upper approximation error $\frac{1}{2}$}
First, observe that 
$$2^{|\bm X|} \ \overset{(a)}{\leq}\  \countnsat(\varphi)\ 
\overset{(b)}{\leq} \ (\frac{3}{2})^{2n}$$
Inequality $(a)$ is implied by the fact that 
each subset of $\bm X$ is a non-satisfying assignment for $\varphi$. 
Inequality $(b)$ holds because of
$2^n < (\frac{3}{2})^{2n}$ $=$ $(\frac{3^2}{2^2})^{n}$  .
Due to these inequalities, there exists an 
$i \in \{1, \ldots , 2n\}$ such that
$$(\frac{3}{2})^{i-1}\cdot 2^{|\bm X|} \ 
\overset{(c)}{\leq}\  \countnsat(\varphi)\ 
\overset{(d)}{\leq} \ (\frac{3}{2})^{i}\cdot 2^{|\bm X|}.$$
Algorithm \textsc{Approx\#NSat}
returns $(\frac{3}{2})^{i}\cdot 2^{|\bm X|}$ for such $i$.
It holds
$$(\frac{3}{2})^{i}\cdot 2^{|\bm X|} 
= \frac{3}{2} (\frac{3}{2})^{i-1}\cdot 2^{|\bm X|}
\leq \frac{3}{2} \countnsat(\varphi),$$
where the last inequality follows from Inequality 
$(c)$. Hence, together with Inequality $(d)$,
we obtain $\countnsat(\varphi) \leq (\frac{3}{2})^i \cdot 2^{|\bm X|}
\leq \frac{3}{2} \cdot \countnsat(\varphi)$.
\end{proof}

We are ready to prove 
Proposition~\ref{prop:comp_banz_hard}.
Given a \pptwodnf function $\varphi$ and 
$k \in \mathbb{N}$, we denote 
by $\varphi^k$ the \pptwodnf function 
$\varphi_1 \vee \cdots \vee \varphi_k$,
where each $\varphi_i$ results from $\varphi$
by replacing each variable with a fresh one. 
Since non-satisfying assignments of $\varphi^k$ consist
of non-satisfying assignments of 
$\varphi_1, \ldots , \varphi_k$, we have
\begin{align} 
\label{eq:countphi_k=countphi^k}
\countnsat(\varphi^k) = \countnsat(\varphi)^k
\end{align}

Assume that the problem $\compbanz_{Q_{nh}}$
can be solved in polynomial time. In the following,
we design an 
\fptas for \countnsat. Then, 
Lemma~\ref{cl:countbunsat_not_fptas} implies 
that there is an \fptas for \countbis, which completes the proof of 
Proposition~\ref{prop:comp_banz_hard}.

Consider an arbitrary  \pptwodnf function $\varphi$ and $0 < 
\epsilon < 1$. It suffices to design an algorithm 
that runs in time polynomial in 
$|\varphi|$ and $\epsilon^{-1}$ and
computes a value $v$ such that 
\begin{align}
\label{eq:upper_approx}
\countnsat(\varphi) \leq v \leq (1+ \epsilon) \cdot \countnsat(\varphi).
\end{align}
We choose a $\lambda$ such that 
$\frac{\epsilon}{2} \leq \lambda
\leq \epsilon$ and $\lambda^{-1}$
is an integer. We explain in the following how to compute a value $v$ such that 
$\countnsat(\varphi) \leq v \leq (1+ \lambda) \cdot \countnsat(\varphi)$, which implies Eq.~\eqref{eq:upper_approx}.

We construct $\varphi^{2\lambda^{-1}}$ and use 
Lemma~\ref{lem:1.5-approximation} to 
compute a value $\hat{v}$ 
such that 
$\countnsat(\varphi^{2\lambda^{-1}})$ 
$\leq$ $\hat{v}$ $\leq$ 
$\frac{3}{2} \cdot \countnsat(\varphi^{2\lambda^{-1}})$.
Due to Eq.~\eqref{eq:countphi_k=countphi^k}, it holds 
\begin{align*}
\countnsat(\varphi)^{2\lambda^{-1}} \overset{(a)}{\leq} \hat{v} 
\overset{(b)}{\leq} 
\frac{3}{2} \cdot \countnsat(\varphi)^{2\lambda^{-1}}.
\end{align*}
Since $|\varphi^{2\lambda^{-1}}|$
is polynomially bounded in $|\varphi|$ and 
$\lambda^{-1}$, hence in 
$\epsilon^{-1}$, the computation time
is polynomial in $|\varphi|$ and
$\epsilon^{-1}$. We show that
for $v = \hat{v}^{\frac{1}{2\lambda^{-1}}}$, 
it holds
\begin{align*}
\countnsat(\varphi) \overset{(c)}{\leq} v \overset{(d)}{\leq} 
(1 + \lambda) \cdot \countnsat(\varphi).
\end{align*}
Inequality $(c)$ follows from Inequality $(a)$.
Inequality $(b)$ implies 
$v \leq (\frac{3}{2})^{\frac{1}{2\lambda^{-1}}}
\cdot \countnsat(\varphi)$.
Then, Inequality $(d)$ follows from 
$(\frac{3}{2})^{\frac{1}{2\lambda^{-1}}} < 1+\lambda$, which holds because:
\begin{align*}
(\frac{3}{2})^{\frac{1}{2\lambda^{-1}}} < 1+\lambda
\Leftrightarrow
(\frac{3}{2})^{\frac{\lambda}{2}} < 1+\lambda
\Leftrightarrow
\frac{\lambda}{2}\cdot \ln(\frac{3}{2}) < \ln(1+\lambda)
\end{align*}
To obtain the last equivalence, 
we take the natural logarithm on both sides of the 
inequality. 
The last inequality holds because of
$0 < \ln(\frac{3}{2}) < 1$
and 
$\frac{\lambda}{2} < 
\frac{\lambda}{1+\lambda} \leq
\ln(1+\lambda)$,
where 
$\frac{\lambda}{1+\lambda} \leq
\ln(1+\lambda)$ is 
the standard inequality for the natural logarithm~\cite{mitrinovic2013classical}. 

\subsection{Intractability in the General Case}
\label{sec:intrac_arbit_non_hierarchical}
The generalization of the intractability result for the basic 
non-hierarchical CQ $Q_{nh}$ in Eq.~\eqref{eq:basic_non_hierarchical} to 
arbitrary non-hierarchical Boolean CQs without self-joins
closely follows 
prior work~\cite{dalvi2007efficient, LivshitsBKS:LMCS:2021}:
We give a polynomial-time reduction from  
$\compbanz_{Q_{nh}}$
to $\compbanz_Q$
for any non-hierarchical Boolean CQ $Q$ without self-joins. 
From this, it follows: A polynomial-time algorithm for 
$\compbanz_Q$ implies 
a polynomial-time algorithm for 
$\compbanz_{Q_{nh}}$, 
which, as explained in Sec.~\ref{sec:intrac_basic_non_hierarchical},  implies that there is an $\fptas$ for $\countbis$. 

We explain the reduction. Consider a non-hierarchical Boolean CQ $Q$ without self-joins
The query $Q$ must contain 
three atoms 
$R(X, \bm X)$, $S(X,Y, \bm Z)$, and 
$T(Y, \bm Y)$ such that 
$X \notin \bm Y$ and $Y \notin \bm X$.
Given an input database $D_{nh}$ for $\compbanz_{Q_{nh}}$ 
containing three relations 
$R_{nh}$, 
$S_{nh}$, and $T_{nh}$,
we construct as follows an input database $D$ for $\compbanz_Q$. 
The values in the  $X$-column of $R_{nh}$ ($Y$-column of $T_{nh}$) are copied 
to the $X$-column of $R$ ($Y$-column of $T$).
The values in the  $X$-column of $S_{nh}$ are copied 
to each $X$-column of all relations besides $R$ in $D$.
Similarly, the values in the  $Y$-column of $S_{nh}$ are copied 
to each $Y$-column of all relations besides $T$ in $D$.
Partial facts, i.e., those for which only some columns are assigned to values, 
are completed using a fixed dummy value for all
columns with missing values.
The facts in $R$ and $T$ are set to be endogenous while 
all other facts in $D$ are set to be exogenous. 
Observe that we have a one-to-one mapping between 
the endogenous facts in $D_{nh}$ and those in $D$.
The Banzhaf value of each endogenous fact in $D_{nh}$
is the same as the Banzhaf value 
of the corresponding fact in $D$. Hence, a polynomial-time algorithm
for $\compbanz_Q$ implies a polynomial-time 
algorithm for $\compbanz_{Q_{nh}}$. 
\section{Missing Details in Section~\ref{sec:related}}
\label{app:shapley}
In this work we investigate the Banzhaf value 
as a measure to quantify the contribution of database facts to 
query results. 
Prior work considered the Shapley value to score facts in query answering~\cite{LivshitsBKS:LMCS:2021}. 
In this section we show that Banzhaf-based and Shapley-based ranking 
of facts can differ already for very simple queries and databases.

\paragraph{Shapley Value}
We recall the definition of the Shapley value of a variable in a Boolean function:

\begin{definition}[Shapley Value of Boolean Variable]
\label{def:shapley}
Given a Boolean function $\varphi$ over $\bm X$, 
the {\em Shapley value} of a variable $x \in \bm X$ in $\varphi$ is:

\begin{align}
\label{eq:shapley_value}
\shap(\varphi,x) \defeq \sum_{\bm Y \subseteq \bm X \setminus \{x\}} 
c_{\bm Y} \cdot 
 \big( \varphi[\bm Y \cup \{x\}] - \varphi[\bm Y] \big),
\end{align}
where $c_{\bm Y} = \frac{|\bm Y|! (|\bm X| - |\bm Y| - 1)!}{|\bm X|!}$.
\end{definition}
Observe that the Shapley value formula in Eq.~\eqref{eq:shapley_value} differs from the 
Banzhaf value formula in Eq.~\eqref{eq:banzhaf_value} in that 
each term $\varphi[\bm Y \cup \{x\}] - \varphi[\bm Y]$ in the former formula is
multiplied by the coefficient $c_{\bm Y}$.

Analogous to the case of Banzhaf values, the Shapley value of a database fact is defined as the Shapley value of that fact in the query lineage. 
Given a Boolean query $Q$, a database $D = (D_n,D_x)$, and an endogenous fact $f \in D_n$, let $v(f)$ be the variable associated to $f$. We define:
\begin{align*}
\shap(Q,D,f) \defeq \shap(\varphi_{Q,D},v(f)),
\end{align*}
where $\varphi_{Q,D}$ is the lineage of $Q$ over $D$.

\paragraph{Critical Sets}
Both the Banzhaf and the Shapley value of a database fact $f$ can be expressed in terms of the number of fact sets for which the inclusion of $f$  turns the query result from $0$ to $1$. 
Consider a Boolean query $Q$, a database $D = (D_n,D_x)$, and an endogenous fact $f \in D_n$.
We call a set $D' \subseteq (D_n \setminus \{f\})$ {\em critical} for $f$ if  
$Q(D' \cup D_x) = 0$ and $Q(D' \cup D_x \cup \{f\}) = 1$. We denote 
by $\#_kC(Q,D,f)$
the number of critical sets of $f$ of size $k$. 
If $Q$ and $D$ are clear from the context, we use 
the abbreviation $\#_kC(f)$.
Observe that the Banzhaf value of  $f$ is
exactly the number of critical sets of $f$.
Hence, we can compute it 
by summing up the
numbers of critical sets of all possible sizes:
\begin{align}
\label{eq:banzhaf_critical_set}
\banz(Q,D,f) = &\ \sum_{k = 0}^{|D_n|-1} \#_kC(Q,D,f) 
\end{align}
We obtain from the formula above the formula for the Shapley value of $f$
by scaling each count $\#_kC(Q,D,f)$ by the coefficient 
$c_k$, which is equal to $c_{\bm Y}$ for any $\bm Y$ with
$|\bm Y| = k$:
\begin{align}
\label{eq:shapley_critical_set}
\shap(Q,D,f) = &\ \sum_{k = 0}^{|D_n|-1} c_k\cdot \#C_k(Q,D,f),
\end{align}
where $c_k = \frac{k! (|D_n| - k - 1)!}{|D_n|!}$.

\paragraph{Difference between Banzhaf-based and Shapley-based Ranking}
We give a query and a database such that the ranking of facts in the database  based on Banzhaf is different
from their ranking based on Shapley. 
Consider the query 
$Q = \exists X\exists Y\exists Z R(X) \wedge S(X,Y), T(X,Z)$ and 
the database consisting of the following three relations $R$, $S$, and $T$. All 18 facts in the database are assumed to be endogenous. 

\begin{center}
\begin{tikzpicture}
\node (R) at (0,1) {$R$};
\node (R_table) at (0,0) {
    \begin{tabular}{c}
         $X$ \\
        \hline
        $a_1$ \\
        $a_2$ \\ 
        \hline
    \end{tabular}
 };

\begin{scope}[xshift= 2cm, yshift= -0.58cm]
\node (S) at (0,1.6) {$S$};
\node (S_table) at (0,0) {
  \begin{tabular}{cc}
         $X$ &  $Y$ \\
        \hline
         $a_1$ & $b_1$ \\
         $a_1$ & $b_2$\\
         $a_1$ & $b_3$\\
         $a_2$ & $b_1$\\
         $a_2$ & $b_2$\\
        \hline
    \end{tabular}
 };
\end{scope}

\begin{scope}[xshift= 4.5cm, yshift= -1.8cm]
  \node (T) at (0,2.8) {$T$};
\node (T_table) at (0,0) {
   \begin{tabular}{cc}
        $X$ & $Z$ \\
        \hline
        $a_1$ & $b_1$ \\
         $a_1$ & $b_2$\\
        $a_1$ & $b_3$\\
        $a_2$ & $b_1$\\
        $a_2$ & $b_2$\\
        $a_2$ & $b_3$\\
        $a_2$ & $b_4$\\
        $a_2$ & $b_5$\\
        $a_2$ & $b_6$\\
        $a_2$ & $b_{7}$\\
        $a_2$ & $b_{8}$\\
        \hline
    \end{tabular}
 };
\end{scope}
 \end{tikzpicture}    
\end{center}

A set $D' \subseteq D \setminus \{R(a_1)\}$ is critical for 
the fact $R(a_1)$ if and only if the following conditions holds:
\begin{enumerate}
    \item $D' \cap \{S(a_1, b_i)| i \in [3]\} \neq \emptyset$
    \item $D' \cap \{T(a_1, b_i)| i \in [3]\} \neq \emptyset$
    \item $R(a_2) \not \in D'$ or $D' \cap \{S(a_2,b_i) | i \in [2]\} = \emptyset$ or $D' \cap \{T(a_2,b_i) | i \in [8]\} = \emptyset$
\end{enumerate}

A set $D' \subseteq D \setminus \{R(a_2)\}$ is critical for
the fact $R(a_2)$ if and only if the following conditions holds:
\begin{enumerate}
    \item $D' \cap \{S(a_2, b_i) \mid i \in [2]\} \neq \emptyset$
    \item $D' \cap \{T(a_2, b_i) \mid i \in [8]\} \neq \emptyset$
    \item $R(a_1) \not \in D'$ or $D' \cap \{S(a_1,b_i) | i \in [3]\} = \emptyset$ or  $D' \cap \{T(a_1,b_i) | i \in [3]\} = \emptyset$
\end{enumerate}

\begin{center}
\begin{tabular}{|c|c c| c c|}
    \hline
     $k$ & $\#_kC(R(a_1))$ & $\#_kC(R(a_2))$ & $c_k \cdot \#_kC(R(a_1))$ & $c_k \cdot \#_kC(R(a_2))$\\
     \hline
     0 & 0 & 0 & 0 & 0 \\
     1 & 0 & 0 & 0 & 0 \\
     2 & 9 & 16 & 0.0037 & 0.0065 \\
     3 & 117 & 176 & 0.0096 & 0.0144 \\
     4 & 708 & 924 & 0.0165 & 0.0216 \\
     5 & 2,502 & 2,936 & 0.0225 & 0.0264 \\
     6 & 5,968 & 6,430 & 0.0268 & 0.0289 \\
     7 & 10,262 & 10,326 & 0.0293 & 0.0295 \\
     8 & 13,129 & 12,526 & 0.03 & 0.0286 \\
     9 & 12,695 & 11,638 & 0.029 & 0.0266 \\
     10 & 9,329 & 8,317 & 0.0266 & 0.0238 \\
     11 & 5,191 & 4,553 & 0.0233 & 0.0204 \\
     12 & 2,156 & 1,883 & 0.0194 & 0.0169 \\
     13 & 649 & 572 & 0.0151 & 0.0134 \\     
     14 & 134 & 121 & 0.0109 & 0.0099 \\
     15 & 17 & 16 & 0.0069 & 0.0065 \\
     16 & 1 & 1 & 0.0033 & 0.0033 \\
     17 & 0 & 0 & 0 & 0 \\
     \hline
     \textbf{Total} & \textbf{62,867} & \textbf{60,435} & \textbf{0.2723} & \textbf{0.2766} \\
     \hline
\end{tabular}
\end{center}

The above table gives for each 
$k \in \{0, \ldots , 17\}$, 
the number 
$\#C_k(R(a_1))$ of critical sets  of size $k$ for $R(a_1)$ (second column),
the number $\#_kC(R(a_2))$ of critical sets of size $k$ for $R(a_2)$  (third column) and the values $c_k\cdot \#_kC(Q,D,a_1)$ and 
$c_k\cdot \#C(Q,D,a_2)$ (fourth and fifth column), where
$c_k = \frac{k! (17 - k)!}{18!}$
(the script computing these numbers is available in the repository of this work \cite{Article_git}). 
The numbers in the fourth and fifth column are rounded to four decimal digits.
By Eq.~\eqref{eq:banzhaf_critical_set}, the sum of the values in the second (third) column is the Banzhaf value of $R(a_1)$ ($R(a_2)$). By Eq.~\eqref{eq:shapley_critical_set}, the sum of the values in the fourth (fifth) column is the Shapley  value of $R(a_1)$ ($R(a_2)$).
 We observe that $\banz(Q,D,R(a_1)) > \banz(Q,D,R(a_2))$ while 
 $\shap(Q,D,R(a_1)) < \shap(Q,D,R(a_2))$.

\section{Missing Details in Section~\ref{sec:experiments}}
\label{app:topk}
We show the execution times of the variant of $\textsc{IchiBan}$ that decides the top-$k$ results with certainty. Table \ref{tab:rankingAlgorithms} presents a breakdown of the execution times and success rates for the different datasets.

\paragraph{Academic Dataset} On the Academic dataset, $\textsc{IchiBan}$ consistently outperforms both $\textsc{ExaBan}$ and $\textsc{AdaBan}0.1$. Specifically, for the tested values of $k$, $\textsc{IchiBan}$ demonstrates a mean execution time that is 13-25 times faster than $\textsc{ExaBan}$ and 5-9 times faster than $\textsc{AdaBan}0.1$.

\paragraph{TPC-H Dataset} 
On the TPC-H dataset,
$\textsc{IchiBan}$ outperforms 
$\textsc{ExaBan}$ and $\textsc{AdaBan}0.1$ only in case of $k = 1$.
For other values of $k$, 
$\textsc{IchiBan}$ shows significantly poorer performance. The mean execution time for these values of $k$ is approximately 2 times slower than $\textsc{ExaBan}$ and about 50 times slower than $\textsc{AdaBan}0.1$.

\paragraph{IMDB Dataset} On the IMDB dataset,
the performance of $\textsc{IchiBan}$ varies with the value of $k$. For $k =1$ and $k =3$, $\textsc{IchiBan}$ shows  faster running times and higher success rates than $\textsc{ExaBan}$ and $\textsc{AdaBan}0.1$.
For larger values of $k$, $\textsc{IchiBan}$'s performance gradually becomes worse. While it still outperforms $\textsc{ExaBan}$, it is  about 2-3 times slower than $\textsc{AdaBan}0.1$ and has lower success rate.

\paragraph{Performance Analysis} 
We attribute the variability in the performance of $\textsc{IchiBan}$ to the different properties of the datasets.
The good performance of $\textsc{IchiBan}$ for $k=1$ can be 
explained by the observation that in almost all of the 
lineages, there is a clear top-1 variable that appears in all 
or almost all of the clauses. Thus, the problem of 
identifying the top-1 variable is easy. 
For other values of $k$, we can still see a speedup compared 
to $\textsc{AdaBan}0.1$ on many of the lineages, which means
that in many cases we do not need a high precision for the
bounds in order to achieve a separation of the Banzhaf values.
In cases where $\textsc{IchiBan}$ performs poorly,
the reason for the bad performance is often 
the high number of ties, especially for  variables with small Banzhaf values. This means that the variant of $\textsc{IchiBan}$ that tries to decide top-k with certainty needs to expand the 
d-tree completely while repeatedly calculating bounds for the Banzhaf values after expansion steps. This results in 
higher execution times than for $\textsc{ExaBan}$, which calculates Banzhaf values only after the d-tree is expanded completely. 

\begin{table}
   \centering
   \resizebox{0.75\textwidth}{!}{%
   \begin{tabular}{| c | c | c | c c c c c c|}
       \hline
        \multirow{2}{*}{\textbf{Dataset}} & \multirow{2}{*}{\textbf{Algorithm}} & {\makecell{\textbf{Success} \\ \textbf{rate}}} &
        \multicolumn{6}{c |}{\textbf{Execution times [sec]}} \\
        & & & & & & & & \\
        & & &\textbf{Mean} &  \textbf{p50} & \textbf{p90} & \textbf{p95} & \textbf{p99} & \textbf{Max} \\
        \hline
        {\multirow{2}{*}{{\texttt{Academic}}}} 
        
       \revision{} & \revision{Top1} & \revision{100\%} & \revision{0.083} & \revision{0.001} &
       \revision{0.002} & \revision{0.004} & \revision{6.635} & 
       \revision{7.738}\\
        
       \revision{} & \revision{Top3} & \revision{99.99\%} & \revision{0.097} & \revision{0.001} &
       \revision{0.006} & \revision{0.035} & \revision{5.841} & 
       \revision{89.54}\\

       \revision{} & \revision{Top5} & \revision{99.99\%} & \revision{0.152} & \revision{0.001} &
       \revision{0.015} & \revision{0.148} & \revision{8.740} & 
       \revision{149.1}\\
        
       \revision{} & \revision{Top10} & \revision{99.99\%} & \revision{0.109} & \revision{0.001} &
       \revision{0.008} & \revision{0.031} & \revision{6.377} & 
       \revision{120.8}\\
        
        \hline

        {\multirow{2}{*}{{\texttt{IMDB}}}} 
       \revision{} & \revision{Top1} & \revision{100\%} & \revision{0.003} & \revision{0.001} &
       \revision{0.002} & \revision{0.003} & \revision{0.018} & 
       \revision{72.188}\\  
         
       \revision{} & \revision{Top3} & \revision{99.94\%} & \revision{0.405} & \revision{0.001} &
       \revision{0.007} & \revision{0.018} & \revision{0.276} & 
       \revision{3568}\\

       \revision{} & \revision{Top5} &   
       \revision{99.82\%} & \revision{1.033} & \revision{0.001} & \revision{0.016} &
       \revision{0.062} & \revision{4.652} & \revision{3330}\\
        
       \revision{} & \revision{Top10} & \revision{99.79\%} & \revision{1.863} & \revision{0.001} &
       \revision{0.024} & \revision{0.114} & \revision{5.752} & 
       \revision{3430}\\

       \hline
        
       {\multirow{2}{*}{{\texttt{TPC-H}}}} 
       \revision{} & \revision{Top1} & \revision{98.18\%} & \revision{0.329} & \revision{0.001} &
       \revision{1.687} & \revision{3.191} & \revision{3.626} & 
       \revision{3.746}\\  
         
       \revision{} & \revision{Top3} & \revision{91.52\%} & \revision{9.903} & \revision{0.001} &
       \revision{0.007} & \revision{86.03} & \revision{207.5} & 
       \revision{273.8}\\

       \revision{} & \revision{Top5} & \revision{91.52\%} & \revision{9.235} & \revision{0.001} &
       \revision{0.010} & \revision{79.88} & \revision{192.9} & 
       \revision{226.7}\\
        
       \revision{} & \revision{Top10} & \revision{91.52\%} & \revision{9.175} & \revision{0.001} &
        \revision{0.001} & \revision{83.69} & \revision{186.6} & 
        \revision{230.5}\\
        
        \hline
    \end{tabular}%
    }
    \caption{Top-k computation for Academic, IMDB and TPC-H datasets; execution times with respect to each lineage expression for which the algorithm succeeded within a timeout of 1 hour.}
   \label{tab:rankingAlgorithms}

\end{table}
\end{document}